\documentclass[12pt]{article}

\usepackage[utf8]{inputenc}
\usepackage[a4paper, margin=2.5cm]{geometry}

\usepackage[colorlinks,citecolor=blue,urlcolor=blue,breaklinks]{hyperref}
\usepackage[nottoc]{tocbibind}
\usepackage[toc]{appendix}
\usepackage{amsmath, amssymb, amsthm, amsfonts, mathtools, bbm, bm, graphics, caption, natbib, float, latexsym, titlesec, enumitem, comment}
\usepackage{algorithm}
\usepackage{algpseudocode}
\usepackage{multirow}
\usepackage[dvipsnames]{xcolor}
\usepackage{lscape}
\usepackage{aligned-overset}
\usepackage{subcaption}
\usepackage{overpic}
\usepackage{tikz}
\usetikzlibrary{positioning}

\makeatletter
\newcommand{\ostar}{\mathbin{\mathpalette\make@circled\star}}
\newcommand{\make@circled}[2]{%
\ooalign{$\m@th#1\smallbigcirc{#1}$\cr\hidewidth$\m@th#1#2$\hidewidth\cr}%
}
\newcommand{\smallbigcirc}[1]{%
  \vcenter{\hbox{\scalebox{0.77778}{$\m@th#1\bigcirc$}}}%
}
\makeatother

\DeclareMathOperator*{\argmax}{argmax}

\DeclareMathOperator*{\sargmax}{sargmax}

\newcommand{\indep}{\perp\!\!\!\perp}

\newtheorem{thm}{Theorem}
\newtheorem{prop}[thm]{Proposition}
\newtheorem{lemma}[thm]{Lemma}

\newtheorem{defn}{Definition}

\newtheorem{example}{Example}
\newtheorem{assumption}{Assumption}

\allowdisplaybreaks

\def\hat{\widehat}
\def\tilde{\widetilde}

\title{Deep learning with missing data}
\author{Tianyi Ma$^1$, Tengyao Wang$^2$ and Richard J. Samworth$^{1}$\\ \\
$^1$Statistical Laboratory, University of Cambridge\\
$^2$Department of Statistics, London School of Economics
}
\date{}

\begin{document}	
\maketitle

\begin{abstract}
In the context of multivariate nonparametric regression with missing covariates, we propose \emph{Pattern Embedded Neural Networks} (PENNs), which can be applied in conjunction with any existing imputation technique.  In addition to a neural network trained on the imputed data, PENNs pass the vectors of observation indicators through a second neural network to provide a compact representation.  The outputs are then combined in a third neural network to produce final predictions.  Our main theoretical result exploits an assumption that the observation patterns can be partitioned into cells on which the Bayes regression function behaves similarly, and belongs to a compositional H\"older class.  It provides a finite-sample excess risk bound that holds for an arbitrary missingness mechanism, and in combination with a complementary minimax lower bound, demonstrates that our PENN estimator attains in typical cases the minimax rate of convergence as if the cells of the partition were known in advance, up to a poly-logarithmic factor in the sample size.  Numerical experiments on simulated, semi-synthetic and real data confirm that the PENN estimator consistently improves, often dramatically, on standard neural networks without pattern embedding.  Code to reproduce our experiments, as well as a tutorial on how to apply our method, is publicly available.
\end{abstract}
\noindent\textbf{Keywords:} Deep learning, Missing data, Nonparametric regression

\section{Introduction}

Over the last decade or so, deep neural networks have achieved stunning empirical successes in learning tasks across diverse application areas, including image classification \citep{krizhevsky2012imagenet,he2016deep}, protein folding \citep{jumper2021highly}, natural language processing \citep{vaswani2017attention,devlin2019bert} and many others.  While their complexity initially led to the temptation to regard them as a black box, recent theory does provide insights into the origins of their impressive practical performance \citep{jacot2018neural,schmidt-hieber2020nonparametric,kohler2021rate,mei2022generalization,suh2024survey}. 

The works described in the previous paragraph all concern settings where the (large) datasets involved are fully observed.  Nevertheless, as argued by \citet{zhu2022high}, missing data play an ever more significant role in high-dimensional learning problems, and new methods have now been introduced to tackle several different contemporary statistical challenges involving missing data, including sparse linear regression \citep{loh2012high,chandrasekher2020imputation}, principal component analysis \citep{elsener2019sparse,zhu2022high}, classification \citep{cai2019high,sell2023nonparametric} and changepoint estimation \citep{follain2022high}.  Most of these papers study the simplest, idealised case where the data and missingness indicators are independent, a setting known as Missing Completely At Random (MCAR).  Indeed, it is now understood that for more general (dependent) missingness mechanisms, even consistent mean estimation may be impossible without further assumptions \citep{ma2024estimation}.

The goal of this paper is to study supervised deep learning problems with missing values among the covariates, where we do not rely on an MCAR assumption, or indeed any other restriction on the missingness mechanism.  Motivated by many applications in which deep learning is applied, we focus on the problem of prediction (i.e.~making a forecast of the response at a new covariate vector) as opposed to estimation (i.e.~learning the underlying parameters of the model).  Our primary methodological contribution is to introduce the idea of (observation) pattern embedding into the neural network framework.  In other words, in addition to training a neural network on our original covariates (with missing values imputed via any existing technique), we pass the vectors of observation indicators, which we call \emph{revelation vectors}, through another neural network to obtain a compact representation that summarises the information they contain.  We can then train a third neural network that combines these two earlier ones to produce final predictions.

The benefits of our approach are illustrated in Figure~\ref{fig:quadratic}, where our training data of size $n = 1{,}000$ and $d=1$ are generated from the model described in Example~\ref{example:f-star} in Section~\ref{sec:setup} below.  Our original covariates $(\bm X_i)_{i=1}^n$ are observed with missingness, and we apply zero (mean) imputation based on the corresponding revelation vectors $(\bm \Omega_i)_{i=1}^n$ to obtain imputed covariates $(\bm Z_i)_{i=1}^n$.  Thus, together with the responses $(Y_i)_{i=1}^n$, which satisfy $Y_i|\bm X_i \sim  N(3\bm X_i^2,0.1^2)$ independently for $i=1,\ldots,n$, we have training data $(\bm Z_i, \bm \Omega_i, Y_i)_{i=1}^n$.  Following the imputation, the standard approach would be to train a model on $(\bm Z_i, Y_i)_{i=1}^n$ (see Figure~\ref{fig:quadratic}(a)), yielding fitted values on an independent test set of size 500 displayed in panel~(c).  By contrast, our approach trains a neural network on the augmented data $(\bm Z_i, \bm \Omega_i, Y_i)_{i=1}^n$ in panel~(b), leading to much more accurate predictions on the test set, as illustrated in panel~(d).  The main point to observe here is that the failure to include the revelation vectors $(\bm\Omega_i)_{i=1}^n$ in panel~(c) corrupts the output in a neighbourhood of the origin, which is a point of discontinuity of the function $\bm z \mapsto \mathbb{E}(Y_1 \, | \, \bm Z_1 = \bm z)$, whereas this is corrected in panel~(d) by the inclusion of the additional covariate.
% \red{The main point to observe here is that the function $\bm z \mapsto \mathbb{E}(Y_1 \, | \, \bm Z_1 = \bm z)$ is discontinuos at the origin, and not including the revelation vectors $(\bm\Omega_i)_{i=1}^n$ in panel~(c) corrupts the output in a neighbourhood of the point of discontinuity, whereas this is corrected in panel~(d) by the inclusion of the additional covariate.}  

\begin{figure}[htbp]
\centering
\begin{subfigure}[t]{0.49\textwidth}
    \centering
    \begin{overpic}[width=\textwidth]{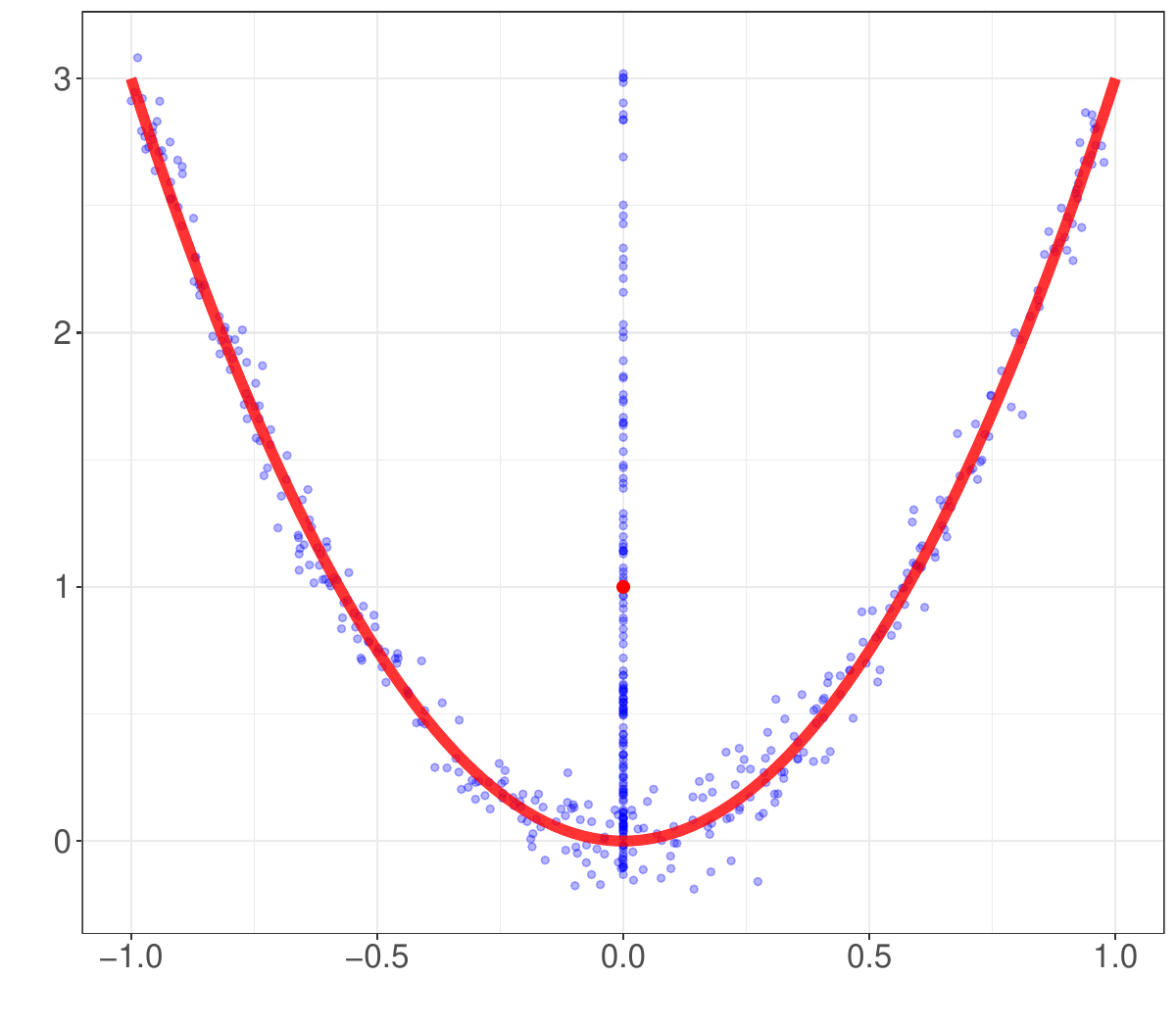}
      \put(52,1){\footnotesize$\bm z$}
    \end{overpic}
    \caption{Test data (blue) and Bayes regression function $\bm z \mapsto \mathbb{E}(Y_1 \, | \, \bm Z_1 = \bm z)$ (red).}
\end{subfigure}
\begin{subfigure}[t]{0.49\textwidth}
    \begin{overpic}[width=\textwidth]{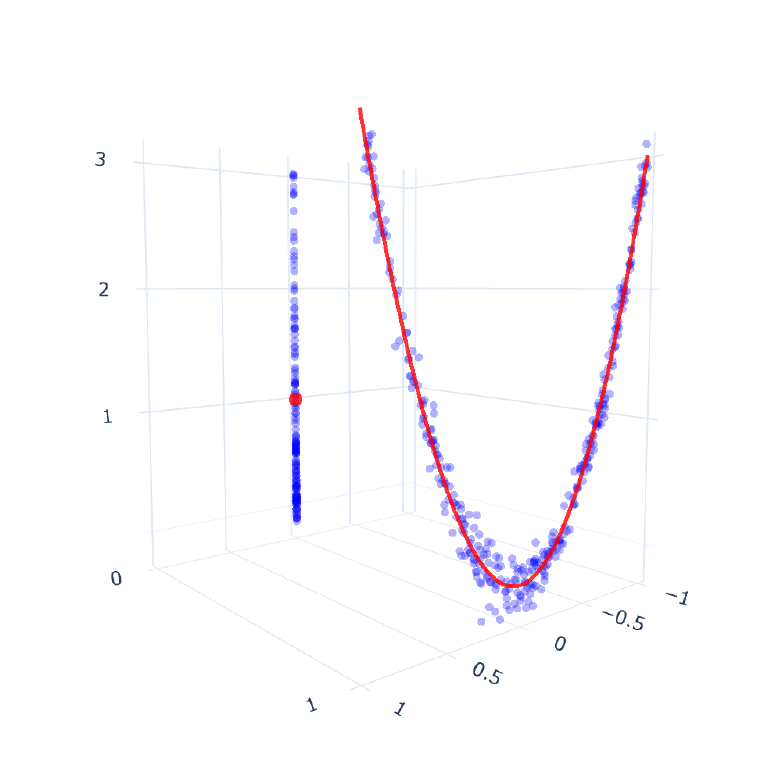}
      \put(70,10){\footnotesize$\bm z$}
      \put(22,14){\footnotesize$\bm \omega$}
    \end{overpic}
    \caption{Test data (blue) and Bayes regression function $(\bm z, \bm \omega) \mapsto \mathbb{E}(Y_1 \, | \, \bm Z_1 = \bm z, \bm \Omega_1 = \bm \omega)$ (red).}
\end{subfigure}\\ \hspace{0.75cm}

\begin{subfigure}[t]{0.49\textwidth}
    \begin{overpic}[width=\textwidth]{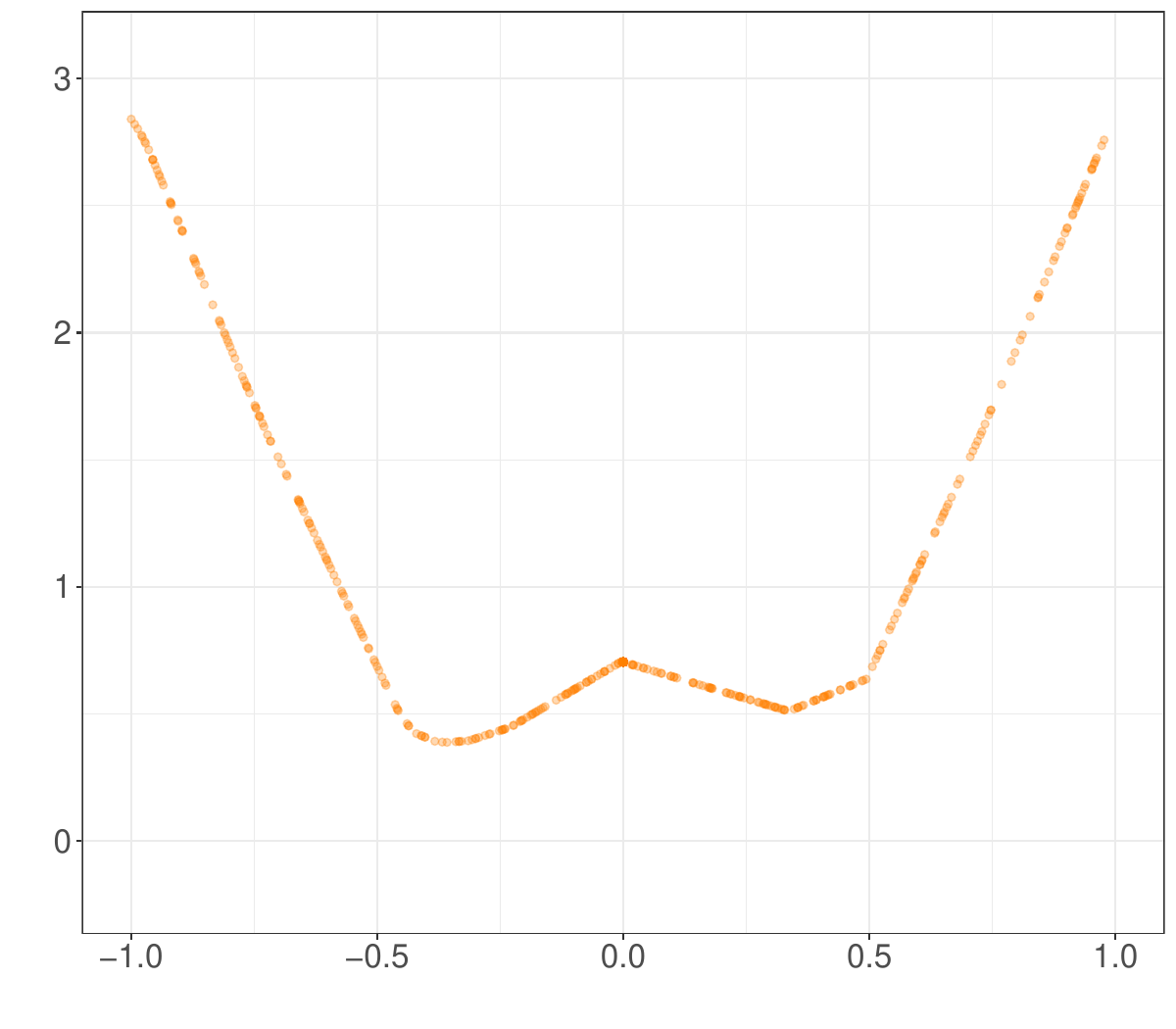}
      \put(52,1){\footnotesize$\bm z$}
    \end{overpic}
    \caption{Output of the neural network trained without the revelation vectors.}
\end{subfigure}
\begin{subfigure}[t]{0.49\textwidth}
    \begin{overpic}[width=\textwidth]{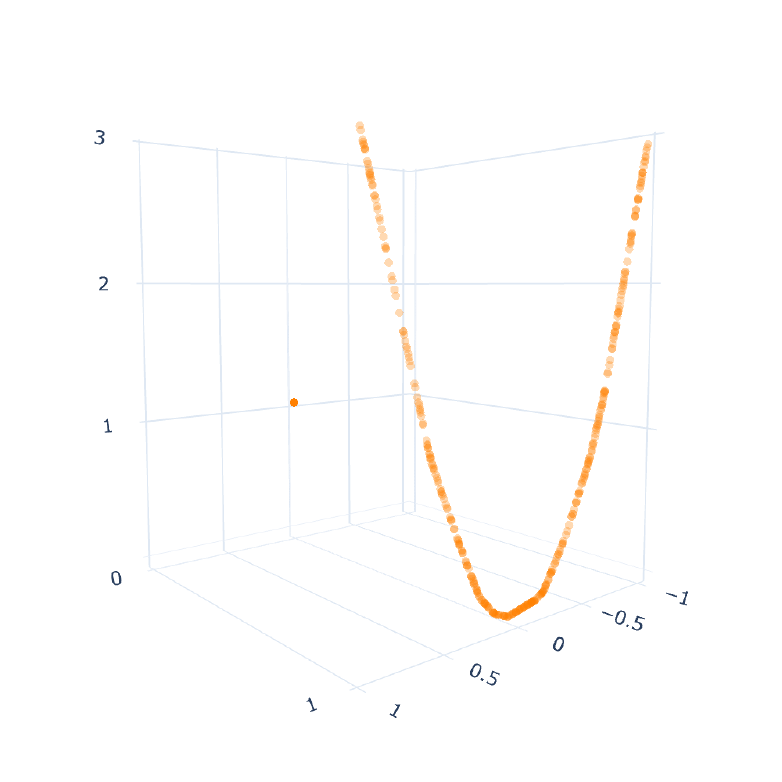}
      \put(70,10){\footnotesize$\bm z$}
      \put(22,14){\footnotesize$\bm \omega$}
    \end{overpic}
    \caption{Output of the neural network trained with the revelation vectors.}
\end{subfigure}
\caption{An illustration of Example~\ref{example:f-star}, and the outputs of neural networks trained without and with the revelation vectors.} 
\label{fig:quadratic}
\end{figure}

On the other hand, as we illustrate in Figure~\ref{fig:quadratic-d} in Section~\ref{sec:PENN}, when $d$ is larger, naive inclusion of the revelation vectors as additional covariates can harm performance, even compared with omitting $(\bm \Omega_i)_{i=1}^n$ entirely.  Moreover, as we also discuss in Section~\ref{sec:PENN}, the predictive value of the revelation vectors may not be well-aligned with their Euclidean geometry.  Thus, a key challenge is to construct a suitable compression of the information contained in the revelation vectors that retains their predictive content to the greatest extent possible. 
% \red{To this end, we map the revelation vectors into a lower dimensional space using a learned neural network.}

% We train one neural network on $(\bm Z_i)_{i=1}^n$ and another neural network on $(\bm Z_i,\bm\Omega_i)_{i=1}^n$, before evaluating their performance on an independent test set of size 500. Figure~\ref{fig:quadratic}(a) plots the test data in blue and $\bm z \mapsto \mathbb{E}(Y_0 \, | \, \bm Z_0 = \bm z)$ in red (note that $\mathbb{E}(Y_0 \, | \, \bm Z_0 = 0) = 1$); Figure~\ref{fig:quadratic}(b) plots the test data in blue and the Bayes regression function $(\bm z, \bm \omega) \mapsto \mathbb{E}(Y_0 \, | \, \bm Z_0 = \bm z, \bm \Omega_0 = \bm \omega) \eqqcolon f^{\star}(\bm z, \bm \omega)$ in red.  In the bottom panels,  Figure~\ref{fig:quadratic}(c) plots the output (evaluated on the test set) of the neural network trained on $(\bm Z_i)_{i=1}^n$, while Figure~\ref{fig:quadratic}(d) plots the corresponding output for the neural network trained on $(\bm Z_i, \bm\Omega_i)_{i=1}^n$.  

After a more formal description of our problem set-up and some background on ReLU neural networks, we introduce our \emph{Pattern Embedded Neural Networks} (PENNs) in Section~\ref{sec:setup-method}.
%\red{Motivated by the observations mentioned in the previous two paragraphs, PENNs first process the imputed covariates and revelation vectors through two separate neural networks, with the latter mapping the revelation vectors to a lower dimensional embedding. The resulting representations are then concatenated and passed through a third neural network to produce the final prediction.}
These can be fitted using standard empirical risk minimisation algorithms such as Adam or AdamW \citep{kingma2015adam, loshchilov2018decoupled}, as implemented in \texttt{PyTorch} in Python, yielding an estimator $\hat{f}$ of the Bayes regression function $f^\star$, given by $f^\star(\bm z,\bm \omega) \coloneqq \mathbb{E}(Y_1 \, | \, \bm Z_1 = \bm z, \bm \Omega_1 = \bm \omega)$.  

In Section~\ref{sec:theoretical-results} we turn our attention to the theoretical properties of our procedure.  We begin with a general oracle inequality (Proposition~\ref{prop:oracle-inequality}), revealing that under a sub-Gaussian condition on the response, the excess risk of our estimator is controlled by the sum of optimisation error, approximation error and estimation error terms, with the latter two reflecting a bias--variance trade-off in the complexity of the neural network class.   
% The optimisation error term vanishes for an exact empirical risk minimiser.  The approximation error term allows for the possibility that the Bayes regression function does not itself belong to the neural network class, and decreases as the complexity of this class (measured in terms of its depth, widths and number of non-zero parameters) increases.  On the other hand, the estimation error term, which decays at the parametric rate in the sample size up to a poly-logarithmic factor, increases with the complexity of the neural network class.  Thus, these two terms may be regarded as reflecting a bias--variance trade-off.  
Our main result (Theorem~\ref{thm:PENN-ub} in Section~\ref{sec:minimax-rate}) specialises the general setting studied previously to give a more explicit excess risk bound.  In particular, since there $2^d$ possible observation patterns, we introduce a new condition on the Bayes regression function $f^\star$ to encapsulate the notion that the set of all possible observation patterns may be partitioned into sets on which $f^\star$ behaves similarly as a function of the covariates.  We further ask that on each cell of this partition, $f^\star$ belongs to a compositional H\"older smoothness class.  This allows us to show that the sum of the approximation error and estimation error terms can be controlled by a weighted average of estimation rates for each cell of the partition, together with an additional term reflecting the complexity of the partition that will typically be negligible, up to a poly-logarithmic factor in the sample size. This theorem is complemented by a minimax lower bound in Theorem~\ref{thm:minimax-lb}, which confirms that the weighted average of the estimation rates over different cells in Theorem~\ref{thm:PENN-ub} is optimal.  All proofs are deferred to the Appendix.

A numerical study of the performance of our PENN estimator is presented in Section~\ref{sec:simulations}.  Our simulated data settings consider both MCAR and Missing Not At Random (MNAR) scenarios, as well as three commonly-used imputation techniques, namely (columnwise) mean imputation, MissForest imputation \citep{stekhoven2012missforest} and Multiple Imputation by Chained Equations (MICE) \citep{van2011mice}.  A consistent pattern emerges whereby the PENN estimator is able to improve, sometimes very significantly, on the corresponding neural network estimator that does not include pattern embedding.  We also compare with other machine learning methods that handle missingness directly, namely XGBoost~\citep{chen2016xgboost} and random forests~\citep{breiman2001random}, and see similar improvements in these settings.  Beyond purely simulated data, we study two semi-synthetic datasets and one real dataset with the same imputation methods.  In the former cases, the original real data, on handwritten digits (where we include Vision Transformers~\citep{dosovitskiy2021an} as an additional competitor) and relative location of Computed Tomography (CT) slices respectively, are observed without missingness, which allows us to introduce either MCAR or MNAR missingness via known mechanisms.  On the other hand, in the latter real dataset on credit scores, missingness is already present in the original data.  By splitting the data into training, validation and test sets, we again see consistent improvements from our pattern embedding approach.  Python code to reproduce all of the experiments in this section, together with an accompanying tutorial on how to apply our method, is available at \url{https://github.com/tianyima2000/DNN_missing_data}.

\subsection{Related literature}

The desire to understand and explain the impressive empirical performance of deep learning represents a major current research theme in statistics and machine learning.  One line of work provides explicit rates of convergence through the lens of approximation theory and empirical process theory, assuming that the empirical risk can be minimised sufficiently well. Minimax optimality (in terms of the sample size) of neural networks has been studied, for example, in the context of nonparametric regression \citep{bauer2019on,imaizumi2019deep,schmidt-hieber2020nonparametric,kohler2021rate,jiao2023deep,bhattacharya2023deep,fan2024factor}, classification \citep{bos2022convergence,zhang2024classification}, the partially linear Cox model \citep{zhong2022deep} and density estimation \citep{bos2024supervised}. In particular, it is now known that neural networks can exploit low-dimensional structure in nonparametric regression under a compositional H\"older assumption \citep{schmidt-hieber2020nonparametric,kohler2021rate}. While most of the bounds in the work mentioned above have pre-factors depending exponentially on the covariate dimension or the number of variables on which the regression function depends, \citet{jiao2023deep} provide upper bounds that depend only polynomially on the dimension, using the approximation scheme of \citet{lu2021deep}, while \citet{imaizumi2019deep} allow discontinuities in the regression function. Other research directions focus on understanding the training dynamics of neural networks, for example through Neural Tangent Kernels \citep{jacot2018neural,du2018gradient,arora2019exact}, Tensor programs \citep{yang2021tensor, littwin2023adaptive}, and mean-field limits \citep{chizat2018global, mei2018mean, mei2022generalization, suzuki2023convergence}. Recently, the in-context learning ability of Transformers has also been exploited to build foundation models for tabular data \citep{hollmann2025accurate}; theoretical results on in-context nonparametric regression have been studied by, e.g., \citet{kim2024transformers,ma2025optimal,ching2026efficient}.
%\blue{These include the analysis of stochastic gradient descent for extremely (or infinitely) wide neural networks, known as the Neural Tangent Kernel regime \citep{jacot2018neural,du2018gradient,arora2019exact}, and single hidden layer neural networks in the proportional asymptotics regime \citep{mei2018mean,mei2022generalization}.}

When fitting neural networks, the embedding of discrete or structured data into a latent space is an attractive approach.  In particular, \emph{categorical embeddings} map discrete objects, such as words or user identifiers, into continuous representations learned jointly with the prediction task \citep{mikolov2013efficient,guo2016entity}.  Relatedly, in unsupervised tasks, \emph{autoencoders} learn compact latent representations via `bottleneck architectures' (with narrow middle layers) that capture essential structure in high-dimensional data \citep{hinton2006reducing,vincent2008extracting}. Viewed in this light, our method learns a low-dimensional embedding of the observation pattern, a categorical variable with up to $2^d$ levels, in which similarities between patterns are inferred from the data rather than specified a priori. Unlike classical autoencoders, in our supervised learning task this embedding is optimised solely for prediction and is not required to reconstruct the missingness pattern itself.

Missing data has been an active research topic in statistics for several decades, but is currently undergoing a renaissance for two main reasons.  First, as already mentioned, missing values are ubiquitous in large datasets, and, unlike for other forms of data corruption, many statistical learning algorithms cannot be applied until the missingness is handled appropriately.  Second, and somewhat related, the complexity of modern data demands new models for missingness and new inferential methods with appropriate guarantees on their performance \citep{berrett2023optimal,ma2024estimation}.  It has long been recognised that the MCAR assumption is unrealistic for many practical applications.  On the other hand, modern data generating mechanisms can often only be adequately described by high- or infinite-dimensional parameter spaces.  In such settings, missingness models such as Missing At Random (MAR) that are predicated on the correctness of low-dimensional parametric models fitted using maximum likelihood may be inappropriate.  

For regression problems with missing covariates, one widely used general strategy is \emph{impute-then-regress}, i.e.~we first impute the missing data and then treat the imputed dataset as complete to train a regression algorithm. Many imputation algorithms have been proposed under MCAR or MAR assumptions, such as MissForest imputation \citep{stekhoven2012missforest}, MICE \citep{van2011mice}, and methods based on deep learning and generative models \citep{li2019misgan,mattei2019miwae,nazabal2020handling,zhang2025diffputer}.  If there exists a universally consistent estimator, then under some conditions that still allow MNAR, the impute-then-regress approach leads to asymptotically vanishing excess risk on new covariate vectors as the sample size diverges to infinity \citep{le2021what,josse2024consistency}. 
However, the Bayes regression function, which in this case is the conditional expectation of the response given the imputed covariate vector, may be discontinuous and hard to learn \citep{le2021what}.  

\citet{efromovich2011nonparametric} studies orthogonal series methods in univariate nonparametric regression with MAR missingness.  Other strategies for regression with missing data include augmenting the covariate space with a distinguished point, reflecting a missing value, in relevant coordinates; this can be applied with regression trees or other (non-orthogonally equivariant) methods such as XGBoost and random forests.  \citet{twala2008good} propose regression trees with the observation patterns included as covariates, \citet{smieja2018processing} replace neurons in the first hidden layer of a neural network by estimated expected values to handle missing data, while \citet{le2021what} train imputation and regression algorithms simultaneously using neural networks.  For image data, MisGAN \citep{li2019misgan} imputes missing pixels, while Vision Transformers (ViTs) can naturally handle occlusions through their patch-based representation and attention mechanism \citep{dosovitskiy2021an, bao2022beit, he2022masked}.  Nevertheless, to the best of our knowledge, our work is the first to provide minimax optimality guarantees for multivariate nonparametric regression with missing data.

\subsection{Notation}

We conclude the introduction with some notation employed throughout the paper.  For $n\in\mathbb{N}$, we define $[n] \coloneqq \{1,\ldots,n\}$, and for $a,b\in\mathbb{R}$, we let $a\wedge b \coloneqq \min\{a,b\}$ and $a\vee b \coloneqq \max\{a,b\}$.  For $\bm{a},\bm{b}\in\mathbb{R}^d$, we let $\bm{a}\odot \bm{b} \in \mathbb{R}^d$ denote the Hadamard product (i.e.~coordinate-wise product) of $\bm{a}$ and $\bm{b}$.  For $q \in [1,\infty)$ and $\bm v = (v_1,\ldots,v_d)^\top \in \mathbb{R}^d$, we write $\|\bm v\|_q \coloneqq \bigl(\sum_{j=1}^d |v_j|^q\bigr)^{1/q}$, as well as $\|\bm v\|_0 \coloneqq \sum_{j=1}^d \mathbbm{1}_{\{v_j \neq 0\}}$ and $\|\bm v\|_{\infty} \coloneqq \max_{j \in [d]} |v_j|$.  The all-ones vector is $\bm{1}_d \coloneqq (1,\ldots,1)^\top \in \mathbb{R}^d$.  If $(\mathcal{X},\mathcal{A},\mu)$ is a measure space and $f:\mathcal{X} \rightarrow \mathbb{R}$ is measurable, then we define $\|f\|_{L_q(\mu)} \coloneqq \bigl(\int_{\mathcal{X}} |f|^q \, d\mu\bigr)^{1/q}$, as well as $\|f\|_\infty \coloneqq \sup_{\bm x \in \mathcal{X}} |f(\bm x)|$.  For function classes $\mathcal{F}$ and $\mathcal{G}$, we define $\mathcal{F}\circ\mathcal{G} \coloneqq \{f \circ g:f \in \mathcal{F}, g \in \mathcal{G}\}$, and for $B\geq 0$, we define the truncation operator $T_B : \mathbb{R} \to [-B,B]$ by $T_B(y) \coloneqq (-B) \vee y \wedge B$.  We say that $\{\mathcal{S}_1,\ldots,\mathcal{S}_K\}$ is a \emph{partition} of a non-empty set $\mathcal{S}$ if $\mathcal{S}_1,\ldots,\mathcal{S}_K$ are pairwise disjoint, non-empty subsets of~$\mathcal{S}$ whose union is $\mathcal{S}$.  The \emph{sub-Gaussian norm} of a real-valued random variable $X$ is $\|X\|_{\psi_2} \coloneqq \inf \{t>0 : \exp(X^2/t^2) \leq 2\}$; its \emph{sub-exponential norm} is $\|X\|_{\psi_1} \coloneqq \inf \{t>0 : \exp(|X|/t) \leq 2\}$.  For positive sequences $(a_n)$, $(b_n)$, we write $a_n \lesssim b_n$ if there exists a universal constant $C > 0$ such that $a_n \leq C b_n$ for all $n \in \mathbb{N}$; if we also have $b_n \lesssim a_n$, then we write $a_n \asymp b_n$. 

\section{Problem set-up and methodology} \label{sec:setup-method}
\subsection{Problem set-up} \label{sec:setup}

Suppose that $(\bm{X}_i,\bm{\Omega}_i,Y_i)_{i=0}^{n}$ are independent and identically distributed random vectors taking values in $\mathbb{R}^d \times \{0,1\}^d \times \mathbb{R}$.  The elements of these triples represent the fully observed covariate vectors, revelation vectors and responses, respectively, and we let $\mathcal{S} \coloneqq \bigl\{\bm{\omega}: \mathbb{P}(\bm{\Omega}_0 = \bm\omega) > 0\bigr\}$ denote the set of possible observation patterns.  In our missing data context, we are in general unable to observe $(\bm{X}_i)_{i=0}^n$; instead, we have access to \emph{partially observed} covariate vectors $(\tilde{\bm{X}}_i)_{i=0}^n$ taking values in the extended space\footnote{Formally, this set is equipped with the natural topology and $\sigma$-algebra described in \citet[Section~2.1]{ma2024estimation}.}  $\mathbb{R}_{\star}^d \coloneqq \bigl(\mathbb{R} \cup \{\star\}\bigr)^d$, where $\star$ denotes a missing value.  Thus $\tilde{\bm{X}}_i$ agrees with $\bm{X}_i$ except in components where~$\bm{\Omega}_i$ is zero, and those components of $\tilde{\bm{X}}_i$ are replaced with $\star$.  
% the coordinates of $\tilde{\bm{X}}_i$ are defined via 
% \begin{align*}
%     \tilde{X}_{i,j} \coloneqq \begin{cases}
%         X_{i,j} \quad&\text{if } \Omega_{i,j}=1\\
%         \star &\text{if } \Omega_{i,j}=0.
%     \end{cases}
% \end{align*}
% \grey{For $\bm x = (x_1,\ldots,x_d)^\top \in \mathbb{R}^d$ and $\bm\omega = (\omega_1,\ldots,\omega_d)^\top \in \{0,1\}^d$, define $\bm x \ostar \bm\omega \in \mathbb{R}_{\star}^d$ by $(\bm x \ostar \bm\omega)_j \coloneqq x_j$ if $\omega_j = 1$ and $(\bm x \ostar \bm\omega)_j \coloneqq \star$ if $\omega_j = 0$; i.e.~we only observe the coordinates of $\bm x$ for which the corresponding coordinates of the revelation vector $\bm\omega$ are one, and the missing entries are encoded by $\star$. and for $i\in\{0,1,\ldots,n\}$, let $\tilde{\bm{X}}_i \coloneqq \bm X_i \ostar \bm\Omega_i$ be the partially observed covariate vector;} 
We emphasise that we do not assume independence between the covariate~$\bm{X}_i$ and revelation vector $\bm{\Omega}_i$, and indeed we allow the entries of~$\tilde{\bm{X}}_i$ to be MNAR. Let $\mathsf{Imp} : \mathbb{R}_{\star}^d \to \mathbb{R}^d$ be a potentially randomised imputation algorithm, where any randomness in the construction of $\mathsf{Imp}$ is independent of $(\bm{X}_i,\bm{\Omega}_i,Y_i)_{i=0}^{n}$, and let $\bm Z_i \coloneqq \mathsf{Imp}(\tilde{\bm X}_i)$ for $i\in\{0,1,\ldots,n\}$. For example, $\mathsf{Imp}$ can be zero (or mean) imputation, or other regression-based imputation algorithms trained on an independent dataset\footnote{This is assumed for theoretical convenience. In practice (and indeed in our simulations in Section~\ref{sec:simulations}), the neural network and imputation algorithms would be trained on the same dataset.}. We observe training data $(\bm{Z}_i,\bm{\Omega}_i,Y_i)_{i=1}^n$ and our goal is to predict~$Y_0$ given a test point $(\bm{Z}_0,\bm{\Omega}_0)$. %We remark that $(\bm{Z}_i,\bm{\Omega}_i,Y_i)_{i=0}^n$ are independent and identically distributed random triples in $\mathbb{R}^d \times \mathcal{S} \times \mathbb{R}$.

Write $\mathcal{G}$ for the set of all possible \emph{prediction functions}, i.e.~the set of Borel measurable functions from $\mathbb{R}^d \times \mathcal{S}$ to $\mathbb{R}$.  We define the \emph{generalisation error} of $\tilde{f} \in \mathcal{G}$ as
\begin{align*}
    R(\tilde{f}) &\coloneqq \mathbb{E}_{(\bm{Z}_0,\bm{\Omega}_0,Y_0)} \bigl\{ \bigl(\tilde{f}(\bm{Z}_0,\bm{\Omega}_0) - Y_0\bigr)^2 \bigr\}\\
    &\hspace{0.09cm}= \int_{\mathbb{R}^d \times \mathcal{S} \times \mathbb{R}} \bigl(\tilde{f}(\bm{z},\bm{\omega}) - y\bigr)^2 \;\mathrm{d} \mu_{\bm{Z}_0,\bm{\Omega}_0,Y_0}(\bm{z},\bm{\omega},y),
\end{align*}
where $\mu_{\bm{Z}_0,\bm{\Omega}_0,Y_0}$ denotes the distribution of $(\bm{Z}_0,\bm{\Omega}_0,Y_0)$. Note that if $\tilde{f}$ is an estimator depending on $(\bm{Z}_i,\bm{\Omega}_i,Y_i)_{i=1}^n$, then $R(\tilde{f})$ is random, since the expectation is taken only over the distribution of $(\bm{Z}_0,\bm{\Omega}_0,Y_0)$.  The \emph{Bayes regression function} $f^{\star} : \mathbb{R}^d \times \mathcal{S} \to \mathbb{R}$ is defined by
\[
    f^{\star}(\bm{z}, \bm{\omega}) \coloneqq 
        \mathbb{E}(Y_0 \,|\, \bm{Z}_0=\bm{z}, \bm{\Omega}_0 = \bm{\omega}).
\]
This function, which satisfies $R(f^{\star}) = \inf_{f \in \mathcal{G}} R(f)$, may therefore depend on the imputation algorithm employed, though $R(f^{\star})$ does not.  For an estimator $\tilde{f}$, we measure its performance by its \emph{excess risk} $\mathbb{E}\bigl\{ R(\tilde{f})\bigr\} - R(f^{\star})$.  Finally, we define the \emph{empirical risk} of~$\tilde{f}$ as 
\[
\hat{R}_n(\tilde{f}) \coloneqq \frac{1}{n} \sum_{i=1}^n \bigl(\tilde{f}(\bm{Z}_i,\bm{\Omega}_i)-Y_i\bigr)^2.
\]

\begin{example} \label{example:f-star}
Let $\bm{X}_0=(X_{0,1},\ldots,X_{0,d})^\top \sim \mathrm{Unif}[-1,1]^d$, and let $Y_0 = 3X_{0,1}^2 + \varepsilon_0$ where $\varepsilon_0 \sim N(0,0.01)$ and $\varepsilon_0 \indep \bm{X}_0$. Suppose that $\bm{\Omega}_0 = (\Omega_{0,1},\ldots,\Omega_{0,d})^\top \indep (\bm{X}_0,\varepsilon_0)$ satisfies $\Omega_{0,1},\ldots,\Omega_{0,d} \overset{\mathrm{iid}}{\sim} \mathrm{Ber}(0.7)$, i.e.~we have MCAR missingness where each coordinate is missing independently with probability $0.3$. Let $\tilde{\bm X}_0$ denote the partially observed covariate vector corresponding to $\bm{X}_0$, and let $\mathsf{Imp}:\mathbb{R}_{\star}^d \to \mathbb{R}^d$ be the zero imputation algorithm (which is the same as mean imputation in this example) that replaces $\star$ by zero, so $\bm Z_0 \coloneqq \mathsf{Imp}(\tilde{\bm X}_0) = \bm X_0 \odot \bm\Omega_0$. For $\bm{z} = (z_1,\ldots,z_d)^\top \in [-1,1]^d$ and $\bm{\omega} = (\omega_1,\ldots,\omega_d)^\top \in \{0,1\}^d$, the Bayes regression function can be written as
\begin{align*}
    f^{\star}(\bm{z}, \bm{\omega}) = \begin{cases}
        3z_1^2 \quad&\text{if } \omega_1 = 1\\
        1 &\text{if } \omega_1 = 0.
    \end{cases}
\end{align*}
%In this case, we can take $K=2$, $\mathcal{S}_1 = \bigl\{\omega\in\{0,1\}^d : \omega_1 = 1\bigr\}$, $\mathcal{S}_2 = \bigl\{\omega\in\{0,1\}^d : \omega_1 = 0\bigr\}$, $f_{\mathcal{S}_1}(z) = (4z_1-2)^2$, $f_{\mathcal{S}_2}(z) = 4/3$, $\mathcal{J}_{\mathcal{S}_1} = \{1\}$ and $\mathcal{J}_{\mathcal{S}_2} = \emptyset$.
\end{example}

\subsection{ReLU neural networks} \label{sec:ReLU-NN}
In this section, we formally define classes of neural networks with ReLU activation function~$\sigma$ given by $\sigma(a) \coloneqq a\vee 0$ for $a\in\mathbb{R}$.  For a vector $\bm{x}=(x_1,\ldots,x_d)^\top\in\mathbb{R}^d$, we define
\begin{align*}
    \bm{\sigma}(\bm{x}) \coloneqq \bigl(\sigma(x_1),\ldots,\sigma(x_d)\bigr)^\top \in [0,\infty)^d,
\end{align*}
though we allow $d$ to vary in different instances of this function without comment.  Given $L\in\mathbb{N}$ and $\bm{p} = (p_0,p_1,\ldots,p_{L+1}) \in \mathbb{N}^{L+2}$, the set of neural networks with architecture $(L,\bm{p})$ is defined as
\begin{align}
    \mathcal{F}(L,\bm{p}) \coloneqq &\Bigl\{ \bm{f}:\mathbb{R}^{p_0} \to \mathbb{R}^{p_{L+1}} :\; \bm{f}(\cdot) = \bm{A}_{L+1} \circ \bm{\sigma} \circ \bm{A}_L \circ \bm{\sigma} \circ \cdots \circ \bm{A}_2 \circ \bm{\sigma} \circ \bm{A}_1(\cdot),\text{ where}\nonumber\\
    & \qquad\qquad \bm{A}_\ell(\bm{v}) = \bm{W}_\ell \bm{v} + \bm{b}_\ell,\, \bm{W}_\ell \in \mathbb{R}^{p_{\ell} \times p_{\ell-1}} \text{ and } \bm{b}_{\ell} \in\mathbb{R}^{p_{\ell}} \;\forall \ell\in[L+1] \Bigr\}. \label{Eq:Flp}
\end{align}
Here, $L$ is the number of hidden layers (or depth of the network), $p_1,\ldots,p_L$ are the widths of the hidden layers, $p_0$ is the input dimension, $p_{L+1}$ is the output dimension, $\bm{W}_1,\ldots,\bm{W}_{L+1}$ are the weight matrices and $\bm{b}_1,\ldots,\bm{b}_{L+1}$ are the bias vectors. For $\bm{f}\in\mathcal{F}(L,\bm{p})$ with weight matrices $(\bm{W}_{\ell})_{\ell=1}^{L+1}$ and bias vectors $(\bm{b}_{\ell})_{\ell=1}^{L+1}$, we define $\bm{\Theta}(\bm{f})$ to be the vector consisting of all the parameters of the neural network $\bm f$, i.e.\ $\bm{\Theta}(\bm{f}) \coloneqq \bigl(\mathrm{vec}(\bm{W}_1)^\top, \bm{b}_1^\top, \ldots,\mathrm{vec}(\bm{W}_{L+1})^\top, \bm{b}_{L+1}^\top\bigr)^\top \in \mathbb{R}^{V}$, where $V \coloneqq \sum_{\ell=1}^{L+1} p_\ell(p_{\ell-1}+1)$ is the total number of parameters. Given $s\in\mathbb{N}$, define
\begin{align*}
    \mathcal{F}(L,\bm{p},s) \coloneqq \bigl\{ \bm{f}\in\mathcal{F}(L,\bm{p}) : \|\bm{\Theta}(\bm{f})\|_0 \leq s \bigr\}
\end{align*}
to be the set of neural networks with architecture $(L,\bm{p})$ and sparsity $s$. Note that $\mathcal{F}(L,\bm p, s) = \mathcal{F}(L,\bm p)$ for all $s\geq V$.

\subsection{Pattern Embedded Neural Network estimators} \label{sec:PENN}

\begin{figure}[ht]
    \centering
    \includegraphics[width=0.8\linewidth]{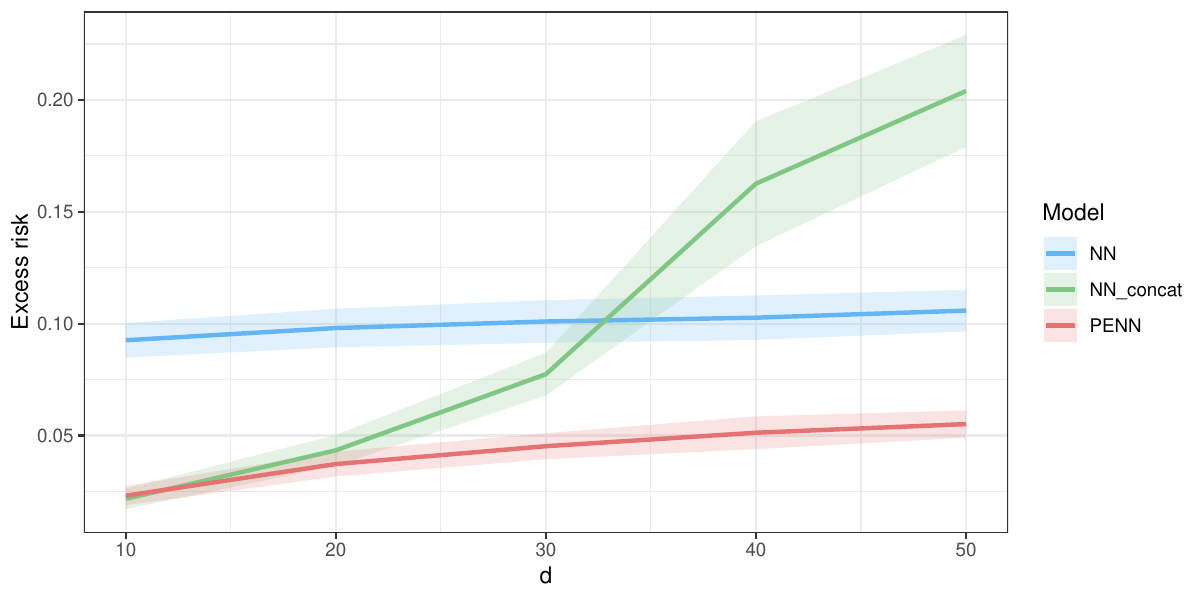}
    \caption{Excess risk for Example \ref{example:f-star} with $d\in\{10,20,\ldots,50\}$ and $n=1000$. NN (blue) is a neural network that regresses $Y_i$ onto $\bm Z_i$; NN$\_$concat (green) is a neural network that directly regresses $Y_i$ onto $(\bm Z_i, \bm\Omega_i)$; and PENN (red) is our pattern embedded neural network (with $m=3$).  The shaded regions represent $90\%$ confidence intervals for the true excess risk, based on 100 repetitions of the experiment.}
    \label{fig:quadratic-d}
\end{figure}

\begin{figure}[ht]
\centering
\begin{subfigure}[t]{0.3\textwidth}
    \includegraphics[width=\textwidth]{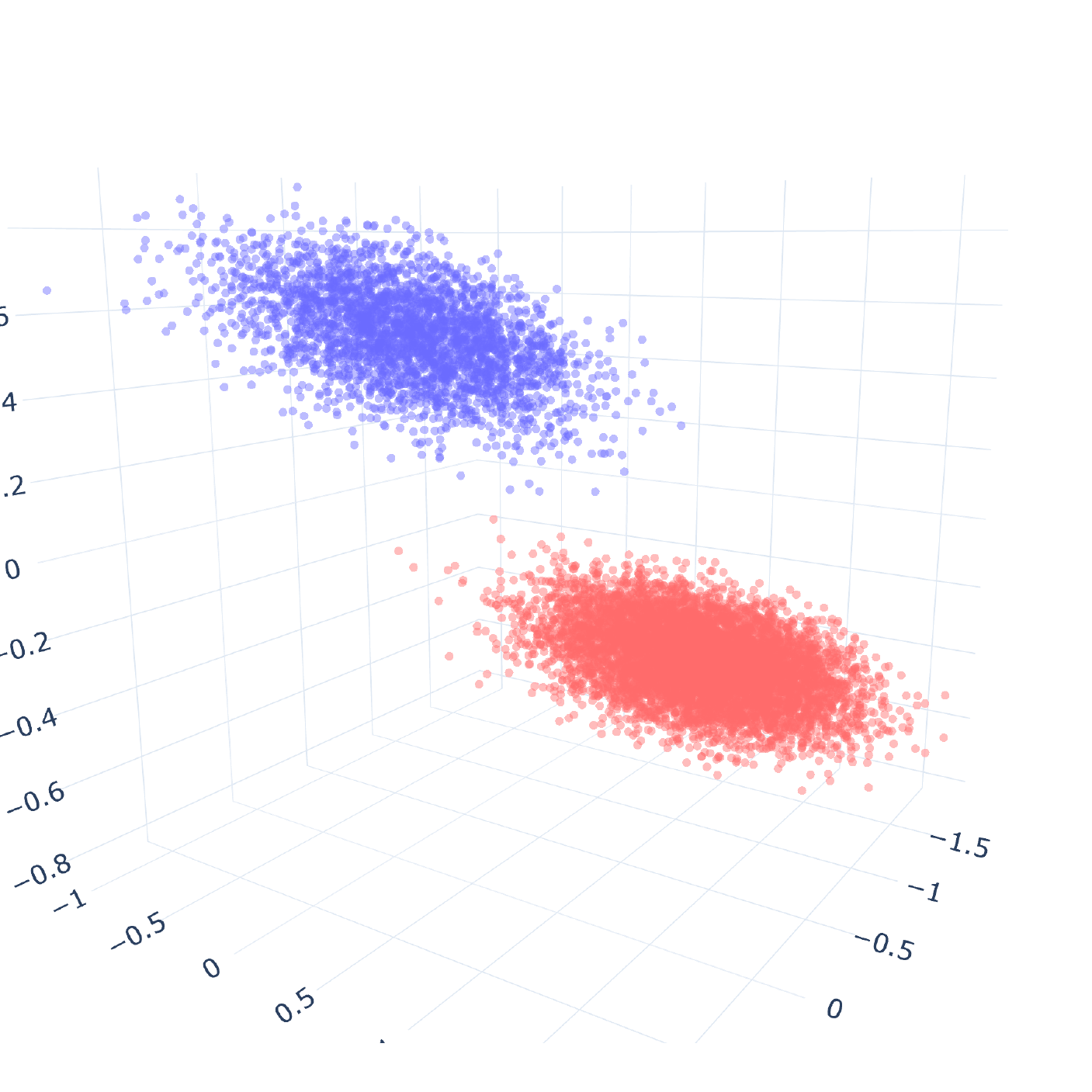}
    \caption{$d=40$,\\ \hphantom{(a)} excess risk 0.0289.}
\end{subfigure}
\begin{subfigure}[t]{0.3\textwidth}
    \includegraphics[width=\textwidth]{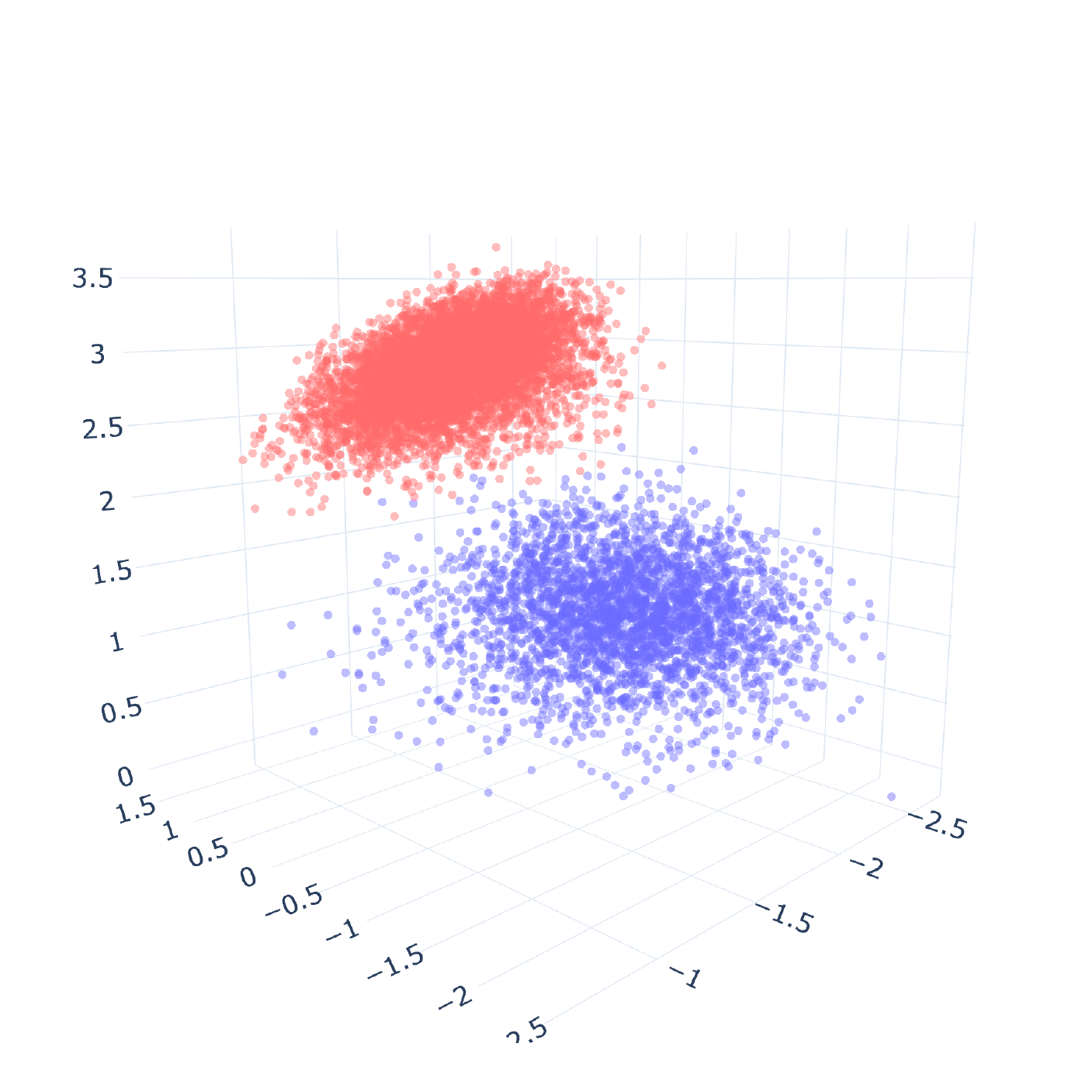}
    \caption{$d=40$,\\ \hphantom{(c)} excess risk 0.0523.}
\end{subfigure}
\begin{subfigure}[t]{0.3\textwidth}
    \includegraphics[width=\textwidth]{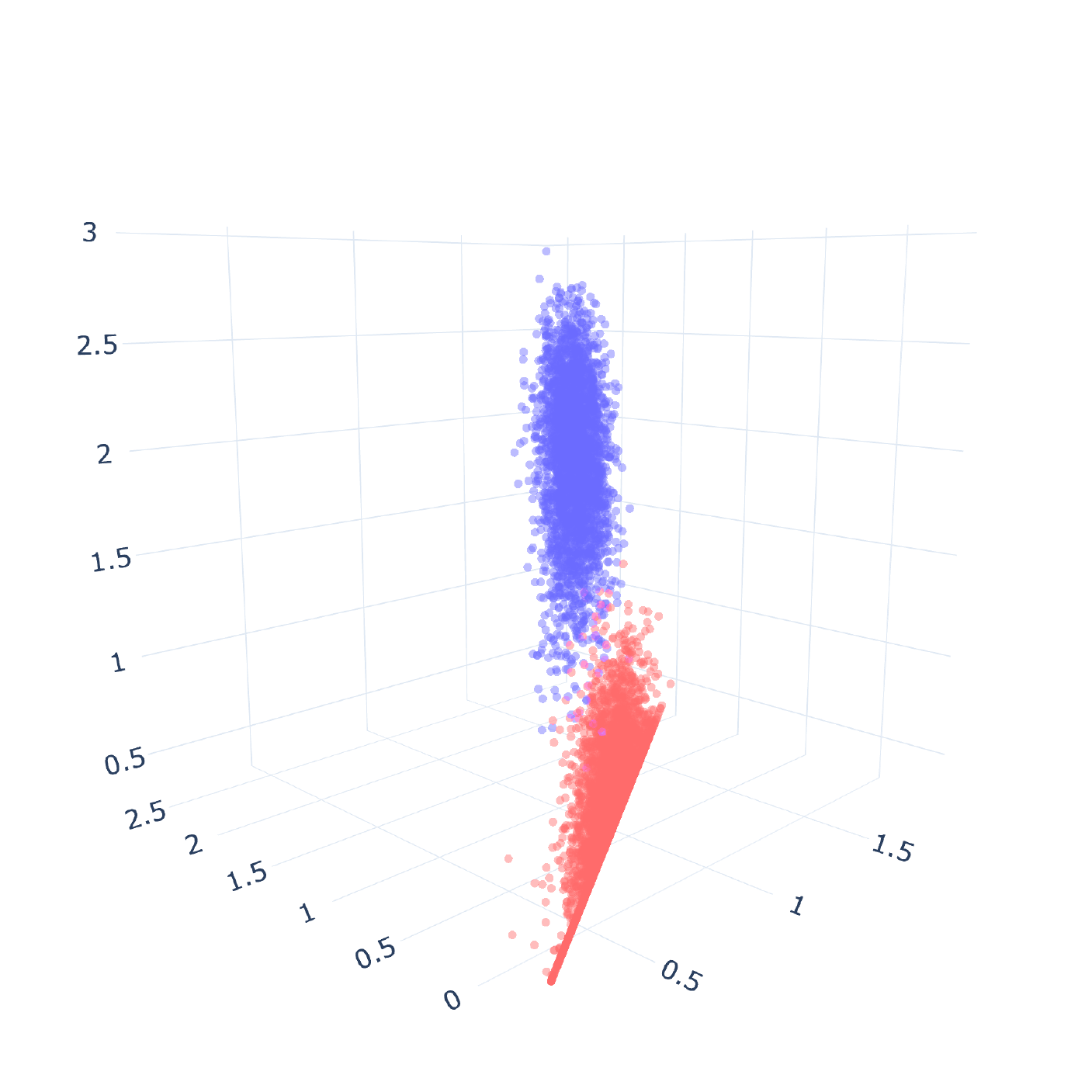}
    \caption{$d=40$,\\ \hphantom{(b)} excess risk 0.0680.}
\end{subfigure}\\ \hspace{0.75cm}

\begin{subfigure}[t]{0.3\textwidth}
    \includegraphics[width=\textwidth]{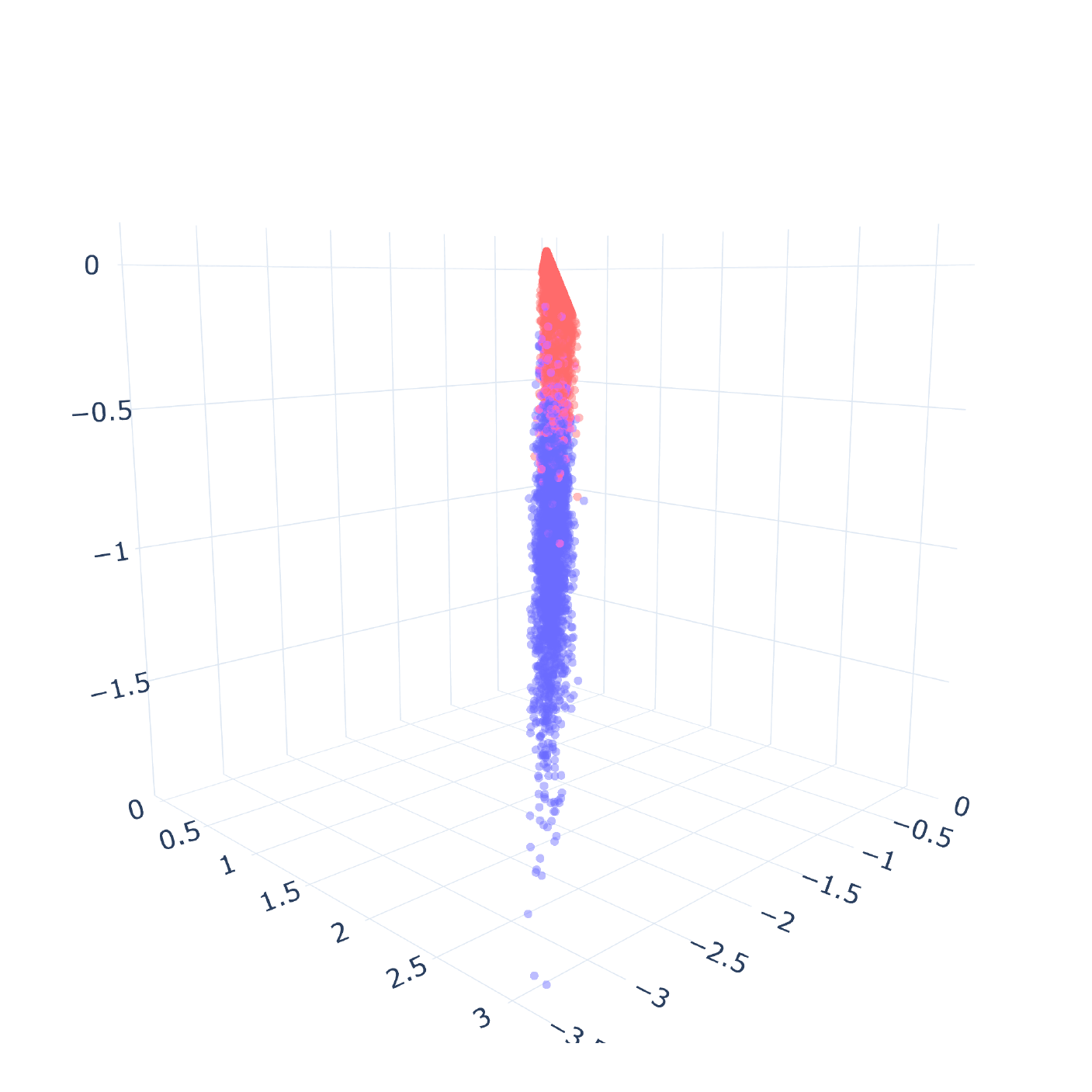}
    \caption{$d=100$, \\\hphantom{(d)} excess risk 0.0421.}
\end{subfigure}
\begin{subfigure}[t]{0.3\textwidth}
    \includegraphics[width=\textwidth]{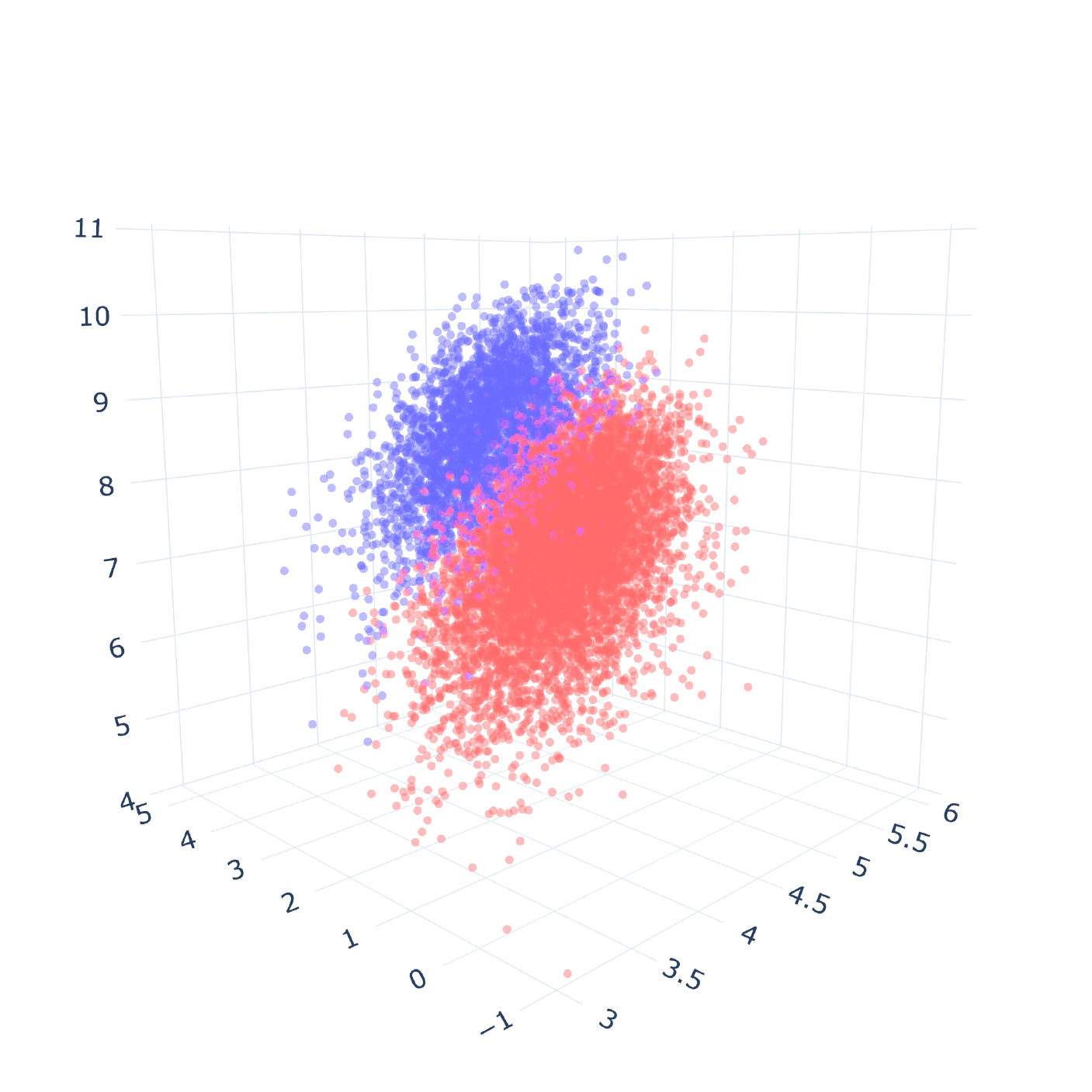}
    \caption{$d=100$, \\\hphantom{(e)} excess risk 0.0866.}
\end{subfigure}
\begin{subfigure}[t]{0.3\textwidth}
    \includegraphics[width=\textwidth]{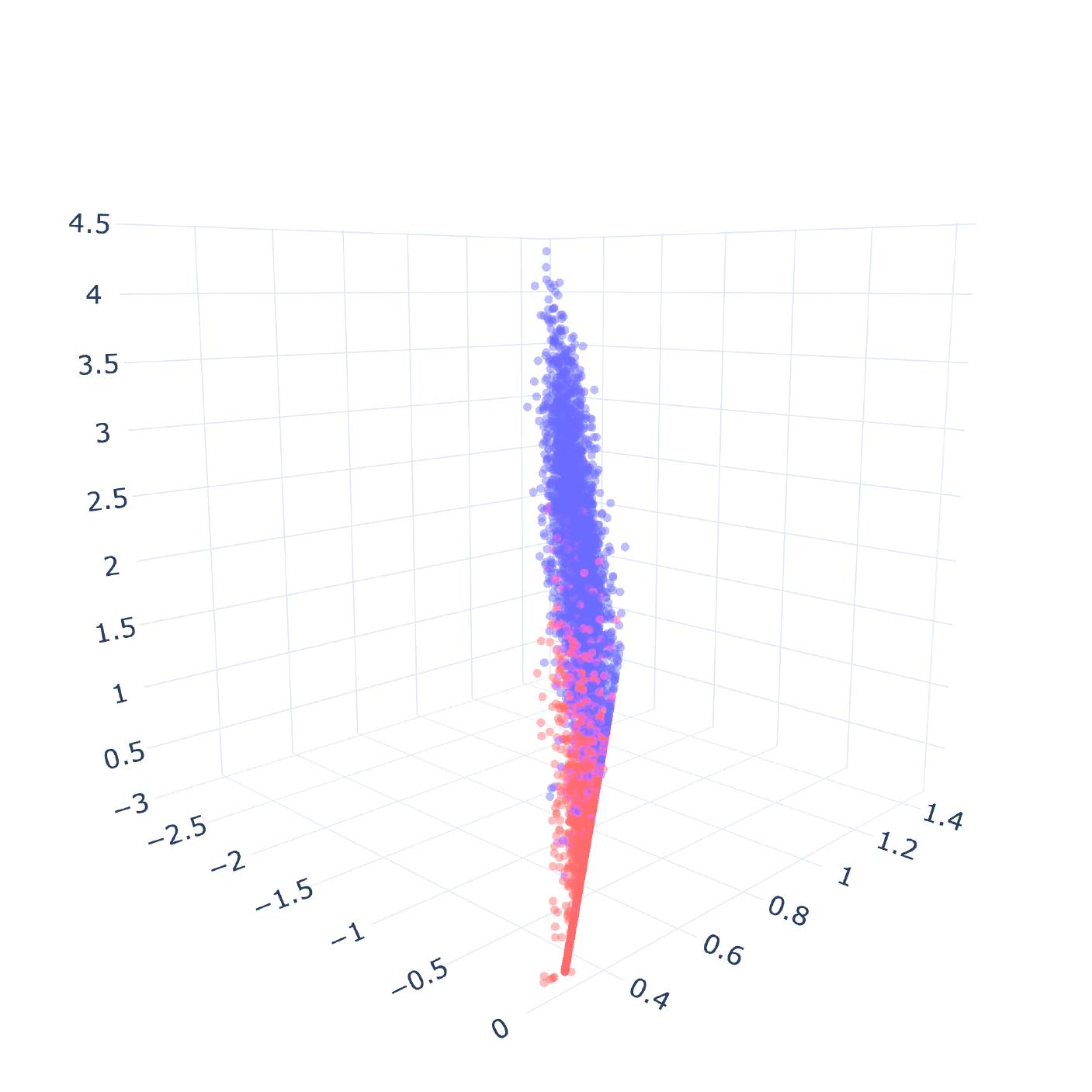}
    \caption{$d=100$, \\\hphantom{(f)} excess risk 0.1074.}
\end{subfigure}

\caption{Visualisations of the 3-dimensional embeddings of the observation patterns extracted from six trained PENN estimators under the settings of Example~\ref{example:f-star} with $d\in\{40,100\}$ and $n=1000$. The blue and red points correspond to the patterns whose first coordinates are zero and one, respectively. The subcaptions indicate dimension and the excess risks of the corresponding PENNs on test data.} 
\label{fig:embedding}
\end{figure}

%Motivated by ideas of \emph{categorical variable embedding} and \emph{autoencoders} in the deep learning literature \citep{hinton2006reducing,mikolov2013efficient}, 
% To address the concern raised in the previous paragraph, we propose a \emph{pattern embedding} method to map the revelation vectors into a lower dimensional space.  The idea is that, while the revelation vectors, which lie on the vertices $\{0,1\}^d$ of a high-dimensional hypercube, often contain highly useful information for prediction, this information may not be adequately captured by the original Euclidean geometry.  For instance, in Example~\ref{example:f-star}, the revelation vectors $(1,0,\ldots,0), (1,1,\ldots,1) \in \{0,1\}^d$ yield identical regression functions, despite the fact that their Euclidean distance of $\sqrt{d-1}$ is almost maximal.  Thus, the revelation vectors can benefit from a significant (nonlinear) compression prior to incorporation in our regression model.  

Our aim is to estimate $f^{\star}$ by regressing $(Y_i)_{i=1}^n$ onto $(\bm Z_i,\bm\Omega_i)_{i=1}^n$ using neural networks, while borrowing strength across similar observation patterns.  Naive training of a neural network using the concatenated vectors $(\bm Z_i,\bm\Omega_i)_{i=1}^n$ as covariates leads to significant overfitting when the dimension $d$ is large (see Figure~\ref{fig:quadratic-d}), and, indeed the excess risk may be worse than that obtained by regressing $(Y_i)_{i=1}^n$ on $(\bm Z_i)_{i=1}^n$, ignoring the revelation vectors.  The main problem is that the number of observation patterns $|\mathcal{S}|$ may grow exponentially in $d$, so each data point may even have a unique observation pattern, and naive training is unable to borrow strength across similar observation patterns.  Moreover, while the revelation vectors $\bm\Omega_i\in\{0,1\}^d$ may carry highly useful information for prediction, the raw Euclidean geometry of the hypercube $\{0,1\}^d$ may align poorly with the predictive structure of the problem. For instance, in Example~\ref{example:f-star}, two revelation vectors that differ in all but the first coordinate yield identical regression functions, despite the fact that their Euclidean distance of $\sqrt{d-1}$ is almost maximal. As a result, direct incorporation of the revelation vectors as covariates forces the model to fit into an extraneous geometric structure that almost amounts to noise. 

To remedy the issue raised in the previous paragraph, we seek to construct a learned embedding of the revelation vectors into a lower-dimensional Euclidean space where geometric proximity corresponds to similarity in predictive behaviour (see Figure~\ref{fig:embedding}). Such an embedding allows the model to borrow statistical strength across related missingness patterns, so as to provide better generalisation for rare patterns.  Figure~\ref{fig:embedding} illustrates the predictive value of our learned embeddings of the revelation vectors in Example~\ref{example:f-star}, and in particular the fact that better separation in our embeddings is associated with lower excess risks.  Naturally, clusters tend to have greater overlap in higher dimensions, reflecting the increased difficulty of the problem; nonetheless, even when $d=100$, our learned three-dimensional embedding still provides considerable predictive information.  The embedding illustrated takes the form of a neural network that maps the $d$-dimensional revelation vectors to vectors in $\mathbb{R}^m$ for some \emph{embedding dimension} $m \leq d$ (so $m=3$ in the figure).  As a secondary consideration, the imputed covariates $\bm Z_1,\ldots,\bm Z_n$ can also benefit from a nonlinear feature map, before both sets of transformed variables are combined in a further flexible regression model that is able to capture appropriate interactions.  We again take these nonlinear functions to be neural networks, which are convenient for both modelling and fitting.

\begin{figure}[htbp]
    \centering
    \begin{tikzpicture}[scale = 0.8,
        input neuron/.style={
            circle,
            draw=green,
            fill=green!20,
            minimum size=1cm,
            inner sep=0pt,
            scale = 0.8
        },
        hidden neuron/.style={
            circle,
            draw=blue,
            fill=blue!20,
            minimum size=1cm,
            inner sep=0pt,
            scale = 0.8
        },
        output neuron/.style={
            circle,
            draw=orange,
            fill=orange!20,
            minimum size=1cm,
            inner sep=0pt,
            scale = 0.8
        },
        neuron missing/.style={
            draw=none,
            minimum width=1cm,
            minimum height=1cm,
            inner sep=0pt,
            text centered,
            execute at begin node={\centering\color{black}$\vdots$}
        },
        arrow style/.style={
        ->, >={latex}, opacity=0.8, line width=0.3mm
        }]

        % input layer
        \node[input neuron] (I-1) at (0,3.75) {$z_1$};
        \node[neuron missing] at (0,2.65) {};
        \node[input neuron] (I-2) at (0,1.25) {$z_d$};
        \node[input neuron] (I-3) at (0,-1.25) {$\omega_1$};
        \node[neuron missing] at (0,-2.35) {};
        \node[input neuron] (I-4) at (0,-3.75) {$\omega_d$};

        % hidden layer 1
        \node[hidden neuron] (H-1-1) at (2.5,4.15) {};
        \node[neuron missing] at (2.5,2.65) {};
        \node[hidden neuron] (H-1-2) at (2.5,0.85) {};
        \node[hidden neuron] (H-1-3) at (2.5,-0.85) {};
        \node[neuron missing] at (2.5,-2.35) {};
        \node[hidden neuron] (H-1-4) at (2.5,-4.15) {};

        % hidden layer 2
        \node[hidden neuron] (H-2-1) at (5,4) {};
        \node[neuron missing] at (5,2.65) {};
        \node[hidden neuron] (H-2-2) at (5,1) {};
        \node[hidden neuron] (H-2-3) at (5,-1.65) {$f_{2,1}$};
        \node[neuron missing] at (5,-2.35) {};
        \node[hidden neuron] (H-2-4) at (5,-3.35) {$f_{2,m}$};

        % hidden layer 3
        \node[hidden neuron] (H-3-1) at (7.5,1.8) {};
        \node[neuron missing] at (7.5,0.2) {};
        \node[hidden neuron] (H-3-2) at (7.5,-1.8) {};

        % hidden layer 4
        \node[hidden neuron] (H-4-1) at (10,1.8) {};
        \node[neuron missing] at (10,0.2) {};
        \node[hidden neuron] (H-4-2) at (10,-1.8) {};

        % output layer
        \node[output neuron] (O-1) at (12.5,0) {$y$};

        \draw[arrow style] (I-1) -- (H-1-1);
        \draw[arrow style] (I-1) -- (H-1-2);
        \draw[arrow style] (I-2) -- (H-1-1);
        \draw[arrow style] (I-2) -- (H-1-2);
        \draw[arrow style] (I-3) -- (H-1-3);
        \draw[arrow style] (I-3) -- (H-1-4);
        \draw[arrow style] (I-4) -- (H-1-3);
        \draw[arrow style] (I-4) -- (H-1-4);

        \draw[arrow style] (H-1-1) -- (H-2-1);
        \draw[arrow style] (H-1-1) -- (H-2-2);
        \draw[arrow style] (H-1-2) -- (H-2-1);
        \draw[arrow style] (H-1-2) -- (H-2-2);
        \draw[arrow style] (H-1-3) -- (H-2-3);
        \draw[arrow style] (H-1-3) -- (H-2-4);
        \draw[arrow style] (H-1-4) -- (H-2-3);
        \draw[arrow style] (H-1-4) -- (H-2-4);

        \draw[arrow style] (H-2-1) -- (H-3-1);
        \draw[arrow style] (H-2-1) -- (H-3-2);
        \draw[arrow style] (H-2-2) -- (H-3-1);
        \draw[arrow style] (H-2-2) -- (H-3-2);
        \draw[arrow style] (H-2-3) -- (H-3-1);
        \draw[arrow style] (H-2-3) -- (H-3-2);
        \draw[arrow style] (H-2-4) -- (H-3-1);
        \draw[arrow style] (H-2-4) -- (H-3-2);

        \draw[arrow style] (H-3-1) -- (H-4-1);
        \draw[arrow style] (H-3-1) -- (H-4-2);
        \draw[arrow style] (H-3-2) -- (H-4-1);
        \draw[arrow style] (H-3-2) -- (H-4-2);

        \draw[arrow style] (H-4-1) -- (O-1);
        \draw[arrow style] (H-4-2) -- (O-1);

        \draw[dashed, rounded corners=5mm] (-0.7, 4.8) rectangle (5.7, 0.2);
        \draw[dashed, rounded corners=5mm] (-0.7, -4.8) rectangle (5.7, -0.2);
        \draw[dashed, rounded corners=5mm] (4.2, -4.1) rectangle (13.3, 4.6);

        \draw[->, thick, bend left=20] (2,4.8) to (3,5.3);
        \draw[->, thick, bend right=20] (2,-4.8) to (3,-5.3);
        \draw[->, thick, bend left=20] (8.5,4.6) to (9.5,5.3);

        \node[right] at (3,5.3) {$\bm f_1$};
        \node[right] at (3,-5.3) {$\bm f_2$};
        \node[right] at (9.5,5.3) {$\bm f_3$};
    \end{tikzpicture}
    \caption{An illustration of the class $\mathcal{F}_{\mathrm{PENN}}\bigl((L_r,\bm{p}_r)_{r=1}^3,s\bigr)$.} \label{fig:PENN-diagram}
\end{figure}

Our \emph{pattern embedded neural networks} (PENNs) therefore take the form $(\bm z,\bm \omega) \mapsto \bm f_3\bigl(\bm f_1(\bm z), \bm f_2(\bm\omega)\bigr)$, where $\bm f_r$ is a neural network for $r \in \{1,2,3\}$.  Here,~$\bm f_1$ denotes the feature map for the imputed covariates, $\bm f_2$ denotes the pattern embedding  for the revelation vectors and $\bm f_3$ combines the learned covariate features and pattern embeddings.  More precisely, for $r \in \{1,2,3\}$, let $L_r \in \mathbb{N}$, suppose that $\bm p_r = (p_{r,0},\ldots,p_{r,L_r+1}) \in \mathbb{N}^{L_r+2}$ satisfies $p_{3,0} = p_{1,L_1+1} + p_{2,L_2+1}$, and let $s\in\mathbb{N}$; writing $p_{\mathrm{in}} \coloneqq p_{1,0} + p_{2,0}$ and $p_{\mathrm{out}} \coloneqq p_{3,L_3+1}$, the class of \emph{pattern embedded neural networks} (PENNs) is given by
\begin{align}
\mathcal{F}_{\mathrm{PENN}}\bigl((L_r,\bm{p}_r)_{r=1}^3 &,s\bigr) \coloneqq \biggl\{ \bm f : \mathbb{R}^{p_{\mathrm{in}}} \to \mathbb{R}^{p_{\mathrm{out}}} : \bm f(\cdot,\cdot\cdot) = \bm f_3\bigl(\bm f_1(\cdot), \bm f_2(\cdot\cdot)\bigr), \nonumber \\
        & \text{where } \bm f_r \in \mathcal{F}(L_r,\bm p_r) \text{ for } r\in\{1,2,3\} \text{ and } \sum_{r=1}^3 \|\bm\Theta(\bm f_r)\|_0 \leq s \biggr\}. \label{eq:PENN-def}
\end{align}
See Figure~\ref{fig:PENN-diagram} for an illustration of a PENN.  In the setting of Section~\ref{sec:setup}, we seek an estimator $\hat{f}$ that minimises the empirical risk $\hat{R}_n(f)$ over an appropriate PENN class as defined in~\eqref{eq:PENN-def}.  To this end, our PENN is implemented as a single neural network (using standard software such as \texttt{PyTorch} \citep{paszke2019pytorch}) and trained by gradient methods, though we remark that the optimisation problem involved is non-convex, and our theory in Section~\ref{sec:theoretical-results} therefore allows for a residual optimisation error.  In our applications, we will set $p_{1,0} = p_{2,0} = d$, $p_{2,L_2+1} = m$ and $p_{3,L_3+1} = 1$.  %The function $\bm f_2 = (f_{2,1}, \ldots, f_{2,m})^\top : \mathbb{R}^d \to \mathbb{R}^m$ is called the \emph{embedding function} of $\bm f$, and $\bm f(\bm z, \bm \omega)$ denotes the output of $\bm f \in \mathcal{F}_{\mathrm{PENN}}$ evaluated at a test point $(\bm z, \bm \omega) \in \mathbb{R}^d \times \mathcal{S}$.  

\section{Theoretical results} \label{sec:theoretical-results}

The main goal of our theory in this section is to derive a theoretical guarantee on the performance of the PENN methodology introduced in Section~\ref{sec:PENN}.  Our overall strategy is first to decompose the excess risk of our estimator into a sum of three terms, representing the optimisation error, approximation error and estimation error respectively; see Proposition~\ref{prop:oracle-inequality} below.  This decomposition is quite general, in the sense that in addition to requiring no conditions on the dependence structure of the data generating and missingness mechanisms, we also make no assumptions on the form of the Bayes regression function $f^\star$ and consider an arbitrary neural network estimator.  A consequence of this generality, however, is that the approximation error term is not very explicit, in that it simply reflects the extent to which $f^\star$ can be approximated by an element of our neural network class.  For this reason, we then introduce natural assumptions that allow us to borrow strength across different observation patterns, as well as smoothness conditions, that provide an explicit form for the approximation error.  Thus, in our main result (Theorem~\ref{thm:PENN-ub}), we are able to choose the architecture of our PENN estimator so as to balance the approximation and estimation error terms to give a finite-sample excess risk bound that holds under no assumptions on the form of the missingness mechanism.  This upper bound is complemented by a minimax lower bound (Theorem~\ref{thm:minimax-lb}) that reveals the optimality of the PENN estimator up to poly-logarithmic factors in the sample size.

\subsection{Oracle inequality} \label{sec:oracle-inequality}

Following the strategy described above, let $L\in\mathbb{N}$ and $\bm{p}=(p_0,\ldots,p_{L+1})^\top \in \mathbb{N}^{L+2}$ with $p_0=2d$ and $p_{L+1}=1$.
By a \emph{neural network estimator} in $\mathcal{F} \subseteq \mathcal{F}(L,\bm p,s)$, we mean a jointly measurable function $\tilde{f}:\mathbb{R}^d \times \mathcal{S} \times (\mathbb{R}^d \times \mathcal{S} \times \mathbb{R})^n  \rightarrow \mathbb{R}$ such that for each $\mathcal{D} \coloneqq (\bm z_i,\bm \omega_i,y_i)_{i=1}^n \in (\mathbb{R}^d \times \mathcal{S} \times \mathbb{R})^n$, the function $\tilde{f}(\cdot,\cdot;\mathcal{D}) \in \mathcal{F}$.  Where the data are clear from context, we often omit the final argument of~$\tilde{f}$. 

\begin{prop}\label{prop:oracle-inequality}
    % \blue{Old:} Consider the setup in Section~\ref{sec:setup}. Assume that $\|Y_0\|_{\psi_2} \leq \xi$ for some $\xi\geq 1$. Let $L\in\mathbb{N}$, $\bm{p}=(p_0,\ldots,p_{L+1})^\top \in \mathbb{N}^{L+2}$ with $p_0=2d$ and $p_{L+1}=1$, and let $V\coloneqq \sum_{\ell=1}^{L+1}p_{\ell}(p_{\ell-1}+1)$ and let $s\in[V]$. Let $B_n \coloneqq \xi\sqrt{2\log n}$, $\tilde{f}$ be a neural network estimator in $\mathcal{F} \subseteq \mathcal{F}(L,\bm p,s)$ based on data $\mathcal{D} \coloneqq (\bm Z_i,\bm \Omega_i,Y_i)_{i=1}^n$ and let $D$ be an arbitrary measurable subset of $(\mathcal{X} \odot \mathcal{S}) \times \mathcal{S}$.  Then there exists a universal constant $C>0$ such that for $n \geq 2$,
    % \begin{align*}
    %     &\mathbb{E}\bigl\{ R(T_{B_n}\tilde{f}) - R(f^{\star}) \bigr\}\\
    %     &\qquad \leq 2\mathbb{E}\Bigl\{ \hat{R}_n(\tilde{f}) - \inf_{f\in\mathcal{F}}\hat{R}_n(f) \Bigl\} + 2\inf_{f\in\mathcal{F}} \mathbb{E}\Bigl\{ \bigl(f(\bm Z_0,\bm \Omega_0) - f^{\star}(\bm Z_0,\bm \Omega_0)\bigr)^2 \mathbbm{1}_{\{(\bm Z_0, \bm\Omega_0) \in D\}} \Bigr\} \\
    %     &\qquad\qquad + 8nB_n^2 \mathbb{P}\bigl((\bm Z_0, \bm\Omega_0) \notin D \bigr) + C\xi^4\log (e\xi) \log^3n \cdot \frac{sL\log (es) + s\log(ed)}{n}.
    % \end{align*}
    In the setting of Section~\ref{sec:setup}, assume that $\|Y_0\|_{\psi_2} \leq \xi$ for some $\xi\geq 1$. Let~$\tilde{f}$ be a neural network estimator in $\mathcal{F}\subseteq \mathcal{F}(L,\bm p, s)$ based on data $\mathcal{D} \coloneqq (\bm Z_i,\bm \Omega_i,Y_i)_{i=1}^n$ and let $B_n \coloneqq \xi\sqrt{2\log n}$.  Then there exists a universal constant $C>0$ such that %\footnote{The fact that $\inf_{f\in\mathcal{F}}\hat{R}_n(f)$ is measurable is also established in the proof.} 
    for $n \geq 2$,
    \begin{equation*}
    \begin{aligned}
        \mathbb{E}\bigl\{ R(T_{B_n}\tilde{f})\bigr\} - R(f^{\star}) &\leq 2\mathbb{E}\Bigl\{ \hat{R}_n(\tilde{f}) - \inf_{f\in\mathcal{F}}\hat{R}_n(f) \Bigl\} + 2\inf_{f\in\mathcal{F}} \mathbb{E}\Bigl\{ \bigl(f(\bm Z_0,\bm \Omega_0) - f^{\star}(\bm Z_0,\bm \Omega_0)\bigr)^2 \Bigr\}\nonumber\\
        &\hspace{3cm} + \frac{C\xi^4\log(e\xi) \log^3 n \cdot \bigl(sL\log(es) + s\log(ed)\bigr)}{n}.
        \end{aligned}
    \end{equation*}
\end{prop}
As a simple illustration of the bound on $\|Y_0\|_{\psi_2}$ in Proposition~\ref{prop:oracle-inequality}, if $Y_0 = f^0(\bm X_0) + \varepsilon_0$ where $\|f^0\|_{\infty} \leq \xi_1$ and $\|\varepsilon_0\|_{\psi_2} \leq \xi_2$, then $\|Y_0\|_{\psi_2} \leq \xi_1/\sqrt{\log 2} + \xi_2$. The upper bound on the excess risk in Proposition~\ref{prop:oracle-inequality} is a sum of three terms, where the first term corresponds to the optimisation error, the second term represents the approximation error, and the last term corresponds to estimation error and reflects the complexity of the function class $\mathcal{F}$.  If~$\tilde{f}$ is an empirical risk minimiser, then the optimisation error term vanishes.  Importantly, we do not insist that $f^\star \in \mathcal{F}$, though if it does, then we may take the approximation error term to be zero.  The estimation error term is of order $sL/n$ up to poly-logarithmic factors.  In general, there is a trade-off between the approximation and estimation error terms that is akin to a bias--variance trade-off: more complex classes $\mathcal{F}$ will have smaller approximation error, but the price to be paid is through larger values of $s$ and $L$ in the estimation error term.  Related results to Proposition~\ref{prop:oracle-inequality} in the context of fully observed data include those of \citet{schmidt-hieber2020nonparametric}, who has an additional boundedness assumption on the parameter and input spaces, and \citet{kohler2021rate}, who work with fully dense neural networks.  The key to our proof is a new bound on the Vapnik--Chervonenkis dimension and covering numbers of the $\mathcal{F}(L,\bm p,s)$ class given in Proposition~\ref{prop:VC-upper-bound}, which allows us to remove constraints in prior work on the boundedness of covariates and the parameters in our neural networks.

\subsection{Minimax rate under a piecewise smoothness assumption} \label{sec:minimax-rate}

% \red{Reasons for section 3.2:
% \begin{enumerate}
%     \item Assumption 2 demonstrates another advantage of PENN: neural networks are known to be able to avoid the curse of dimensionality when there are compositional structures in the regression function with complete data (whereas linear estimator cannot do this), we show that our PENN also enjoys this advantage when each function $f^{\mathcal{S}_k}$ is a composition of H\"older functions. Also note that tree-base methods are not optimal when the smoothness is larger than $1$.
% \end{enumerate}} 

In principle, the functions $\bm z \mapsto f^\star(\bm z, \bm \omega)$ may be arbitrarily different as $\bm \omega$ varies through~$\mathcal{S}$, and indeed Proposition~\ref{prop:oracle-inequality} does not restrict the differences between these functions as~$\bm \omega$ varies.  In the worst case, then, we may have $|\mathcal{S}| = 2^d$ different regression functions to learn!  Fortunately, in many circumstances, it is reasonable to postulate that the Bayes regression function at certain different observation patterns may be similar, and this will provide great advantages for our theory. 

%As outlined at the beginning of Section~\ref{sec:theoretical-results}, we now introduce conditions on the Bayes regression function that allow us to exploit the fact that the Bayes regression function at certain different observation patterns may be similar.  

%The general theory of Section~\ref{sec:oracle-inequality} allows us to provide an explicit upper bound on the excess risk of our PENN estimator under a piecewise smoothness assumption on the regression function; see Theorem~\ref{thm:PENN-ub} below.  Our bound reveals that improved bounds are achievable in settings where the regression function is a composition of (smooth) functions that depend only on a subset of the variables, thereby providing a sense in which the estimator is able to evade the curse of dimensionality. 

%In addition to tail conditions on the covariates and response, Assumption~\ref{assumption:piecewise-assumption} introduces a structural condition on the Bayes regression function, recognising that this regression function may be the same for certain different missingness patterns.
\begin{assumption} \label{assumption:piecewise-assumption}
    Assume that each coordinate of $\bm Z_0$ is sub-exponential and that $Y_0$ is sub-Gaussian, i.e.~there exist $\xi_1,\xi_2>0$ such that $\|Z_{0,j}\|_{\psi_1} \leq \xi_1$ for all $j\in[d]$ and $\|Y_0\|_{\psi_2} \leq \xi_2$. Further assume that there exist $K\leq |\mathcal{S}|$, a partition $\{\mathcal{S}_1,\ldots,\mathcal{S}_K\}$ of $\mathcal{S}$, and functions $f^{\mathcal{S}_1},\ldots,f^{\mathcal{S}_K} : \mathbb{R}^d \to \mathbb{R}$ such that\footnote{More formally, since $f^\star$ is only defined up to sets having zero measure under the distribution of $(\bm{Z}_0,\bm{\Omega}_0)$, we ask that there exists a version of $f^\star$ for which the statement in Assumption~\ref{assumption:piecewise-assumption} holds.}
    \begin{align}
        f^{\star}(\bm z, \bm\omega) = \sum_{k=1}^K f^{\mathcal{S}_k}(\bm z) \mathbbm{1}_{\{\bm\omega \in \mathcal{S}_k\}} \label{eq:piecewise-f-star}
    \end{align}
    for all $\bm z \in \mathbb{R}^d$ and $\bm\omega \in \mathcal{S}$.  For $k \in [K]$, we write $\pi_k \coloneqq \mathbb{P}(\bm \Omega_0 \in \mathcal{S}_k)$ and write $n_k \coloneqq n\pi_k$ for the \emph{effective sample size} in the $k$th cell. 
\end{assumption}
The tail condition on $\bm Z_0$ allows our covariates to be unbounded, in contrast to much of the literature on nonparametric regression.  Of course, any Bayes regression function~$f^{\star}$ always satisfies~\eqref{eq:piecewise-f-star} with $K = |\mathcal{S}|$. However,~$K$ may be much smaller than $|\mathcal{S}|$; e.g.~in Example~\ref{example:f-star} we have $|\mathcal{S}| = 2^d$ and $K=2$, since we may take $\mathcal{S}_1 = \bigl\{\bm{\omega} \in \{0,1\}^d:\omega_1 = 1\bigr\}$ and $\mathcal{S}_2 = \bigl\{\bm{\omega} \in \{0,1\}^d:\omega_1 = 0\bigr\}$.  In such circumstances, we can borrow strength across data whose revelation vectors lie in the same cell of the partition.  It turns out that Assumption~\ref{assumption:piecewise-assumption} may be weakened to requiring only that~\eqref{eq:piecewise-f-star} holds approximately; see the discussion following Theorem~\ref{thm:PENN-ub} below. 

Given an arbitrary partition $\{\mathcal{S}_1,\ldots,\mathcal{S}_K\}$ of $\mathcal{S}$, there exists a Bayes regression function~$f^{\star}$ that satisfies Assumption~\ref{assumption:piecewise-assumption} with that partition; see Lemma~\ref{lemma:arbitrary-partition}. 
% Indeed, let $\bm{X} = (X_1,\ldots,X_d)^\top \sim \mathrm{Unif}[0,1]^d$ and let $Y = \sum_{k=1}^K \mathbbm{1}_{\{X_{1}\geq (k-1)/K\}} + \varepsilon$ where $\varepsilon \sim N(0,1)$ is independent of $\bm X$. Further define $\bm \Omega$ by $\bm \Omega \,|\, \bm X \sim \mathrm{Unif}( \mathcal{S}_k)$ when $X_1 \in \bigl[(k-1)/K,k/K\bigr)$. Then $f^{\star}(\bm z,\bm \omega) = k$ for all $\bm z \in [0,1]^d$ whenever $\bm \omega\in\mathcal{S}_k$.
Beyond the cardinality $K$ of this partition, its complexity can be measured through the complexity of classes of functions (e.g.~neural networks) that can cluster observation patterns in the same cell of the partition.  To this end, we introduce a notion of separability of a partition of~$\mathcal{S}$.
%However, in practice, the partition typically exhibits additional structure that enables a representation by a function class of controlled complexity in a low-dimensional latent space with clear separation between its parts.
\begin{defn} \label{defn:F(L,p,s)-separable}
Given a class $\mathcal{F}$ of functions from $\mathbb{R}^d$ to $\mathbb{R}^m$, we say a partition $\{\mathcal{S}_1,\ldots,\mathcal{S}_K\}$ of~$\mathcal{S} \subseteq \{0,1\}^d$ is \emph{$\mathcal{F}$-separable} if there exist $\bm f \in \mathcal{F}$, $\bm v_1,\ldots,\bm v_K \in \mathbb{R}^m$ and $\epsilon>0$ such that 
    \begin{itemize}
        \item[(i)] $\|\bm f(\bm\omega) - \bm v_k\|_{\infty} \leq \epsilon/2$ for all $k\in[K]$ and $\bm\omega \in \mathcal{S}_k$;
        \item[(ii)] $\|\bm v_k - \bm v_{k'}\|_{\infty} \geq 2\epsilon$ for all $k \neq k'$.
    \end{itemize}
    In this case, we say that $\{\mathcal{S}_1,\ldots,\mathcal{S}_K\}$ is separated by $\bm f$.
\end{defn}
Thus $\{\mathcal{S}_1,\ldots,\mathcal{S}_K\}$ is separated by $\bm f$ if the scale over which $\bm f$ varies on each $\mathcal{S}_k$ is small by comparison with its variability across different cells of the partition. Proposition~\ref{prop:F(L,p,s)-separable} guarantees that an arbitrary partition $\{\mathcal{S}_1, \ldots, \mathcal{S}_K\}$ of $\mathcal{S}$ is $\mathcal{F}(2,\bm p)$-separable for a suitable choice of $\bm p$. Moreover, if the partition is defined by a small number of coordinates (as in part~\emph{(a)} of that result) or a small number of halfspaces (as in part~\emph{(b)}), then the partition may be separated by a class of neural networks with fewer parameters.

We now introduce smoothness assumptions that will be imposed on each~$f^{\mathcal{S}_k}$ in~\eqref{eq:piecewise-f-star}.    
% \begin{defn}
%     Let $D \subseteq \mathbb{R}^d$, $f:D \to \mathbb{R}$ and $\mathcal{J} \subseteq [d]$. We say that $f$ \emph{depends only on the coordinates in $\mathcal{J}$} if there exists $g:\mathbb{R}^{|\mathcal{J}|} \to \mathbb{R}$ such that for all $\bm x = (x_1,\ldots,x_d)^\top \in D$, we have $f(\bm x) = g(\bm x_{\mathcal{J}})$ where $\bm x_{\mathcal{J}} \coloneqq (x_j)_{j\in\mathcal{J}}$; when $\mathcal{J} = \emptyset$, this means that $f$ is constant. For $t\in [d] \cup \{0\}$, we say that $f$ \emph{depends only on $t$ variables} if there exists $\mathcal{J} \subseteq [d]$ with $|\mathcal{J}| = t$ such that $f$ depends only on the coordinates in $\mathcal{J}$.
% \end{defn}
It is convenient to invoke multi-index notation for partial derivatives, whereby for $\bm\alpha=(\alpha_1,\ldots,\alpha_d)^\top \in \mathbb{N}_0^d$ and an $\|\bm\alpha\|_1$-times differentiable real-valued function $f$ defined on a subset of $\mathbb{R}^d$, we set $\partial^{\bm\alpha}f \coloneqq \partial^{\alpha_1}\cdots\partial^{\alpha_d}f$.  We also write $\bm{x}^{\bm{\alpha}} := \prod_{j=1}^d x_j^{\alpha_j}$ for $\bm{x} = (x_1,\ldots,x_d)^\top \in \mathbb{R}^d$.
\begin{defn}
    For $\beta,\gamma>0$, $d\in\mathbb{N}$, $t\in[d] \cup \{0\}$ and $D\subseteq \mathbb{R}^d$, we write $\beta_0 \coloneqq \lceil \beta \rceil -1$ and define the class of $(\beta,\gamma)$-H\"older functions that depend only on $t$ variables\footnote{For $D \subseteq \mathbb{R}^d$ and $f:D \to \mathbb{R}$, we say that $f$ \emph{depends only on $t$ variables} if there exists $\mathcal{J} \subseteq [d]$ with $|\mathcal{J}| = t$ such that for all $\bm{x}, \bm{y} \in \mathbb{R}^d$ with $\bm{x}_{\mathcal{J}} = \bm{y}_{\mathcal{J}}$, we have $f(\bm{x}) = f(\bm{y})$, where $\bm x_{\mathcal{J}}, \bm{y}_{\mathcal{J}} \in \mathbb{R}^t$ denote the subvectors of $\bm x$ and $\bm y$ respectively with indices in $\mathcal{J}$; when $\mathcal{J} = \emptyset$, this means that $f$ is constant.} as
    \begin{align*}
        \mathcal{H}_t^{\beta}(D,\gamma) &\coloneqq \biggl\{ f:D\to\mathbb{R} : \text{$f$ depends only on $t$ variables, $f$ is $\beta_0$-times differentiable,}\\
        &\hspace{1cm}\max_{\bm \alpha\in\mathbb{N}_0^d \,:\, \|\bm\alpha\|_1 \leq \beta_0} \|\partial^{\bm\alpha} f\|_{\infty} \leq \gamma,\, \max_{\bm\alpha\in\mathbb{N}_0^d \,:\, \|\bm\alpha\|_1 = \beta_0}\; \sup_{\bm x\neq \bm y \in D} \frac{|\partial^{\bm\alpha} f(\bm x) - \partial^{\bm\alpha} f(\bm y)|}{\|\bm x-\bm y\|_{2}^{\beta-\beta_0}} \leq \gamma \biggr\}.
    \end{align*}
    We also write $\mathcal{H}_t^\beta(D,\mathbb{R}^{d'},\gamma)$ to denote the set of vector-valued functions $\bm g:D \rightarrow \mathbb{R}^{d'}$ all of whose component functions belong to $\mathcal{H}_t^\beta(D,\gamma)$.
\end{defn}
Our functions $f^{\mathcal{S}_k}$ will be assumed to be compositions of vector-valued functions whose corresponding component functions belong to these H\"older classes.  Definition~\ref{defn:holder-composition-class} below generalises the standard notion of H\"older smoothness to compositional classes of functions; such definitions have been widely employed in the theoretical analysis of deep learning \citep[e.g.][]{schmidt-hieber2020nonparametric,kohler2021rate,fan2024factor}.
\begin{defn} \label{defn:holder-composition-class}
    Let $q\in\mathbb{N}$, $\bm d = (d_r)_{r=1}^{q+1} \in \mathbb{N}^{q+1}$ with $d_1=d$ and $d_{q+1} = 1$, $\bm t = (t_r)_{r=1}^q \in \mathbb{N}_0^q$ with $t_r \in [d_r] \cup \{0\} \;\forall r\in[q]$, $\bm\beta = (\beta_r)_{r=1}^q \in (0,\infty)^q$ and $\bm\gamma = (\gamma_r)_{r=1}^q \in (0,\infty)^q$. We define $\mathcal{H}_{\mathrm{comp}}(q,\bm d, \bm t, \bm\beta, \bm\gamma)$ to be the class of functions $f:\mathbb{R}^d \to \mathbb{R}$ of the form
    \begin{align*}
        f = \bm g_q \circ \cdots \circ \bm g_1 
    \end{align*} 
    where\footnote{Although $\bm g_{q}$ is a real-valued function, we write it with a bold letter for convenience.} $\bm g_{r} \in \mathcal{H}_{t_r}^{\beta_r}(\mathbb{R}^{d_r}, \mathbb{R}^{d_{r+1}}, \gamma_r)$ for all $r \in [q]$.  
    % \begin{gather*}
    %     \bm g_r = (g_{r,1}, \ldots, g_{r,d_{r+1}})^\top : \mathbb{R}^{d_r} \to \mathbb{R}^{d_{r+1}} \text{ and}\\
    %     g_{r,j} \in \mathcal{H}_{t_r}^{\beta_r} \bigl(\mathbb{R}^{d_r}, \mathbb{R}^{d_{r+1}}, \gamma_r\bigr) \text{ for all } r\in[q],\, j\in[d_{r+1}].
    % \end{gather*}
    % \begin{align*}
    %     \mathcal{H}_{\mathrm{comp}}(q,\bm d, \bm t, \bm\beta, \bm\gamma) := \{\text{$\bm g_q \circ \cdots \circ \bm g_1 : \bm g_r = (g_{r,1}, \ldots, g_{r,d_{r+1}})^\top \in \mathbb{R}^{d_r} \to \mathbb{R}^{d_{r+1}}$ such that} \\
    %     \text{$g_{r,j}\in \mathcal{H}_{t_r}^{\beta_r} \bigl(\mathbb{R}^{d_r}, \gamma_r\bigr)$ for all  $r\in[q],\, j\in[d_{r+1}]$}\}
    % \end{align*}
\end{defn}
As an example, if $f$ has a single index structure, i.e. $f(x) = g_2(a^\top x)$, for some $a \in \mathbb{R}^d$ and $g_2\in\mathcal{H}_1^{\beta}(\mathbb{R},\gamma)$, then $f$ belongs to the compositional H\"older class with $g_1:x\mapsto a^\top x$ being a linear function (which is infinitely smooth) depending on all $d$ variables, and~$g_2$ being a $\beta$-H\"older function on $\mathbb{R}$. 

% then we may take $\bm d = \bm t = (d,1)^\top$, $\bm \beta = (\infty, \beta)$ and $\bm \gamma = ()$ 
Given $f = \bm g_q \circ \cdots \circ \bm g_1 \in \mathcal{H}_{\mathrm{comp}}(q,\bm d, \bm t, \bm\beta, \bm\gamma)$ and $r \in [q]$, the H\"older smoothness of $\bm g_q \circ \cdots \circ \bm g_r$ is
\[
\bar{\beta}_r \coloneqq \beta_r \prod_{\ell=r+1}^{q} (\beta_{\ell} \wedge 1).
\]
This then leads to the following definition.
\begin{defn} \label{defn:effective-smoothness}
The \emph{critical composition index} of $f = \bm g_q \circ \cdots \circ \bm g_1 \in \mathcal{H}_{\mathrm{comp}}(q,\bm d, \bm t, \bm\beta, \bm\gamma)$ is\footnote{Here, $\mathrm{sargmax}$ denotes the smallest element of the $\argmax$ set.}
\[
r_* \coloneqq \sargmax_{r\in [q]} \frac{t_r}{\bar{\beta}_r}.
\]
The \emph{effective smoothness} and \emph{effective dimension} of $f$ are then given by $\bar{\beta}_* \coloneqq \bar{\beta}_{r_*}$ and $t_* := t_{r_*}$ respectively.  
\end{defn}
It is important to note that the effective smoothness and effective dimension arise from a worst-case ratio, as opposed to, for instance, the individual worst cases.  For instance, returning to our single index example below Definition~\ref{defn:holder-composition-class}, we have $\bar{\beta}_* = \beta$ and $t_* = 1$, even though $f$ depends on all~$d$ covariates.  We will see in Theorem~\ref{thm:PENN-ub} below that the small effective dimension induced through its compositional structure allows us to obtain a rate of convergence that avoids the curse of dimensionality.

% We will invoke two assumptions on the Bayes estimator $f^{\star}$.

\begin{assumption} \label{assumption:composition-of-smooth-functions}
    % For $k\in[K]$, let $q_k\in\mathbb{N}$, let $\bm d^{(k)} = (d_r^{(k)})_{r=1}^{q_k+1} \in \mathbb{N}^{q_k+1}$ with $d_1^{(k)}=d$ and $d_{q_k+1}^{(k)} = 1$, let $\bm t^{(k)} = (t_r^{(k)})_{r=1}^{q_k} \in \mathbb{N}_0^{q_k}$ with $t_r^{(k)} \in [d_r^{(k)}] \cup \{0\}$ for all $r\in[q_k]$, let $\bm\beta^{(k)} = (\beta_r^{(k)})_{r=1}^{q_k} \in (0,\infty)^{q_k}$ and let $\bm\gamma^{(k)} = (\gamma_r^{(k)})_{r=1}^{q_k} \in (0,\infty)^{q_k}$. 
    We assume that $f^{\mathcal{S}_k} \in \mathcal{H}_{\mathrm{comp}}(q_k,\bm d^{(k)}, \bm t^{(k)}, \bm \beta^{(k)}, \bm\gamma^{(k)})$ for each $k\in[K]$.
\end{assumption}
We write $r_*^{(k)}$, $\bar{\beta}_*^{(k)}$ and $t_*^{(k)}$ for the critical composition index, effective smoothness and effective dimension of $f^{\mathcal{S}_k}$ respectively.  Finally, then, we can state our main result: 

% \begin{prop}
%     Suppose that $\mathcal{F}(L,\bm p, s)$ is such  that any partition $\mathcal{S}_1,\mathcal{S}_2$ of $\mathcal{S} \subseteq \{0,1\}^d$ is $\mathcal{F}(L,\bm p, s)$-separable, with $p_0 = d$. Then $|\mathcal{S}| \lesssim sL\log(es) + s\log(ed)$.
% \end{prop}
% \begin{proof}
%     For any fixed partition $\mathcal{S}_1,\mathcal{S}_2$ of $\mathcal{S}$, there exists $\bm f \in \mathcal{F}(L,\bm p, s)$, $\bm v_1, \bm v_2 \in \mathbb{R}^{p_{L+1}}$ and $\epsilon>0$ such that $\|\bm f(\bm\omega) - \bm v_k\|_{\infty} \leq \epsilon/2$ for all $k\in\{1,2\}$, $\bm\omega \in \mathcal{S}_k$ and $\|\bm v_1 - \bm v_2\|_{\infty} \geq 2\epsilon$. Without loss of generality, assume that $v_{1,1} < v_{2,1}$. Let $f_1$ be the first coordinate function of $\bm f$ and let $y \coloneqq \frac{v_{1,1} + v_{2,1}}{2}$.  Then
%     \begin{align*}
%         f_1(\bm\omega) - y < 0 \text{ for all }\bm\omega\in\mathcal{S}_1 \quad\text{and}\quad f_1(\bm\omega) - y > 0 \text{ for all }\bm\omega\in\mathcal{S}_2.
%     \end{align*}
%     This holds for any partition $\mathcal{S}_1,\mathcal{S}_2$ of $\mathcal{S}$, so $\mathrm{Pdim}\bigl(\mathcal{F}(L,\bm p, s)\bigr) \geq |\mathcal{S}|$. However, by Proposition~\ref{prop:VC-upper-bound}, we have $\mathrm{Pdim}\bigl(\mathcal{F}(L,\bm p, s)\bigr) \lesssim sL\log(es) + s\log(ed)$. 
% \end{proof}

\begin{thm} \label{thm:PENN-ub}
    Suppose that Assumptions~\ref{assumption:piecewise-assumption} and~\ref{assumption:composition-of-smooth-functions} hold.  Suppose further that $\{\mathcal{S}_1,\ldots,\mathcal{S}_K\}$ is $\mathcal{F}(L_2,\bm p_2,s_2)$-separable for some $L_2\in\mathbb{N}$, $\bm p_2 \in \mathbb{N}^{L_2+2}$ and $s_2\in\mathbb{N}$, and let $B_n \coloneqq \xi_2\sqrt{2\log n}$.  Then there exist $L_1,L_3 \in \mathbb{N}$, $\bm p_1 \in \mathbb{N}^{L_1+2}$, $\bm p_3 \in \mathbb{N}^{L_3+2}$ and $s \in \mathbb{N}$ such that, writing 
    \begin{align*}
        \mathcal{F} \coloneqq \mathcal{F}_{\mathrm{PENN}}\bigl((L_r,\bm{p}_r)_{r=1}^3,s\bigr),
    \end{align*}
    and letting $\hat{f}$ denote any neural network estimator in $\mathcal{F}$ based on data $\mathcal{D} \coloneqq (\bm Z_i,\bm\Omega_i,Y_i)_{i=1}^n$, we have for $n \geq 2$ that
    \begin{align*}
        \mathbb{E}\bigl\{ R(T_{B_n}\hat{f})\bigr\} - R(f^{\star}) &\leq C \biggl\{\sum_{k=1}^K \pi_k n_k^{-2\bar\beta^{(k)}_*/(2\bar\beta^{(k)}_* + t^{(k)}_*)} + \frac{s_2\log s_2}{n}\biggr\} \cdot (\log n)^{2\max_{k\in[K]} \bar{\beta}^{(k)}_1 \vee 6}\\
        &\hspace{7cm} + 2\mathbb{E}\Bigl\{\hat{R}_n(\hat{f}) - \inf_{f\in\mathcal{F}} \hat{R}_n(f) \Bigr\},
    \end{align*}
    where $C>0$ does not depend on $n$, $(\pi_k)_{k=1}^K$ or $s_2$. %depends only on $\xi_1,\xi_2,m,M_1,N_3$ and $(\bm d^{(k)}, \bm t^{(k)}, \bm\beta^{(k)}, \bm\gamma^{(k)})_{k=1}^K$.
\end{thm}    
%\red{Remark that different imputation algorithms may have different $f^\star$, hence different rates of convergence.}

In order to understand the main messages of Theorem~\ref{thm:PENN-ub}, first suppose that we are able to compute the empirical risk minimiser exactly, so that the optimisation error $\mathbb{E}\bigl\{\hat{R}_n(\hat{f}) - \inf_{f\in\mathcal{F}} \hat{R}_n(f) \bigr\}$ is zero.  Then the excess risk of a truncated PENN estimator is controlled by the sum of two interpretable terms.  The first of these would be the minimax rate for estimating the Bayes regression function if the partition $\{\mathcal{S}_1,\ldots,\mathcal{S}_K\}$ of Assumption~\ref{assumption:piecewise-assumption} were known, up to a poly-logarithmic factor in the sample size.  It comprises a weighted average over $k \in [K]$ of the minimax rates of estimating the Bayes regression function $f^{\mathcal{S}_k}$ on the $k$th cell of the partition, with effective sample size~$n_k$, effective smoothness~$\bar{\beta}_*^{(k)}$ and effective dimension $t_*^{(k)}$; see Theorem~\ref{thm:minimax-lb} below.  %\orange{Discuss: In this sense, our PENN estimator is able to adapt to the unknown effective smoothness and effective dimension in each cell of the partition.}  
As an attraction of this weighted average form, we see that the $k$th summand
\[
\pi_k n_k^{-2\bar{\beta}_*^{(k)} / (2\bar{\beta}_*^{(k)} + t_*^{(k)})} = \pi_k^{t_*^{(k)} / (2\bar{\beta}_*^{(k)} + t_*^{(k)})} n^{-2\bar{\beta}_*^{(k)} / (2\bar{\beta}_*^{(k)} + t_*^{(k)})} \to 0 \quad\text{as $\pi_k \rightarrow 0$,}
\]
for fixed $n$; thus, rarely observed missingness patterns have negligible effect on the excess risk.  The second term in the bound in Theorem~\ref{thm:PENN-ub} is the additional error incurred due to the fact that the $\mathcal{F}(L_2,\bm p_2,s_2)$-separable partition $\{\mathcal{S}_1,\ldots,\mathcal{S}_K\}$ is unknown.  For example, in cases where the partition is defined by a small number of coordinates or a small number of halfspaces, we can expect this second term to be dominated by the first.  In particular, this occurs in the setting of Proposition~\ref{prop:F(L,p,s)-separable}\emph{(a)} with $|\mathcal{S}_{\mathcal{J}}| \lesssim \sum_{k=1}^K n_k^{t^{(k)}_*/(2\bar\beta^{(k)}_* + t^{(k)}_*)}$ and in the setting of Proposition~\ref{prop:F(L,p,s)-separable}\emph{(b)} with $K\sum_{k=1}^K P_k \lesssim \sum_{k=1}^K n_k^{t^{(k)}_*/(2\bar\beta^{(k)}_* + t^{(k)}_*)}$.   Theorem~\ref{thm:PENN-ub} also accounts for cases where the empirical risk minimiser is not computed exactly, in which case we incur an additional optimisation error.  

We next characterise the order in $n_1,\ldots,n_K$ of the parameters $L_1,L_3,\bm p_1, \bm p_3$ and~$s$ needed for Theorem~\ref{thm:PENN-ub} to hold.  Writing $\bm p_1 = (d, p_{1,*}, \ldots, p_{1,*}, p_{1,L_1+1})^\top \in \mathbb{N}^{L_1+2}$ and $\bm p_3 = (d, p_{3,*}, \ldots, p_{3,*}, 1)^\top \in \mathbb{N}^{L_3+2}$, we can see from the proof of Theorem~\ref{thm:PENN-ub} that it suffices to choose $L_1$ and $L_3$ of constant order in the sample size, together with 
\begin{gather*}
    p_{1,*} \asymp \sum_{k=1}^K n_k^{t^{(k)}_* / (4\bar{\beta}^{(k)}_* + 2t^{(k)}_*)}\log n, \quad
    p_{3,*} \asymp \sum_{k=1}^K n_k^{t^{(k)}_* / \{(L_3-3)(2\bar{\beta}^{(k)}_* + t^{(k)}_*)\}}, \\
    s \asymp s_2 + \sum_{k=1}^K n_k^{t^{(k)}_* / (2\bar{\beta}^{(k)}_* + t^{(k)}_*)}\log^2 n. 
\end{gather*}
Moreover, by~\ref{NN-enlarging} in Appendix~\ref{appendix-a} and the fact that the bound in Proposition~\ref{prop:oracle-inequality} does not depend on the widths of the hidden layers, we may increase $p_{1,*}$ and $p_{3,*}$ without affecting the bound in Theorem~\ref{thm:PENN-ub}.  This means that Theorem~\ref{thm:PENN-ub} applies to over-parametrised neural networks of appropriate sparsity. 

From the proof, we can see that the quantity $C > 0$ in Theorem~\ref{thm:PENN-ub} may be chosen as the maximum over $k\in[K]$ of polynomial functions of $d, \bm d^{(k)}, \bm t^{(k)}$, whose degrees depend only on $\bm\beta^{(k)}$ and whose coefficients depend only on $p_{2,L_2+1}, L_1, L_2, L_3, \xi_1, \xi_2, d, \bm\beta^{(k)}, \bm\gamma^{(k)}$.  In particular, if $\pi_k = 1/K$, $\beta^{(k)}_r = \beta \geq 1$ and $t^{(k)}_r = t$ for all $k\in[K]$ and $r\in[q_k]$, then the first term in the upper bound simplifies to $(K/n)^{2\beta/(2\beta+t)}$, up to a poly-logarithmic factor in $n$ that does not depend on $K$.   A final observation from the proof is that Assumption~\ref{assumption:piecewise-assumption} can be weakened so that the Bayes regression function is not required to be identical on different elements within the same $\mathcal{S}_k$; in fact, it suffices to assume that $|f^{\star}(\bm z,\bm\omega) - f^{\mathcal{S}_k}(\bm z)| \lesssim n_k^{-\bar{\beta}_*^{(k)} / (2\bar{\beta}_*^{(k)} + t_*^{(k)})}$ for $\bm z\in\mathbb{R}^d$, $k\in[K]$ and $\bm\omega \in \mathcal{S}_k$.

% Assuming that the imputation algorithm \textsf{Imp} does not change the observed coordinates, $R(f^{\star})$ does not depend on the imputation algorithm. However, different imputation algorithms may result in different effective smoothness and effective dimension in $f^{\star}$. For example, suppose $Y=X_1+X_2$, $X_2=g(X_1)$ and only the second coordinate of $X$ may be missing. If we impute missing values by $g(X_1)$, then $Z=(X_1,g(X_1))$ and we have $f^{\star}(Z,\Omega) = Z_1+Z_2 = Z_1+g(Z_1)$, but the first version of $f^{\star}$ (which is a linear function) is more regular in general. However, if we impute the missing values by zero, then the most regular $f^{\star}$ will not be linear in general. Nevertheless, $Z_1+g(Z_1)$ is always a version of $f^{\star}$, so $R(f^{\star})$ will not depend on the imputation algorithm.

It is worth highlighting the novel aspects of Theorem~\ref{thm:PENN-ub} relative to prior work with fully observed data \citep[e.g.][]{schmidt-hieber2020nonparametric,kohler2021rate}.  One key challenge for us is to formulate an appropriate model to allow us to borrow strength across different missingness patterns.  This is essential in view of the fact that there can be up to $2^d$ different missingness patterns in total, and is achieved through our Assumption~\ref{assumption:piecewise-assumption}.  Likewise, the $\mathcal{F}$-separability condition in Definition~\ref{defn:F(L,p,s)-separable} that captures the complexity of the cell partition of the Bayes regression function is new.  The main challenge in the proof of Theorem~\ref{thm:PENN-ub} is to achieve the delicate balance of the  estimation and approximation errors needed to obtain the $\sum_{k=1}^K \pi_k n_k^{-2\bar{\beta}_*^{(k)}/(2\bar{\beta}_*^{(k)} + t_*^{(k)})}$ term in the bound; this term does not arise when no missingness is present.

We complement the upper bound of Theorem~\ref{thm:PENN-ub} with a corresponding minimax lower bound, that uses the notation of Assumptions~\ref{assumption:piecewise-assumption} and~\ref{assumption:composition-of-smooth-functions}, as well as Definition~\ref{defn:effective-smoothness}. Further, for $k\in[K]$, define $\mathcal{J}^{(k)} \coloneqq \{j\in[d] : \omega_j = 1 \text{ for all } \bm\omega \in \mathcal{S}_k\}$, as well as $d_*^{(k)} \coloneqq d_{r_*^{(k)}}^{(k)}$ and $\gamma_*^{(k)} \coloneqq \gamma^{(k)}_{r_*^{(k)}}$. 

\begin{thm} \label{thm:minimax-lb} 
    Let $\mathcal{P}$ be the set of all distributions of $(\bm Z_0,\bm\Omega_0,Y_0)$ where $\bm Z_0 = \mathsf{Imp}(\tilde{\bm X}_0)$, where $\tilde{\bm{X}}_0$ denotes the partially observed covariate vector and where $(\bm X_0, \bm \Omega_0, Y_0)$ satisfies Assumptions~\ref{assumption:piecewise-assumption} and~\ref{assumption:composition-of-smooth-functions} with $\xi_1,\xi_2\geq 1$.  Suppose further that for some $j_*\in[d]$, we have 
    \[  
    t_*^{(k)} \leq \min\bigl\{d_1^{(k)},\ldots,d_*^{(k)}\bigr\} \wedge \bigl|\mathcal{J}^{(k)} \setminus \{j_*\}\bigr| \wedge \gamma_*^{(k)}
    \]
    for all $k \in [K]$.  For $P\in\mathcal{P}$, let $f^{\star} \equiv f^{\star}_P \coloneqq \mathbb{E}_P(Y_0 \,|\, \bm Z_0,\bm\Omega_0)$, and let $\hat{\mathcal{F}}$ be the set of all estimators of $f^\star$ based on a sample of size $n$, i.e.~the set of Borel measurable functions from $\mathbb{R}^d \times \mathcal{S} \times (\mathbb{R}^d \times \mathcal{S} \times \mathbb{R})^n$ to $\mathbb{R}$. Then there exists $c>0$, depending only on $\xi_2$ and $(\bar{\beta}_*^{(k)}, t_*^{(k)})_{k=1}^K$, such that
    \begin{align*}
        \inf_{\hat{f} \in \hat{\mathcal{F}}} \sup_{P\in \mathcal{P}} \mathbb{E}_{P^{\otimes n}} \bigl\{R(\hat{f}) - R(f^{\star})\bigr\} \geq c \sum_{k=1}^K \pi_k n_k^{-2\bar{\beta}_*^{(k)} / (2\bar{\beta}_*^{(k)} + t_*^{(k)})}
    \end{align*}
    for all $n \in \mathbb{N}$.
\end{thm}
Theorem~\ref{thm:minimax-lb} reveals that the main term in the upper bound in Theorem~\ref{thm:PENN-ub}, namely the first estimation error term, is minimax optimal in $n$ and $\pi_1,\ldots,\pi_K$, up to a poly-logarithmic factor in $n$.  

\section{Simulations} \label{sec:simulations}

In this section, we study the empirical performance of Pattern Embedded Neural Network (PENN) estimators on simulated, semi-synthetic and real data.  Since PENN estimators can be used in conjunction with any imputation technique, we consider (columnwise) mean imputation (MI), MissForest imputation (MF) \citep{stekhoven2012missforest} and a python implementation of Multiple Imputation by Chained Equations (MICE) \citep{van2011mice} called  \texttt{IterativeImputer} (II) from \texttt{scikit-learn} \citep{pedregosa2011scikit}.  In each case, we compare PENN with standard neural networks that do not incorporate the revelation vectors as covariates. We also compare with XGBoost \citep{chen2016xgboost} and random forests \citep{breiman2001random}, which are two popular tree-based algorithms that handle missing values directly by learning default split directions for observations with missingness.  % Previous studies have argued that XGBoost provides a strong baseline for tabular data, at least for complete data \citep{grinsztajn2022tree}.

For all of our numerical experiments, we divide the data into training, validation and test sets.  We fit the neural networks on the training data by first running the stochastic optimisation method AdamW \citep{loshchilov2018decoupled} for 10 epochs, and then keeping a proportion $\lambda \in \{0.1,0.2,0.4,0.8\}$ of weights of largest magnitude (in addition to all of the bias vectors) to obtain a sparse network.  Following the recommendation of \citet{liu2018rethinking}, we then randomly reinitialise the non-zero parameters after this pruning step using the \texttt{PyTorch} function \texttt{torch.nn.init.kaiming\_uniform\_} and retrain the sparse model via AdamW.  Early stopping \citep{prechelt2002early} is incorporated, so that the training process for each value of $\lambda$ is terminated once the validation loss fails to decrease by at least $0{.}001$ in 10 epochs.  The tuning parameter $\lambda$ is then chosen to minimise the average loss on the validation set, and finally the performance of this selected estimator on the test set is reported.  For XGBoost, we choose the \texttt{max\_depth} parameter from $\{3,6,9,12\}$ based on its performance on the validation set, while for random forests, we tune its \texttt{max\_features} parameter from $\{\lceil0.1d\rceil, \lceil0.2d\rceil, \lceil0.4d\rceil, \lceil0.8d\rceil\}$; other tuning parameters for these methods were chosen using default settings.  We remark that although our theory requires the data used for imputation to be independent of the training data, in our simulations we train the imputation algorithms on the whole dataset (training, validation and test sets) and then impute the missing entries.  

In our simulations, for a dataset with (training) sample size $n$ and covariate dimension~$d$, we define $w_1 \coloneqq \lceil \sqrt{n} \rceil$, $w_2 \coloneqq \lceil w_1/5 \rceil$ and embedding dimension $m\coloneqq \lceil \sqrt{d} \rceil \wedge \lceil w_1/20 \rceil$. Then
\begin{itemize}
    \item PENN uses $\mathcal{F}_{\mathrm{PENN}}\bigl((L_r,\bm p_r)_{r=1}^3, s\bigr)$, where $L_1 = 3$, $\bm{p}_1 = (d, w_1, w_1, w_1, w_1)$, $L_2 = 2$, $\bm{p}_2 = (d,w_2,w_2,m)$ and $L_3=3$, $\bm{p}_3=(w_1+m,w_1,w_1,w_1,1)$;
    \item NN uses $\mathcal{F}\bigl(6, (d,w_1,w_1,w_1,w_1,w_1,w_1,1),s\bigr)$.
\end{itemize}
By \ref{NN-composition}, the PENN class above has six hidden layers in total, not including the embedding function $\bm{f}_2$; see~\eqref{eq:PENN-def} and Figure~\ref{fig:PENN-diagram}.  We therefore compare it with a standard neural network with six hidden layers of the same width. The sparsity $s$ is chosen as described above.

\subsection{Simulated data} \label{sec:simulated-data}

For each of our experiments on simulated data, we take $d=20$, with a training set of size $n=10{,}000$, meaning that $m = 4$, as well as validation and test sets each of size $5{,}000$.  Our data generating mechanisms were chosen as follows:

\paragraph{Model 1:} $\bm X_0 \sim \mathrm{Unif}[-1,1]^d$, $Y_0 = \exp(X_{0,1} + X_{0,2}) + 4X_{0,3}^2 + \varepsilon_0$, $\varepsilon_0\sim N(0,0.25)$, $\varepsilon_0 \indep \bm{X}_0$, $\bm\Omega_0 \indep (\bm X_0, \varepsilon_0)$ and $\Omega_{0,j} \overset{\mathrm{iid}}{\sim} \mathrm{Ber}(0.7)$ for $j \in [d]$.

\paragraph{Model 2:} $\bm X_0 \sim \mathrm{Unif}[-1,1]^d$, $Y_0 = 2\sin(2X_{0,1} + 2X_{0,2} + 2X_{0,3}) + 2X_{0,3} + \varepsilon_0$, $\varepsilon_0\sim N(0,0.25)$, $\varepsilon_0 \indep \bm{X}_0$, $\Omega_{0,j} \overset{\mathrm{iid}}{\sim} \mathrm{Ber}(0.7)$ independent of $\bm X_0$ for $j\notin\{2,3\}$, and $\Omega_{0,j} = \mathbbm{1}_{\{X_{0,j} \leq 0.4\}}$ for $j\in\{2,3\}$.

\paragraph{Model 3:} As for Model 1 except that we set $X_{0,1} = \sqrt{X_{0,4}+1} - 0.7 + \mathrm{Unif}[-0.3,0.3]$ and $X_{0,3} = 0.7X_{0,5} + \mathrm{Unif}[-0.3,0.3]$.

\paragraph{Model 4:} As for Model 2 except that we set $X_{0,1} = \sqrt{X_{0,4}+1} - 0.7 + \mathrm{Unif}[-0.3,0.3]$ and $X_{0,3} = 0.7X_{0,5} + \mathrm{Unif}[-0.3,0.3]$.

\medskip

Thus, Models~1 and~3 are MCAR, while Models~2 and~4 are MNAR.  In Models~1 and~2, the different coordinates of the covariates are independent, while in Models~3 and~4 we introduce positive correlation between $X_{0,1}$ and $X_{0,4}$, and between $X_{0,3}$ and $X_{0,5}$.  In each case, the true regression function depends only on the first three components of $\bm X_0$; this facilitates computation of the Bayes risk via Monte Carlo integration.  

The results of our simulations over 100 repetitions are presented in Figure~\ref{fig:synthetic-data}.  We observe that for all four models, PENN improves the performance of vanilla NNs for each imputation technique, often dramatically.  In particular, it is able to substantially remedy the strikingly poor performance of \texttt{NN\_MF} for Models~2 and~4.  It seems that, to some extent, more successful imputation techniques for the vanilla neural network estimator tend to yield more successful PENN estimators.  The XGBoost and random forest methods are competitive for Model~1, but less so for the other three models.  In Appendix~\ref{sec:embeddings-syn}, we also provide visualisations of the learned pattern embeddings.

\begin{figure}[htbp]
    \centering
    \begin{subfigure}{0.49\textwidth}
    \includegraphics[width=\textwidth]{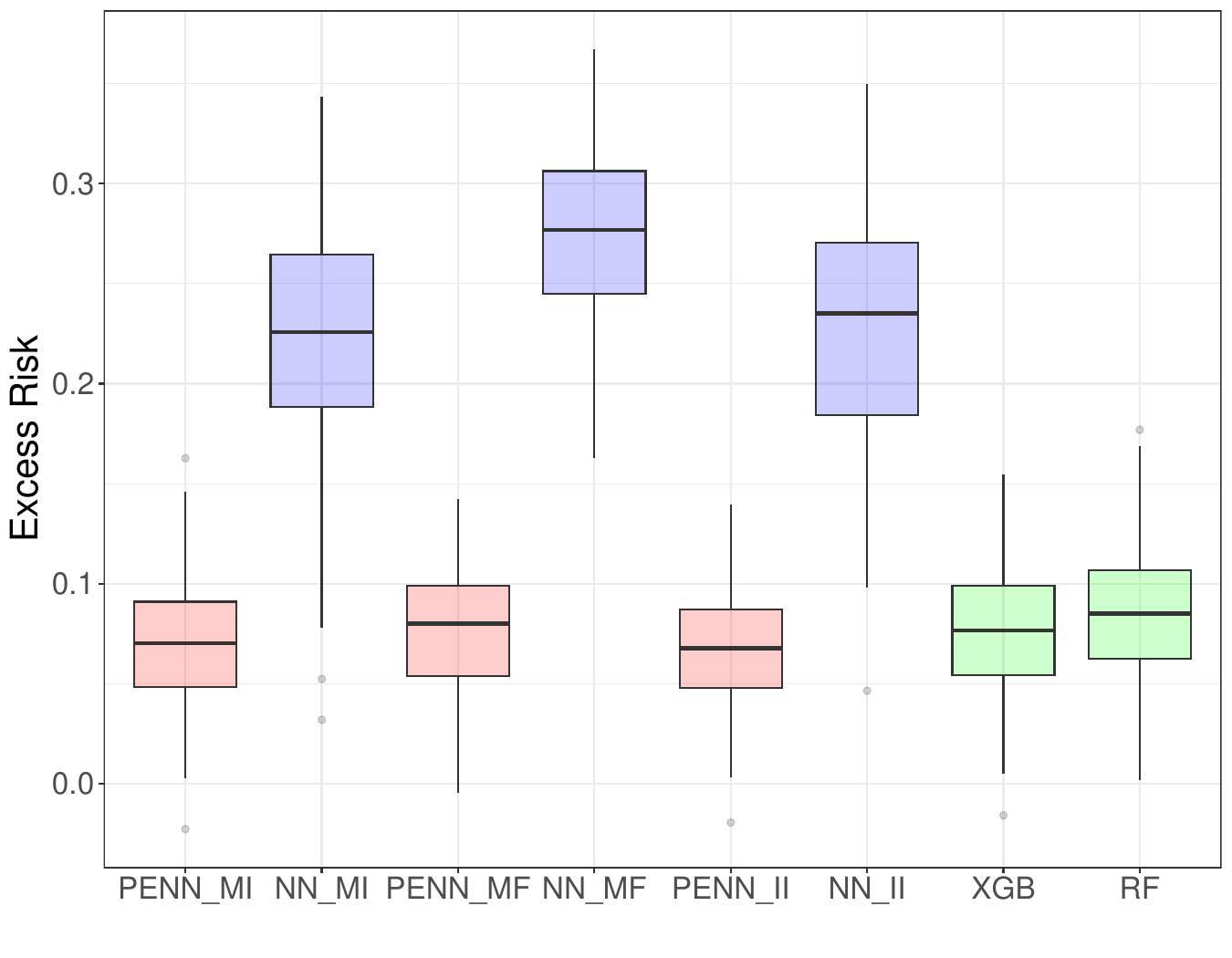} 
    \caption{Model 1}
    \end{subfigure}
    \begin{subfigure}{0.49\textwidth}
    \includegraphics[width=\textwidth]{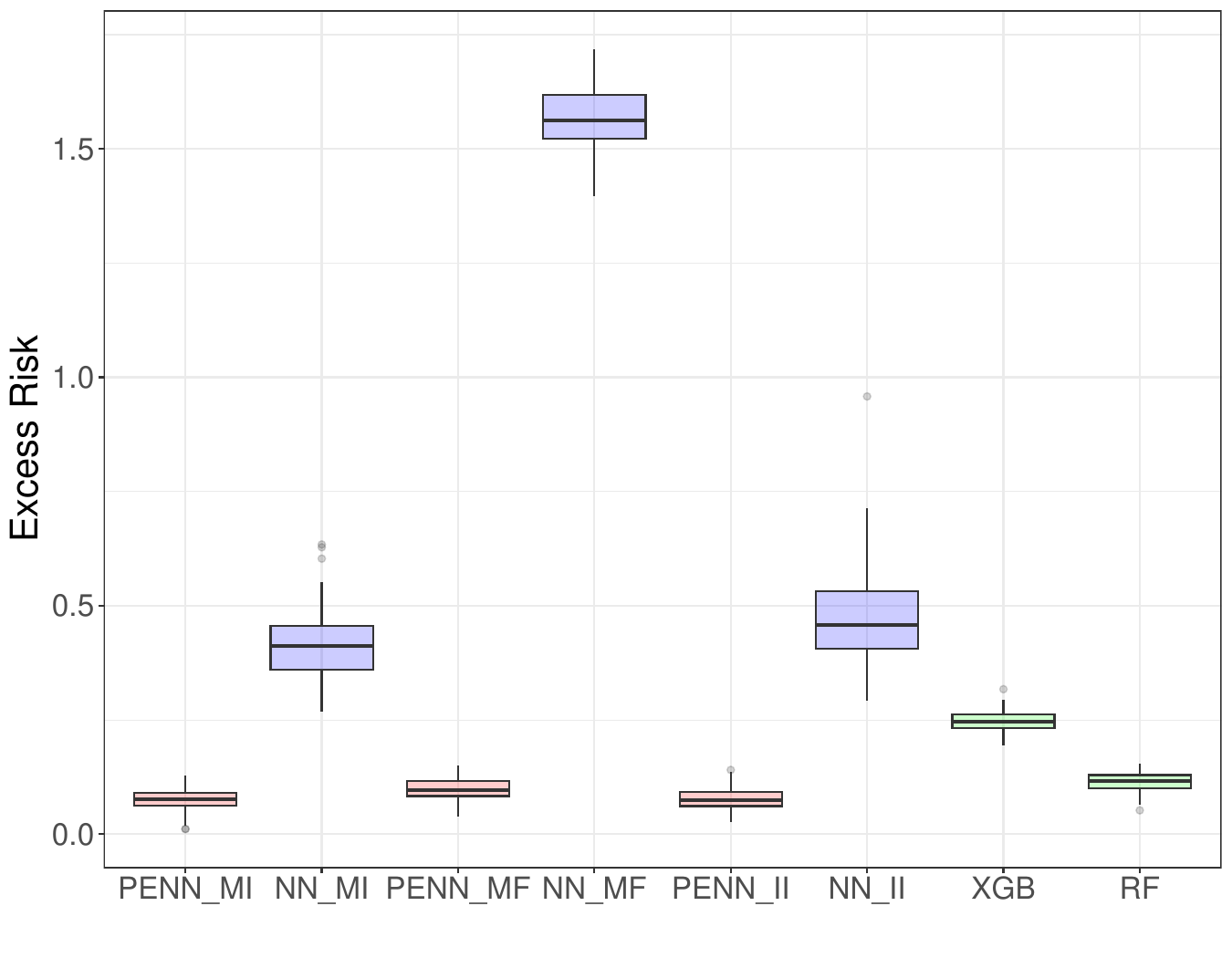}
    \caption{Model 2}
    \end{subfigure} \\
    \vspace*{0.75cm}
    \begin{subfigure}{0.49\textwidth}
    \includegraphics[width=\textwidth]{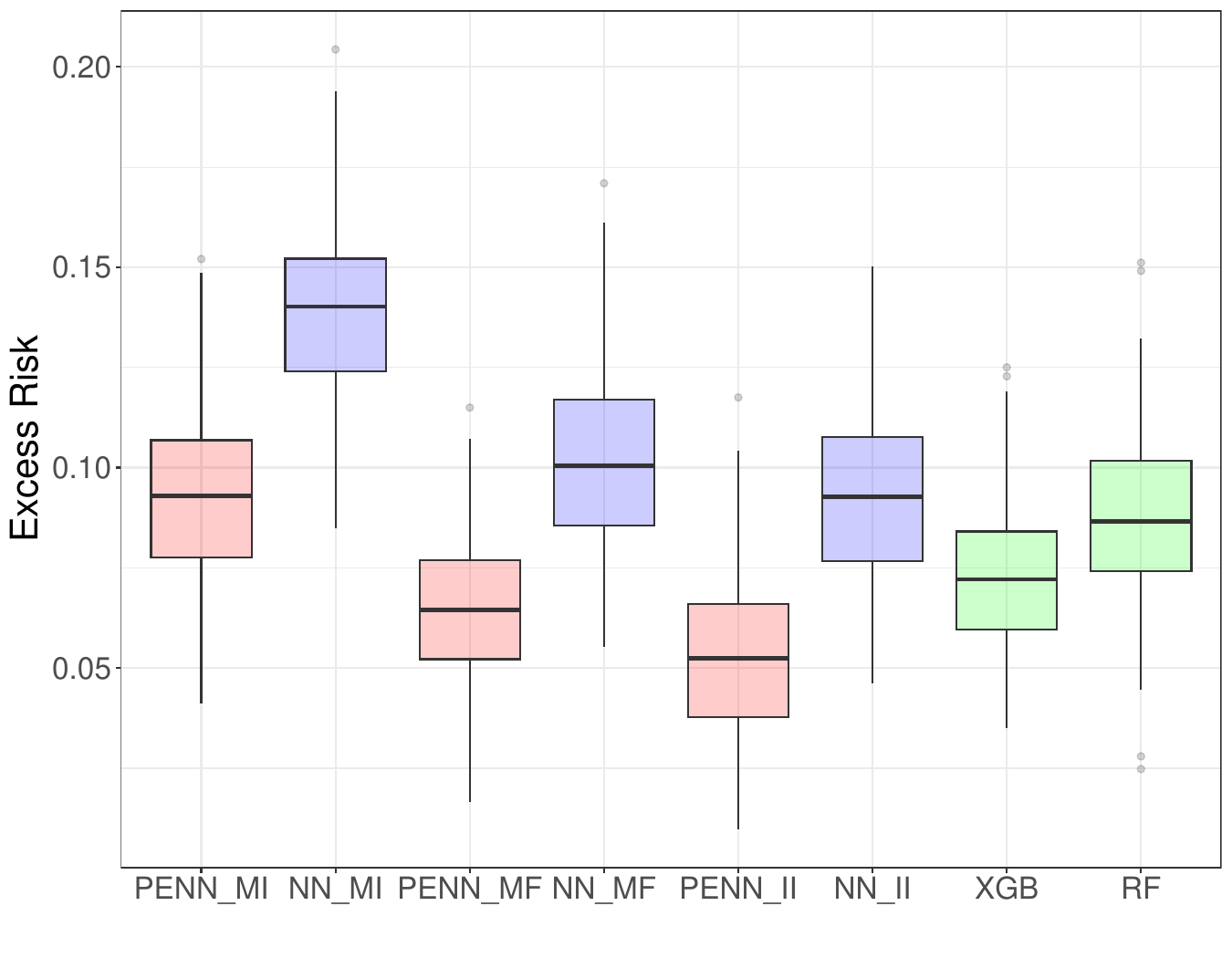}
    \caption{Model 3}
    \end{subfigure}
    \begin{subfigure}{0.49\textwidth}
    \includegraphics[width=\textwidth]{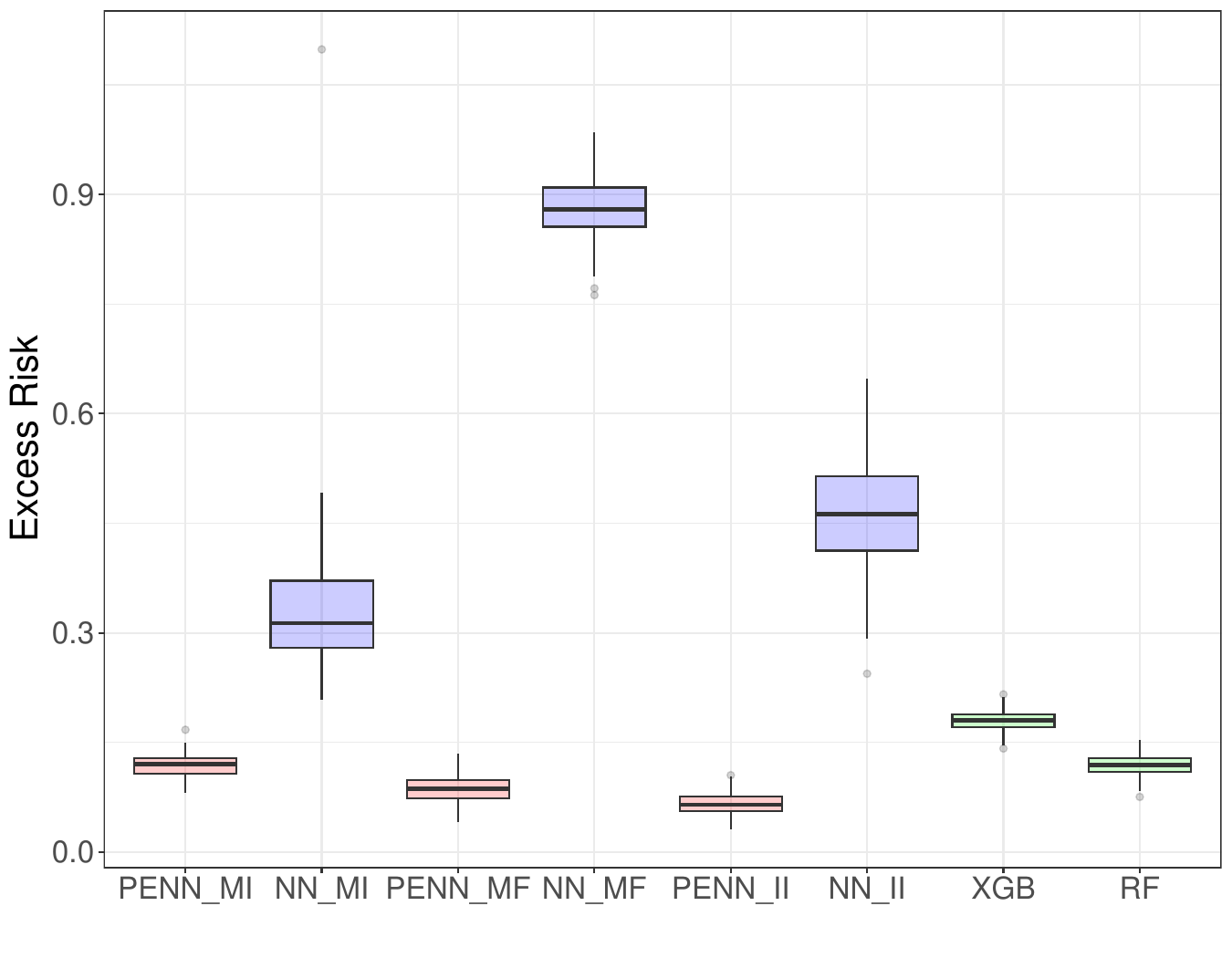}
    \caption{Model 4}
    \end{subfigure}

    \caption{Estimates of excess risks for simulated data Models 1--4.  The PENN estimators are shown in red, with the vanilla neural networks (NN) in blue and XGBoost (XGB) and random forests (RF) in green. On the $x$-axis, the abbreviation of the imputation technique appears after the underscore symbol.} \label{fig:synthetic-data}
\end{figure}

\subsection{Semi-synthetic data}

The aim of this subsection is to consider two semi-synthetic datasets.  By this we mean that we take two real datasets without missingness, and artificially introduce missingness according to two different prescribed mechanisms that we articulate below.  An attraction of this approach is that it allows us to study the effects of different types of missingness on real data.

\paragraph{Relative location of CT dataset:} The relative location of Computed Tomography (CT) dataset from \url{https://www.openml.org/search?type=data&status=active&id=46300} is a complete dataset with $d=384$ and total sample size $n=53{,}500$.  The response variable is the relative location of the CT slice on the axial axis of the human body, while the covariates describe the location of bone structures and the location of air inclusions.  Here we again consider MCAR missingness with homogeneous observation probability 0.7 in each coordinate.  For MNAR missingness, the probability of observing $X_{ij}$ is given by $\mathrm{sigmoid}(-2X_{ij} + 4\bar{X}_i + 1)$, where $\mathrm{sigmoid}(x) \coloneqq 1/(1+e^{-x})$ and $\bar{X}_i \coloneqq \frac{1}{d} \sum_{k=1}^d X_{ik}$.

\paragraph{MNIST dataset:} The MNIST dataset \citep{lecun1998mnist} in \texttt{PyTorch} consists of $70{,}000$ greyscale images of handwritten digits from 0 to 9, with each image represented as a $28 \times 28$ pixel grid. We normalise the greyscale so that it takes values in $[0,1]$, and the goal is to predict the labels based on the images.  To reflect the fact that missingness is likely to be correlated among nearby pixels, we partition each image uniformly into 49 blocks, each of $4\times 4$ pixels.  In our MCAR setting, each block is observed independently with probability $0.5$; in our MNAR setting, for each $\ell\in\{0,\ldots,9\}$, we draw a random vector $\bm p_\ell = (p_{\ell,1},\ldots,p_{\ell,49})^\top$ uniformly from $[0.2,0.8]^{49}$, and generate $\Omega_{i,j} \sim \mathrm{Ber}(p_{\ell_i,j})$ independently for $j \in [49]$, where $\ell_i$ is the label for the $i$th observation. In other words, in our MNAR setting, the missingness probability varies across different labels, while the overall expected missingness probability is 0.5 for all observations.  For the MNIST dataset, we adopt only mean imputation, since MissForest and MICE are not well-suited for image data.  Figure~\ref{fig:mnist-numbers} presents the first 16 images from the MNIST dataset with MCAR and MNAR missingness.
% In our MCAR setting, each block is observed independently with probability $0.7$; in our MNAR setting, for each block, we first compute its average grayscale $x$; the block is then observed with probability $\frac{1}{e^{10x-4}+1}$. For the MNIST dataset, we use zero imputation (ZI) so that if a block is missing, then it is set to be black. Figure~\ref{fig:mnist-numbers} presents the first 16 images from the MNIST dataset with MCAR and MNAR missingness.

\begin{figure}[htbp]
\centering
\begin{subfigure}{0.4\textwidth}
    \includegraphics[width=\textwidth]{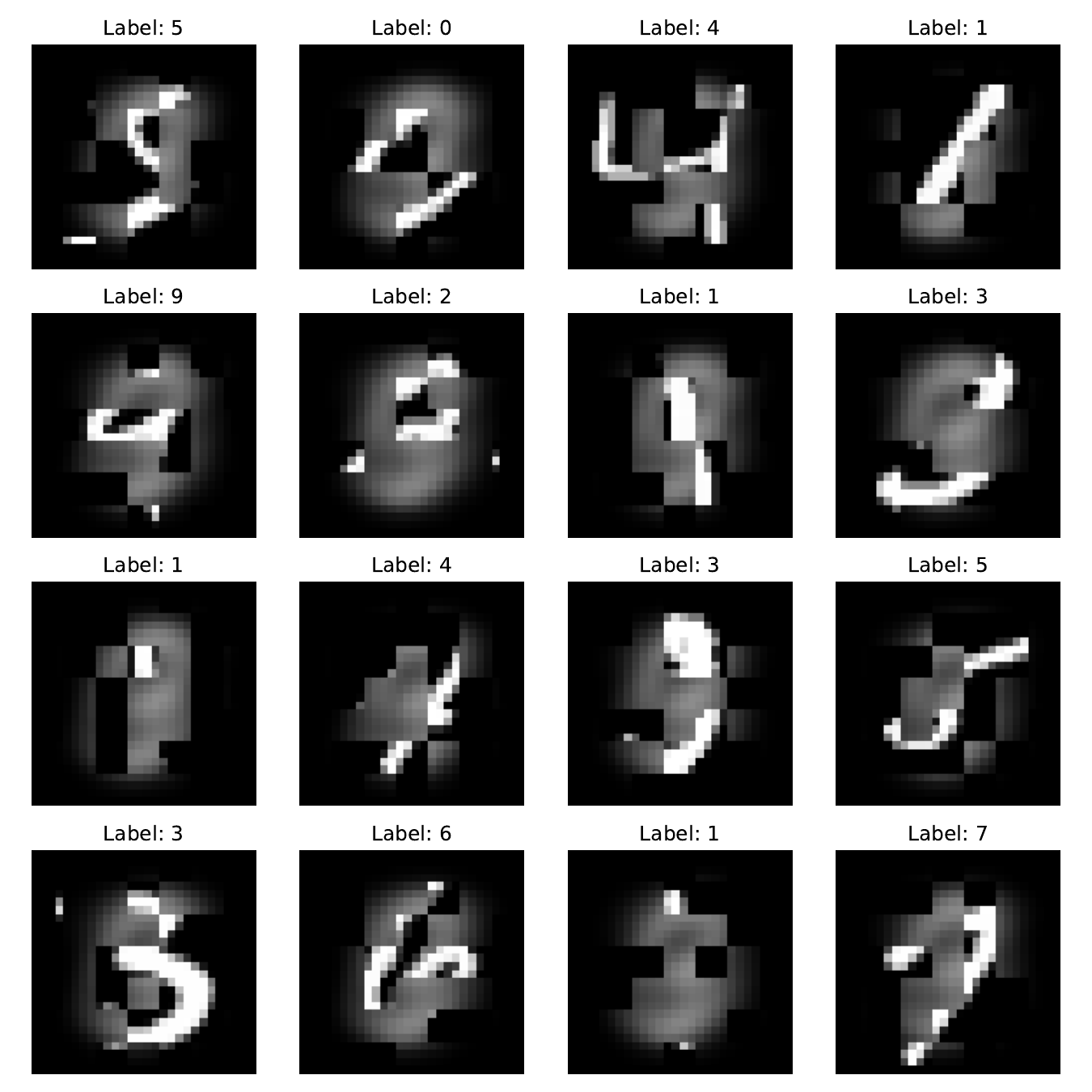}
    \caption{MCAR missingness}
\end{subfigure}
\begin{subfigure}{0.4\textwidth}
    \includegraphics[width=\textwidth]{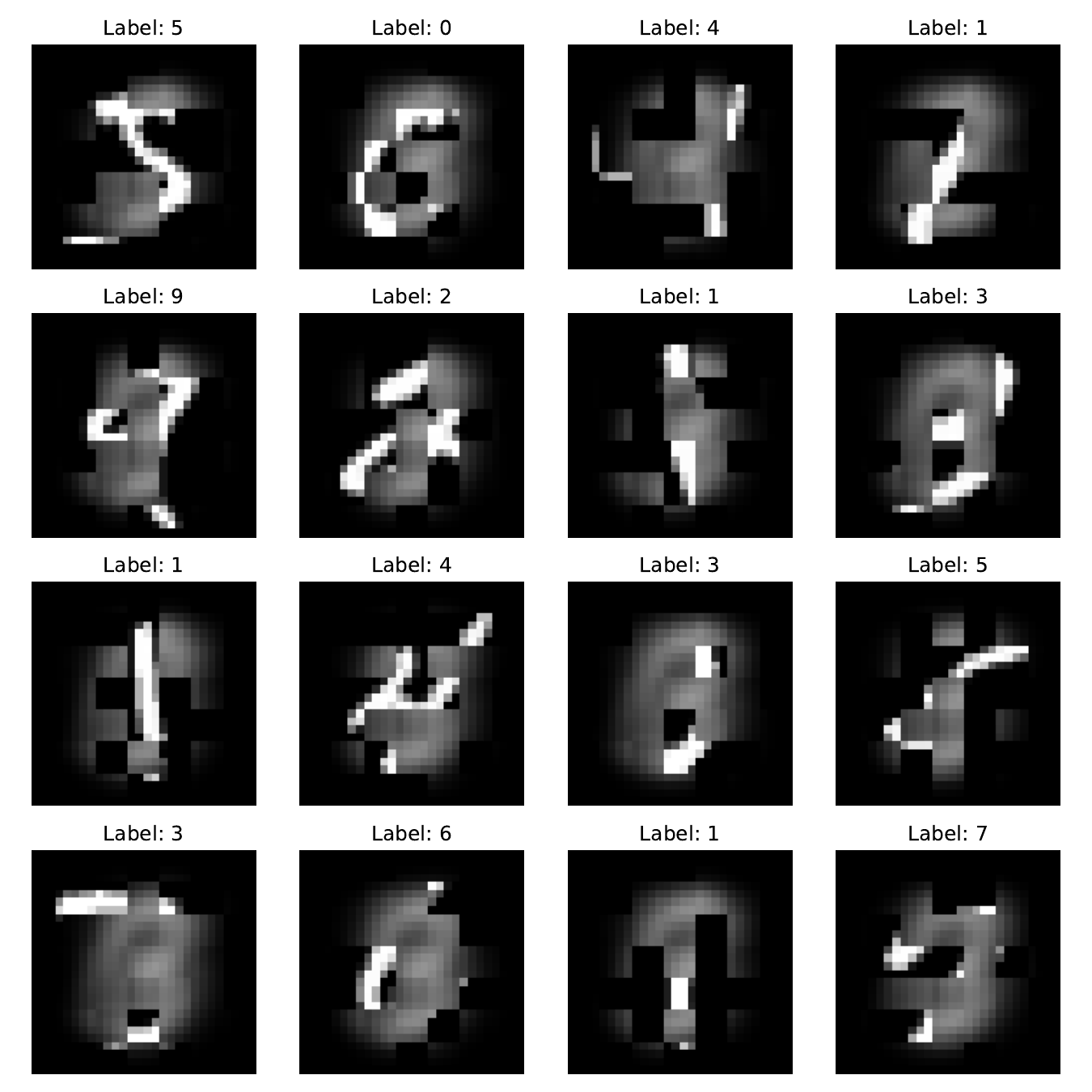}
    \caption{MNAR missingness}
\end{subfigure}
\caption{First 16 images (after mean imputation) from the MNIST dataset with MCAR and MNAR missingness, together with the true labels above each panel.} 
\label{fig:mnist-numbers}
\end{figure}

\medskip
In image classification tasks, convolutional neural networks \citep{zhang1988shift,lecun1998gradient, krizhevsky2012imagenet} have achieved great successes; see also \citet[][Section~3.1]{fan2020selective} for an introduction.
Thus, for the MNIST dataset, we replace each of the first two hidden layers in $\bm f_1$ in the definition of PENN~\eqref{eq:PENN-def}, and each of the first two hidden layers in NN, by a composition of a 2D convolutional layer, ReLU activation layer and max pooling layer.  Details of the architectures of PENN and NN with convolutional layers used in our simulations can be found at \url{https://github.com/tianyima2000/DNN_missing_data}. Again, we compare a PENN with a neural network of the same architecture except that the embedding function has been removed, as well as with XGBoost and random forests.  In addition, we include a Vision Transformer (ViT) \citep{dosovitskiy2021an} as a benchmark\footnote{We used the implementation from the GitHub repository \url{https://github.com/gejinchen/PyTorch-Vision-Transformer-ViT-MNIST-CIFAR10}.} for the MNIST dataset, since  ViT architectures can be applied directly to images with missing or occluded patches by downweighting attention weights on patches exhibiting missingness or occlusion.  Our training employs the cross-entropy loss, which is the negative log-likelihood of the relevant  multinomial distribution, for the MNIST classification task. 
%\red{Maybe remark somewhere in the introduction that our PENN framework is also compatible with more complicated architectures like CNN.}

% In these semi-synthetic examples,
% \begin{itemize}
%     \item PENN uses
%     \begin{align*}
%         \mathcal{F}_{\mathrm{PENN}} 
%         &\biggl(
%         \begin{bmatrix}
%         \bigl(3,\,(d, 100, 100, 100, 100)\bigr) & \multirow{2}{*}{$\bigl(3,\,(103,100,100,100,p_{\mathrm{out}})\bigr)$} \\
%         \bigl(2,\,(d,30,30,3)\bigr) & 
%         \end{bmatrix}, s \biggr).
%     \end{align*}
%     \item NN uses $\mathcal{F}\bigl(6, (d,100,100,100,100,100,100,p_{\mathrm{out}}),s\bigr)$,
% \end{itemize}
% where $p_{\mathrm{out}}$ is $1$ for the bank loan data, and $10$ for the MNIST data.  Thus, we again compare PENN with a vanilla neural network of the same width and depth without the embedding function.  For these classification tasks, we use the cross-entropy loss, which is the negative log-likelihood of the relevant  multinomial distribution.

\begin{figure}[htbp]
\centering
\begin{subfigure}{0.49\textwidth}
    \includegraphics[width=\textwidth]{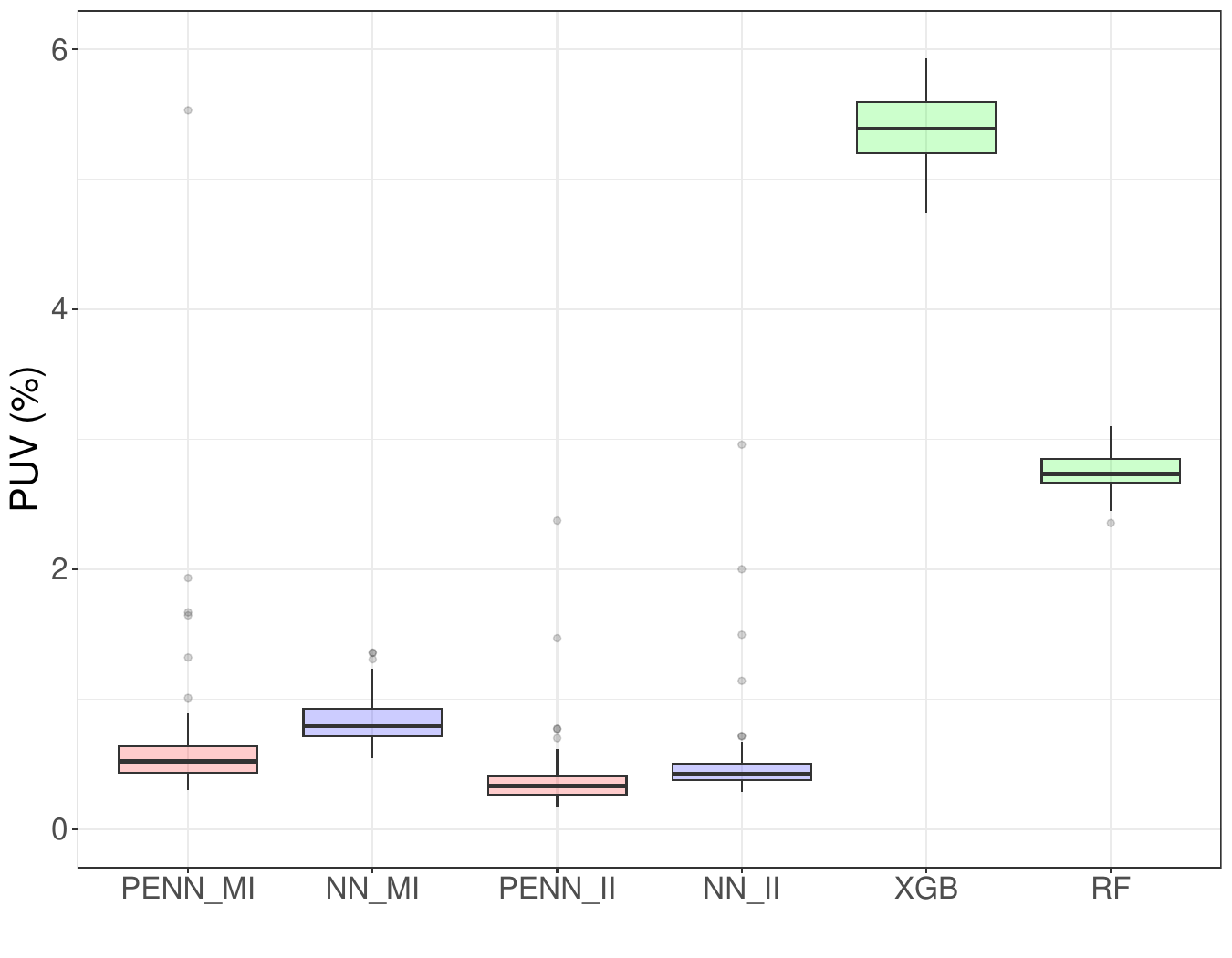}
    \caption{MCAR missingness}
\end{subfigure}
\begin{subfigure}{0.49\textwidth}
    \includegraphics[width=\textwidth]{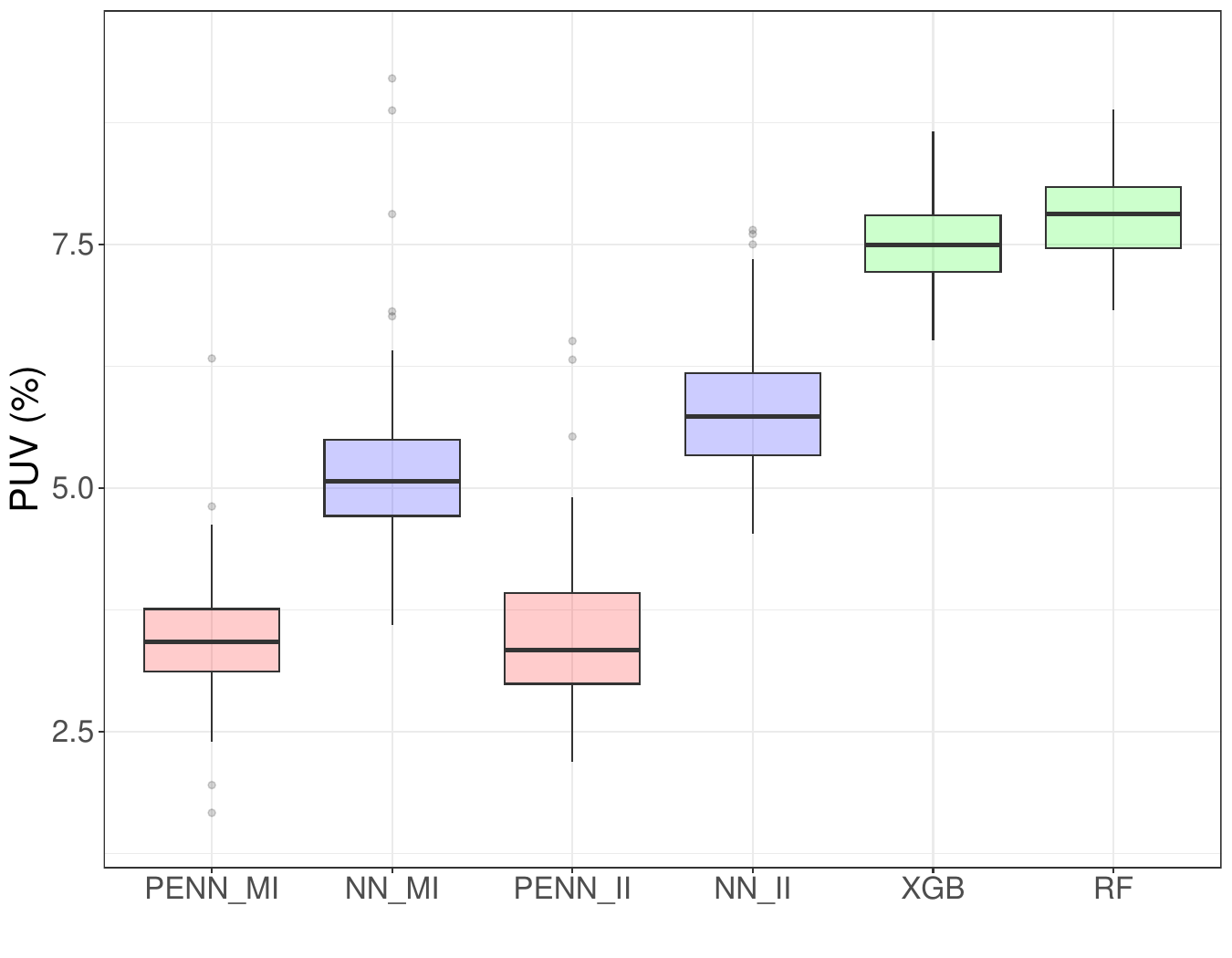}
    \caption{MNAR missingness}
\end{subfigure}
\caption{Proportion of unexplained variance (PUV) for relative location of CT dataset with different missingness mechanisms.} 
\label{fig:CT}
\end{figure}

\begin{figure}[htbp]
\centering
\begin{subfigure}{0.49\textwidth}
    \includegraphics[width=\textwidth]{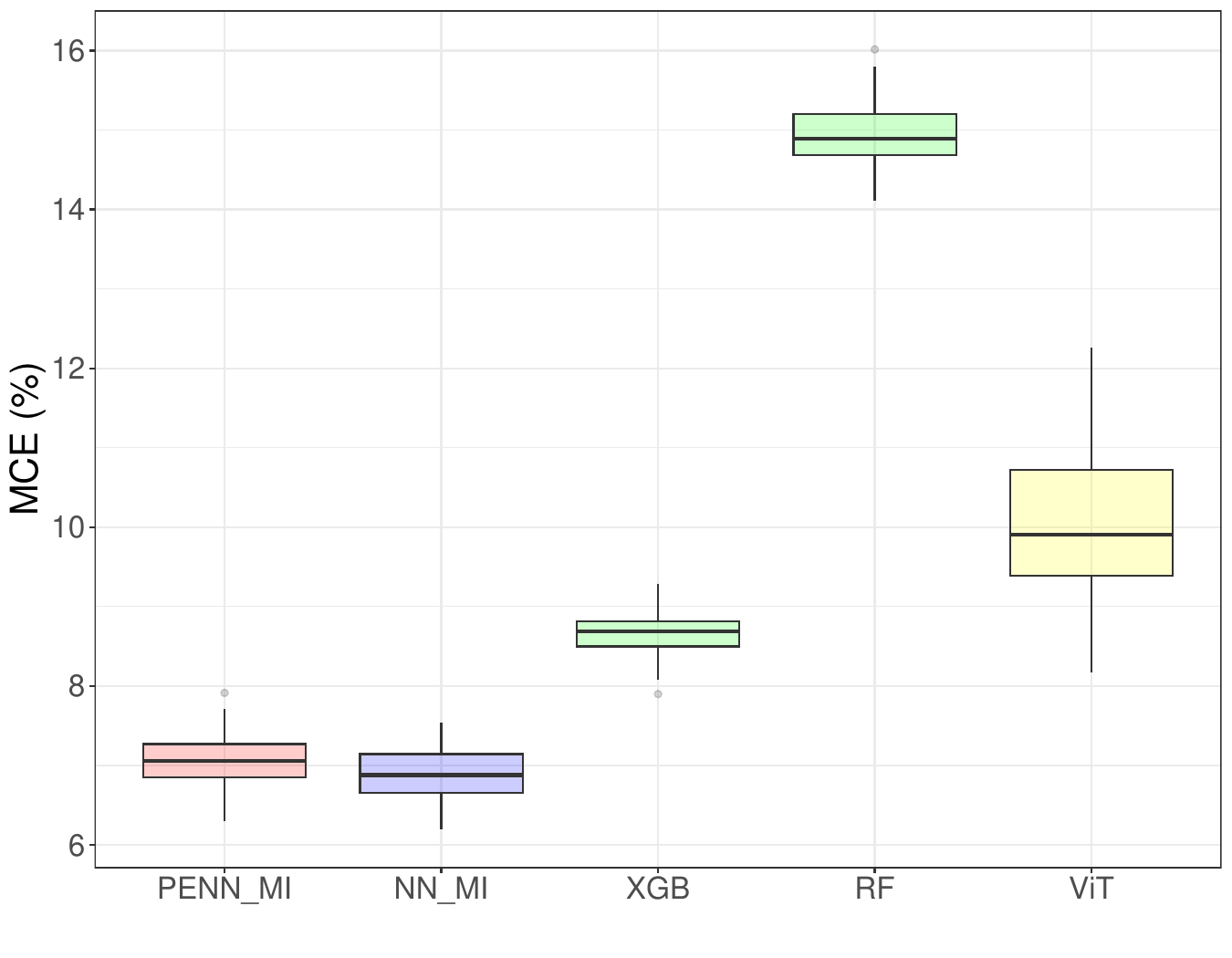}
    \caption{MCAR missingness}
\end{subfigure}
\begin{subfigure}{0.49\textwidth}
    \includegraphics[width=\textwidth]{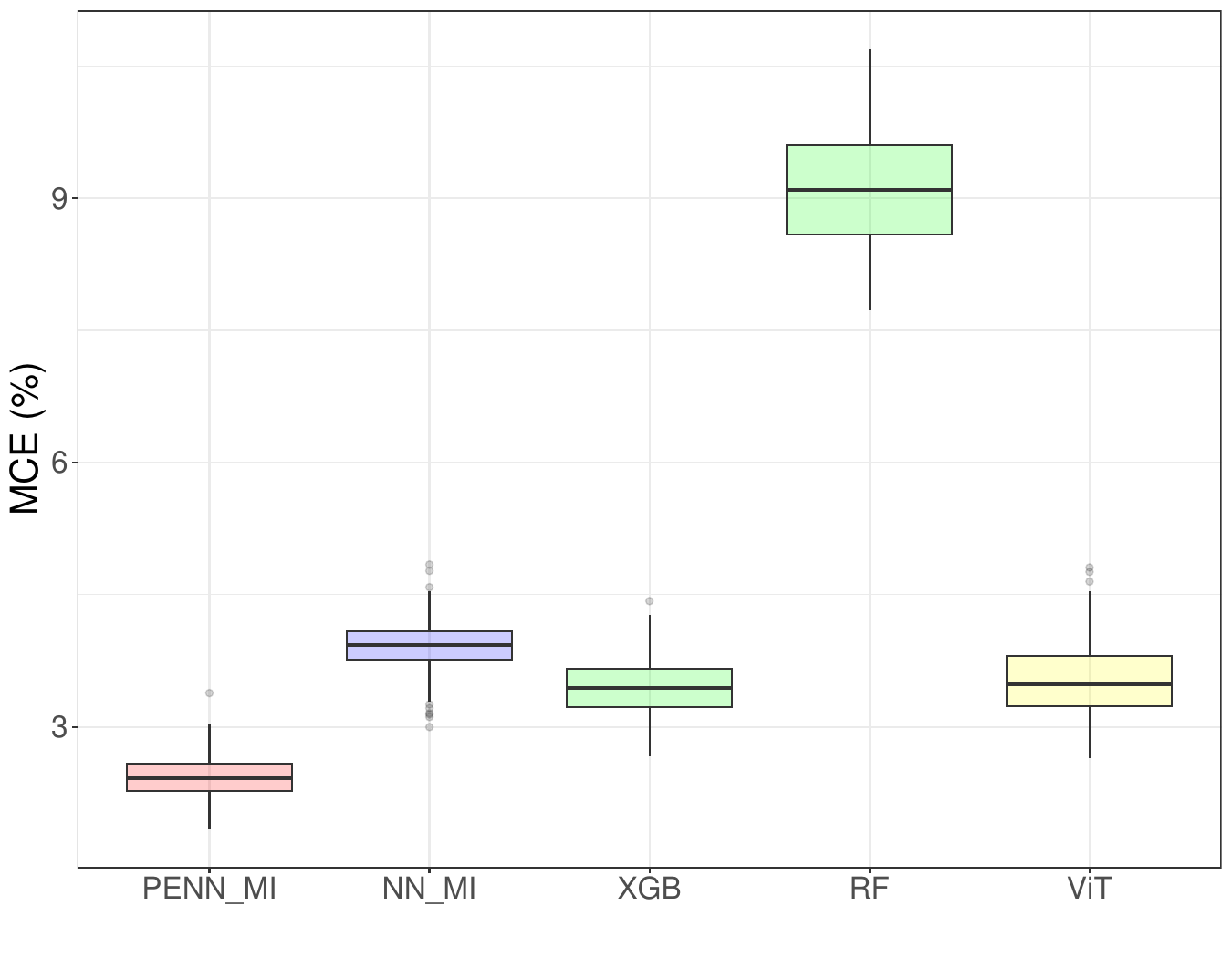}
    \caption{MNAR missingness}
\end{subfigure}
\caption{Misclassification error (MCE) for MNIST dataset with different missingness mechanisms.} 
\label{fig:mnist}
\end{figure}

For both datasets, and for each of 100 repetitions, we randomly split the dataset into training, validation and test sets with sizes in the ratio 8:1:1 after introducing the missingness.  For the regression task, we measure the performance of the algorithms via the proportion of unexplained variance (PUV), which is the ratio of the mean squared prediction error on the test set to the sample variance of the response on the same test set.
On the other hand, for the classification task, we measure the performance of the algorithms by the misclassification error (MCE) on the test set. 
The results for the two different datasets are presented in Figures~\ref{fig:CT} and~\ref{fig:mnist}.  For MCAR missingness, the PENN estimator is at least comparable with the vanilla NN estimator in terms of overall performance, while for MNAR missingness, it provides a significant improvement.  PENN also improves substantially on XGBoost, random forests and ViT.

\subsection{Real data}

Finally, we now turn to a further real dataset that already has some missing values.

 \paragraph{Credit score prediction dataset:}  The credit score prediction dataset from \url{https://www.kaggle.com/datasets/prasy46/credit-score-prediction} has $d=304$ and total sample size of $100{,}000$.  The response variable of interest is the credit rating, quantified as a positive integer ranging from 300 to 839 in the available data.  There are 41 columns with missingness, and the observation probabilities for the columns with missingness are given in the left panel of Figure~\ref{fig:credit-score-prediction}.
 % \begin{figure}[htbp]
 % \centering
 % \includegraphics[width=0.4\textwidth]{figures/credit-score-hist.pdf}
 % \caption{Histograms of the observation probabilities for the columns with missingness in credit score prediction dataset.} 
 % \label{fig:observational-probability-hist}
 % \end{figure}

% \red{Find a replacement for this!}
% \paragraph{Public procurement dataset:} For the public procurement dataset from \url{https://www.openml.org/search?type=data&status=active&id=42163}, we use only the numeric variables, yielding $d=25$ variables and a total sample size of $565{,}163$.  The original response \texttt{award\_value\_euro} records the values of public procurement contracts (in Euros), ranging from $-1$ to over thirteen billion, so we use a log transformation and predict $\log(\texttt{award\_value\_euro}+2)$ instead.  There are 14 columns with missingness, and the observation probabilities for the columns with missingness are given in Figure~\ref{fig:observational-probability-hist}(b).

Again, for each of 30 repetitions, we randomly split the dataset into training, validation and test sets with sizes in the ratio 8:1:1.  The PUVs are presented in the right panel of Figure~\ref{fig:credit-score-prediction} for the credit score data.  The improvements of PENN over both the vanilla NN estimator, as well as XGBoost and random forests, are again very notable.

\begin{figure}[htbp]
    \centering
    \includegraphics[width=0.444\textwidth]{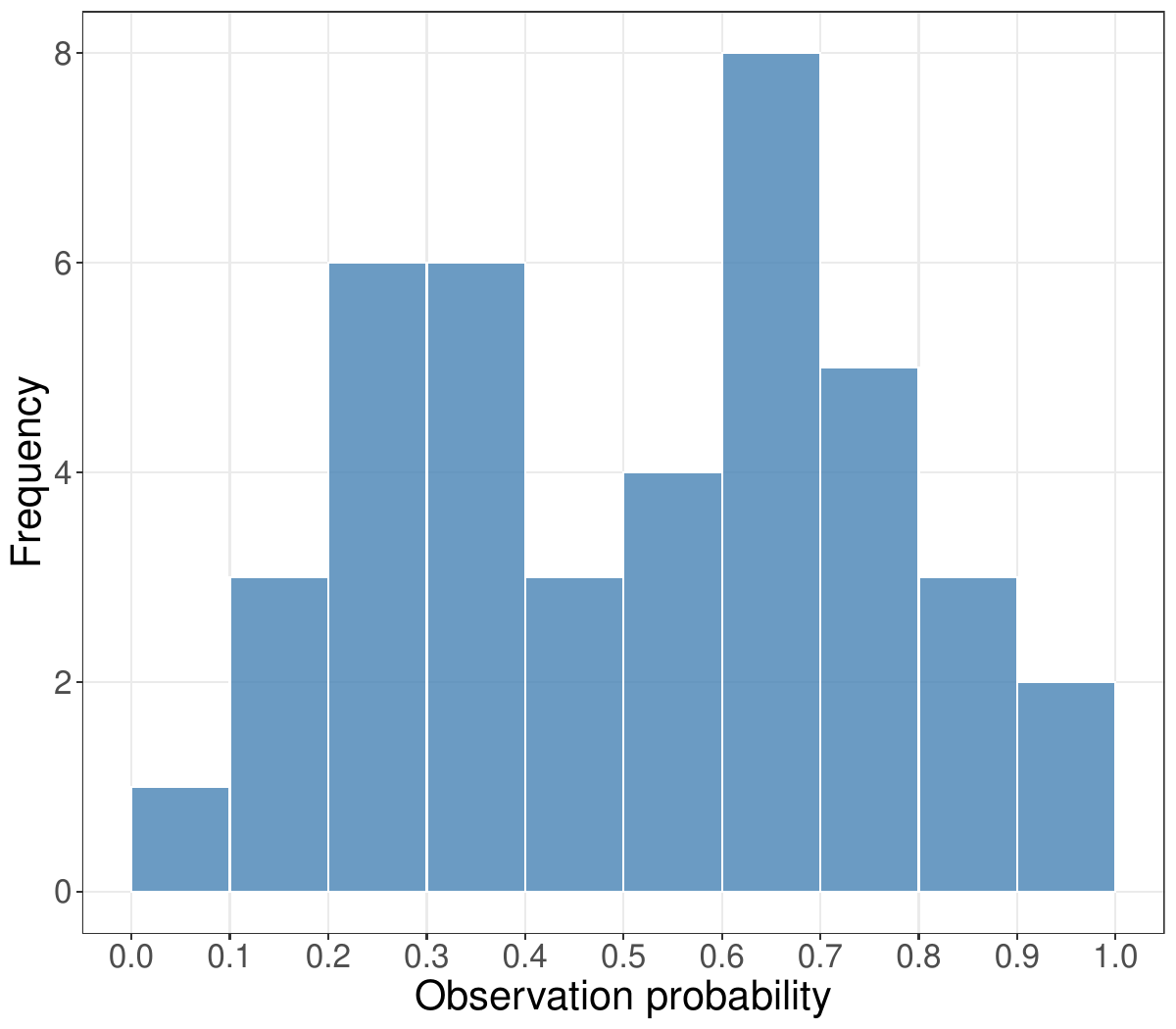}\quad
    \includegraphics[width=0.5\linewidth]{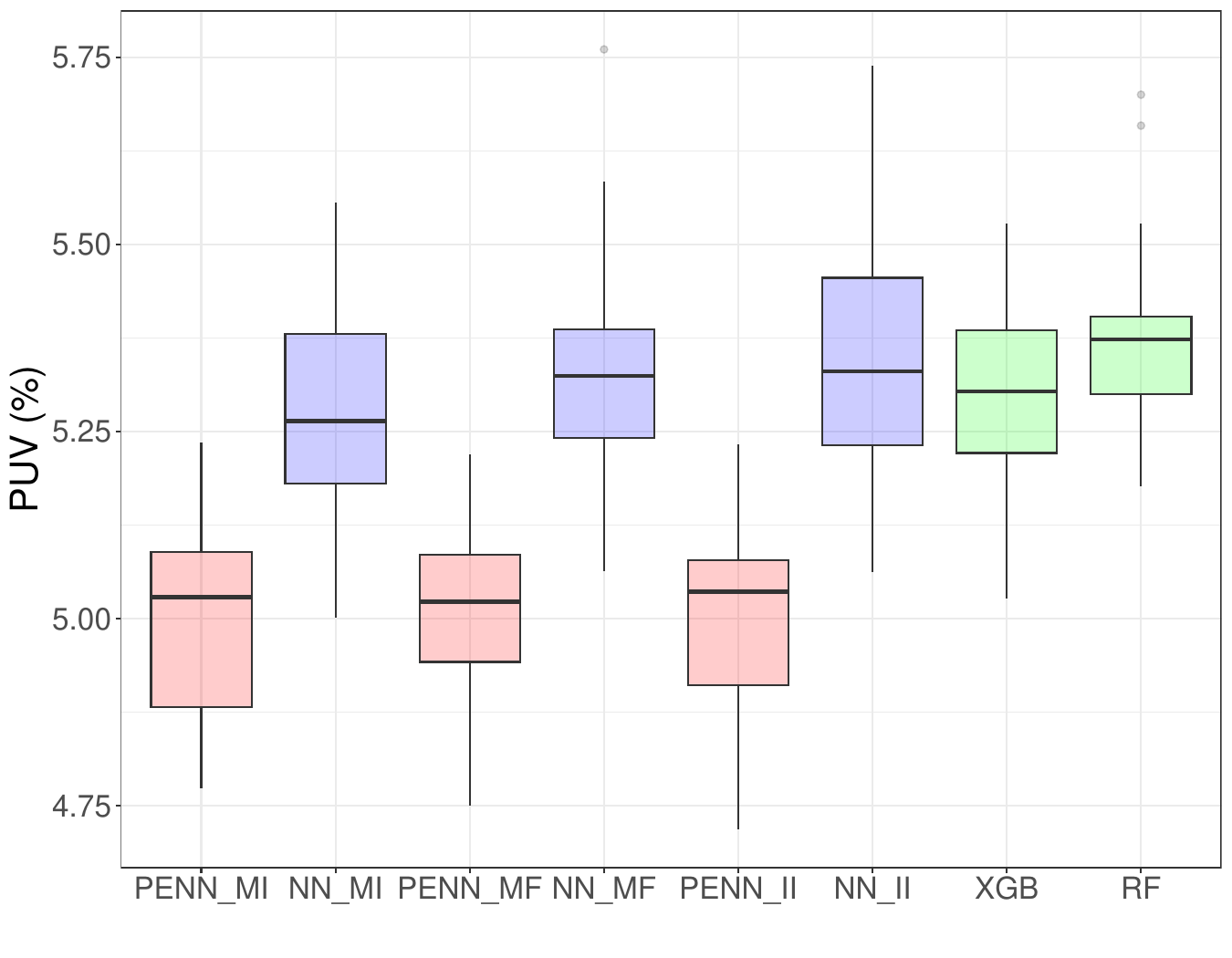}
    % \begin{subfigure}[t]{0.49\textwidth}
    %     \includegraphics[width=\textwidth]{figures/credit-score-hist.pdf}
    %     \caption{}
    % \end{subfigure}
    % \begin{subfigure}[t]{0.49\textwidth}
    %     \includegraphics[width=\linewidth]{figures/credit_score_prediction.pdf}
    % \caption{}
    % \end{subfigure}
    \caption{Left: Histogram of the observation probabilities for the columns with missingness.  Right: PUVs for the credit score prediction dataset.}
    \label{fig:credit-score-prediction}
\end{figure}

% \begin{figure}[htbp]
%     \centering
%     \includegraphics[width=0.7\linewidth]{figures/public-procurement-prediction.pdf}
%     \caption{PUV for public procurement dataset.}
%     \label{fig:public-procurement-prediction}
% \end{figure}

\section{Discussion} \label{sec:discussion}

In this paper, we have introduced Pattern Embedded Neural Networks (PENNs) as a novel deep learning architecture for nonparametric regression with missing covariates.  Example~\ref{example:f-star}, which is illustrated in Figure~\ref{fig:quadratic}, shows that the function $\bm{z} \mapsto \mathbb{E}(\bm{Y} \, | \, \bm{Z} = \bm{z})$ may have discontinuities caused by the different observation patterns, so one of the main messages of this example is that the revelation vectors should be incorporated into the modelling process.

On the other hand, Figure~\ref{fig:quadratic-d} provides the insight that simply including the revelation vectors as additional covariates (thereby increasing the dimension of the input space by a factor of 2) may even harm predictive performance.  The primary methodological contribution of PENNs, then, is to provide an effective compression of the information in the revelation vectors by embedding them into a relatively low-dimensional space.  An attractive feature of the methodology is that it can be used in conjunction with any existing imputation technique.

We provide theoretical support for the use of PENNs in the form of a finite-sample upper bound on their excess risk that holds without any assumption on the relationship between the data generating and missingness mechanisms.  Our Assumption~\ref{assumption:piecewise-assumption}, which partitions the Bayes regression function into different cells, allows us to borrow strength across different observation patterns, and the main conclusion of Theorems~\ref{thm:PENN-ub} and~\ref{thm:minimax-lb} is that PENNs typically achieve the minimax rate of convergence as if the cells in this partition were known in advance, up to a poly-logarithmic factor in the sample size.

Our theoretical results are complemented by numerical experiments on synthetic, semi-synthetic and real data, which confirm the practical improvements offered by PENNs.

\section*{Acknowledgement}
The research of TM and RJS was supported by RJS's European Research Council Advanced Grant 101019498.

\bibliographystyle{custom}
\bibliography{bibliography}

\appendix

\section{Some properties of neural networks} \label{appendix-a}
We summarise some basic operations associated with neural networks:
\begin{enumerate}[label=(P\arabic*)]
    \item \label{NN-enlarging} \textbf{Neural network enlarging:} Let $L\in\mathbb{N}$, let $\bm p_1, \bm p_2 \in \mathbb{N}^{L+2}$ be such that $\bm p_1 \leq \bm p_2$ coordinate-wise and let $s\in\mathbb{N}$. Then, from the definition, $\mathcal{F}(L,\bm p_1, s) \subseteq \mathcal{F}(L,\bm p_2, s)$.

    \item \label{NN-composition} \textbf{Neural network composition:} $\bm p_1 = (p_{1,0},\ldots,p_{1,L_1+1})^\top \in \mathbb{N}^{L_1+2}$ and $\bm p_2 = (p_{2,0},\ldots,p_{2,L_2+1})^\top \in \mathbb{N}^{L_2+2}$ satisfy $p_{1,L_1+1} = p_{2,0}$; let $\bm f_1 \in\mathcal{F}(L_1,\bm p_1)$ and $\bm f_2 \in\mathcal{F}(L_2,\bm p_2)$ (if either $\bm f_1$ or $\bm f_2$ is an affine function, then we set the corresponding $L_1$ or $L_2$ to be zero). Then
    \begin{gather*}
        \bm f_2 \circ \bm f_1 \in \mathcal{F}(L_1+L_2,\, \bm p_3),\text{ where}\\
        \bm p_3 \coloneqq (p_{1,0},\, \ldots,\, p_{1,L_1},\, p_{2,1},\, \ldots,\, p_{2,L_2+1}) \in \mathbb{N}^{L_1+L_2+2}. 
    \end{gather*}
    To see this, suppose that $\bm f_1(\cdot) = \bm{A}_{L_1+1}^{(1)} \circ \bm{\sigma} \circ \bm{A}_{L_1}^{(1)} \circ \bm{\sigma} \circ \cdots \circ \bm{A}_2^{(1)} \circ \bm{\sigma} \circ \bm{A}_1^{(1)}(\cdot)$ and $\bm f_2(\cdot) = \bm{A}_{L_2+1}^{(2)} \circ \bm{\sigma} \circ \bm{A}_{L_2}^{(2)} \circ \bm{\sigma} \circ \cdots \circ \bm{A}_2^{(2)} \circ \bm{\sigma} \circ \bm{A}_1^{(2)}(\cdot)$.  Then
    \begin{align*}
        \bm f_2 \circ \bm f_1 (\cdot) = \bm{A}_{L_2+1}^{(2)} \circ \bm{\sigma} \circ \cdots \circ \bm{A}_2^{(2)} \circ \bm{\sigma} \circ \bigl(\bm{A}_1^{(2)} \circ \bm{A}_{L_1+1}^{(1)}\bigr) \circ \bm{\sigma} \circ \cdots \circ \bm{A}_2^{(1)} \circ \bm{\sigma} \circ \bm{A}_1^{(1)}(\cdot)
    \end{align*}
    where we note that $\bm{A}_1^{(2)} \circ \bm{A}_{L_1+1}^{(1)}: \mathbb{R}^{p_{1,L_1}} \to \mathbb{R}^{p_{2,1}}$ is an affine function. 

    \item \label{NN-padding} \textbf{Neural network padding:} Let $L_1,L_2 \in \mathbb{N}$ be such that $L_1 < L_2$, let $\bm p_1 = (p_{1,0},\ldots,p_{1,L_1+1})^\top \in \mathbb{N}^{L_1+2}$, $\bm p_2 = (\bm p_1, 2p_{1,L_1+1},\ldots, 2p_{1,L_1+1})^\top \in \mathbb{N}^{L_2+2}$ and $s\in\mathbb{N}$. Then $\mathcal{F}(L_1, \bm p_1, s) \subseteq \mathcal{F}\bigl(L_2, \bm p_2, 2s+2p_{1,L_1+1}(L_2-L_1)\bigr)$.  To see this, let $\bm \phi \in \mathcal{F}\bigl(1, (p_{1,L_1+1}, 2p_{1,L_1+1}, p_{1,L_1+1})\bigr)$ be defined by $\bm \phi (\bm y) \coloneqq \bm\sigma(\bm y) - \bm\sigma(-\bm y)$ for $\bm y \in \mathbb{R}^{p_{1,L_1+1}}$, so that $\bm\phi(\bm y) = \bm y$. Let $\bm f_1 \in \mathcal{F}(L_1, \bm p_1, s)$, and define $\bm f_2 \coloneqq \bm\phi \circ \cdots\circ \bm\phi \circ \bm f_1$, where $\bm\phi$ is applied $(L_2-L_1)$-times. Then $\bm f_1(\bm x) = \bm f_2(\bm x)$ for all $\bm x \in \mathbb{R}^{p_{1,0}}$, and $\bm f_2 \in \mathcal{F}\bigl(L_2, \bm p_2, 2s + 2p_{1,L_1+1}(L_2-L_1)\bigr)$ by~\ref{NN-composition} and counting parameters.
    
    \item \label{NN-parallelisation} \textbf{Neural network parallelisation:} Let $N, L \in\mathbb{N}$ and for $i\in[N]$, let $\bm f_i \in \mathcal{F}(L,\bm p_i,s_i)$ where $\bm p_i = (d,p_{i,1},\ldots,p_{i,L+1})^\top \in \mathbb{N}^{L+2}$ and $s_i \in \mathbb{N}$. Then, writing $\bm f(\cdot) = \bigl(\bm f_1(\cdot)^\top, \ldots, \bm f_N(\cdot)^\top \bigr)^\top$, we have from the definition that 
    \begin{align*}
        \bm f \in \mathcal{F}\biggl(L,\, \Bigl(d,\sum_{i=1}^N p_{i,1}, \ldots, \sum_{i=1}^N p_{i,L+1}\Bigr),\, \sum_{i=1}^N s_i\biggr).
    \end{align*}
\end{enumerate}

\section{Covering number bounds for sparse neural networks}

Our oracle inequality for the excess risk of sparse neural network estimators in Proposition~\ref{prop:oracle-inequality} relies on bounds on the covering number of the relevant class that we develop in this subsection.  In fact, our proof proceeds via a bound on the \emph{pseudo-dimension} of this class, which is related to its Vapnik--Chervonenkis (VC) dimension, so we begin by defining the notions we require.
\begin{defn}
    Let $\mathcal{H}$ be a set of functions from $\mathcal{X}\subseteq\mathbb{R}^d$ to $\{0,1\}$. For $m\in\mathbb{N}$, the \emph{shattering coefficient} is defined as
    \begin{align*}
        \mathrm{shat}(\mathcal{H},m) \coloneqq \max_{\bm{x}_1,\ldots,\bm{x}_m \in \mathcal{X}} \bigl| \bigl\{\bigl(h(\bm{x}_1), \ldots, h(\bm{x}_m)\bigr) : h\in\mathcal{H}\bigr\} \bigr|.
    \end{align*}
    The \emph{VC-dimension} of $\mathcal{H}$ is defined as
    \begin{align*}
        \mathrm{VCdim}(\mathcal{H}) \coloneqq \sup \{m\in\mathbb{N} : \mathrm{shat}(\mathcal{H},m) = 2^m\}.
    \end{align*}
\end{defn}

We define the \emph{sign function} $\mathrm{sgn} : \mathbb{R} \to \{0,1\}$ as\footnote{Although this is not the standard definition of the sign function, it is used elsewhere in the neural network literature \citep[e.g.][]{bartlett2019nearly}, and it is convenient for our purposes here.} $\mathrm{sgn}(x) \coloneqq \mathbbm{1}_{\{x>0\}}$. If $\mathcal{F}$ is a class of real-valued functions, then we define $\mathrm{sgn}\circ\mathcal{F} \coloneqq \{\mathrm{sgn}\circ f : f\in\mathcal{F}\}$.
\begin{defn}
    Let $\mathcal{F}$ be a set of functions from $\mathcal{X}\subseteq\mathbb{R}^d$ to $\mathbb{R}$. The \emph{pseudo-dimension} of~$\mathcal{F}$ is defined as
    \begin{align*}
        \mathrm{Pdim}(\mathcal{F}) \coloneqq \mathrm{VCdim}\bigl( \bigl\{(\bm{x},y) \mapsto \mathrm{sgn}\bigl(f(\bm{x})-y\bigr) : f\in\mathcal{F}\bigr\} \bigr).
    \end{align*}
    For $\epsilon>0$, we say that a collection of functions $g_1,\ldots,g_N:\mathcal{X}\to\mathbb{R}$ is an \emph{$\epsilon$-cover} of~$\mathcal{F}$ with respect to $\|\cdot\|_{L_q(\mu)}$ if for every $f\in\mathcal{F}$, there exists $j=j(f) \in [N]$ such that $\|f-g_j\|_{L_q(\mu)} \leq \epsilon$. The \emph{$\epsilon$-covering number} of $\mathcal{F}$ with respect to $\|\cdot\|_{L_q(\mu)}$, written as $\mathcal{N}(\epsilon,\mathcal{F},\|\cdot\|_{L_q(\mu)})$, is the cardinality of the smallest $\epsilon$-cover of $\mathcal{F}$ with respect to $\|\cdot\|_{L_q(\mu)}$ (if no finite $\epsilon$-cover exists, then the $\epsilon$-covering number is $\infty$).
\end{defn}

The following proposition provides upper bounds on the VC-dimension and covering number of the class $\mathcal{F}(L,\bm{p},s)$. The proof is based on \citet[Theorem~7]{bartlett2019nearly}, which gives an upper bound on the VC-dimension of the class $\mathcal{F}(L,\bm{p})$. We also remark that upper bounds on the covering number of the class $\{f\in\mathcal{F}(L,\bm{p},s) : \|\bm{\Theta}(\bm{f})\|_{\infty} \leq 1\}$ have been obtained by, e.g.~\citet[Lemma~5]{schmidt-hieber2020nonparametric}; however, these do not imply an upper bound on its VC-dimension, and they do not generalise to the class $\mathcal{F}(L,\bm{p},s)$.

\begin{prop}\label{prop:VC-upper-bound}
    Let $L\in\mathbb{N}$, $\bm{p} = (p_0,\ldots,p_{L+1})^\top \in \mathbb{N}^{L+2}$ with $p_{L+1}=1$, let $V \coloneqq \sum_{\ell=1}^{L+1} p_\ell(p_{\ell-1}+1)$ and let $s\in[V]$.
    \begin{itemize}
        \item[(a)] We have
        \begin{align*}
            \mathrm{VCdim}\bigl(\mathrm{sgn}\circ\mathcal{F}(L,\bm{p},s)\bigr) &\leq \mathrm{Pdim}\bigl(\mathcal{F}(L,\bm{p},s)\bigr)\\
            &\leq 6s(L+1)\log_2(3s) + 2s\log_2(2p_0) \\
            &\lesssim sL\log (es) + s\log(ep_0).
        \end{align*}
        \item[(b)] For $q \in [1,\infty)$, $B>0$, $\epsilon\in(0, B/2)$ and any probability measure $\mu$ on $\mathbb{R}^d$, we have
        \begin{align*}
            \log \mathcal{N}\bigl(\epsilon,\, T_B \circ \mathcal{F}(L,\bm{p},s),\, \|\cdot\|_{L_q(\mu)} \bigr) &\leq 2q\cdot\mathrm{Pdim}\bigl(\mathcal{F}(L,\bm{p},s)\bigr)\log(10B/\epsilon) \\
            &\lesssim q\bigl\{sL\log (es) + s\log(ep_0)\bigr\}\log(B/\epsilon).
        \end{align*}
    \end{itemize}
\end{prop}
By \citet[Theorem~3]{bartlett2019nearly}, there exists a universal constant $C_3 > 0$ such that if $s \geq C_3 L \geq C_3^2$, then there exists $\bm{p}' = (p'_0,\ldots,p'_{L+1})^\top \in \mathbb{N}^{L+2}$ with $V' \coloneqq \sum_{\ell=1}^{L+1} p_\ell'(p_{\ell-1}'+1) \leq s$, such that $\mathrm{VCdim}\bigl(\mathrm{sgn} \circ \mathcal{F}(L,\bm{p}')\bigr) \geq c_4\cdot sL\log(s/L)$ for some universal constant $c_4>0$.  By~\ref{NN-enlarging}, $\mathcal{F}(L,\bm{p}') \subseteq \mathcal{F}(L,\bm{p},s)$ for all $\bm{p}\in\mathbb{N}^{L+2}$ such that $\bm{p} \geq \bm{p}'$ coordinate-wise, so the same lower bound then applies to $\mathrm{VCdim}\bigl(\mathrm{sgn} \circ \mathcal{F}(L,\bm{p},s)\bigr)$.  Thus, under the above condition on $s$ and $L$, our upper bound in Proposition~\ref{prop:VC-upper-bound}\emph{(a)} is tight up to a logarithmic factor in $n$ when $\bm{p}$ is sufficiently large. %\red{Comment on that $s\log(ep_0/s)$ term is necessary in general} 
\begin{proof}
    \emph{(a)} We first observe that 
    \begin{align*}
    \mathrm{Pdim}\bigl(\mathcal{F}(L,\bm{p},s)\bigr) &= \mathrm{VCdim}\bigl( \bigl\{(\bm{x},y) \mapsto \mathrm{sgn}\bigl(f(\bm{x})-y\bigr) : f\in\mathcal{F}(L,\bm{p},s)\bigr\} \bigr) \\
    &\geq \mathrm{VCdim}\bigl( \bigl\{\bm{x} \mapsto \mathrm{sgn}\bigl(f(\bm{x})\bigr) : f\in\mathcal{F}(L,\bm{p},s)\bigr\} \bigr) \\
    &= \mathrm{VCdim}\bigl(\mathrm{sgn}\circ\mathcal{F}(L,\bm{p},s)\bigr).
    \end{align*}
    Writing $\tilde{\bm p} \coloneqq (p_0,p_1\wedge s,\ldots,p_L\wedge s, p_{L+1})^\top$, we have by~\citet[Equation~(19)]{schmidt-hieber2020nonparametric} that $\mathcal{F}(L,\bm p,s) = \mathcal{F}(L,\tilde{\bm p},s)$.
    Define a class of functions $\mathcal{H} \coloneqq \bigl\{(\bm{x},y) \mapsto \mathrm{sgn}\bigl(f(\bm{x}) - y\bigr) : f\in\mathcal{F}(L,\tilde{\bm p},s),\, \gamma\in\mathbb{R}\bigr\}$ from $\mathbb{R}^{p_0} \times\mathbb{R}$ to $\mathbb{R}$. Then $\mathrm{Pdim}\bigl(\mathcal{F}(L,\bm{p},s)\bigr) = \mathrm{VCdim}(\mathcal{H})$, so it suffices to upper bound $m\coloneqq \mathrm{VCdim}(\mathcal{H})$. If $m < s$, then part~\emph{(a)} of the proposition follows since the right-hand side is at least $s$. Therefore, we may assume without loss of generality that $m \geq s$. Suppose that $(\bm{x}_1,y_1),\ldots,(\bm{x}_m,y_m)\in\mathbb{R}^{p_0}\times \mathbb{R}$ are shattered by $\mathcal{H}$. Let $\tilde{V}\coloneqq \sum_{\ell=1}^{L+1} \tilde{p}_\ell (\tilde{p}_{\ell-1}+1)$ denote the total number of parameters of a neural network in $\mathcal{F}(L,\tilde{\bm p},s)$. 
    For $\bm{x}\in\mathbb{R}^{p_0}$ and $\bm{\theta} \in \{\bm{\theta}'\in\mathbb{R}^{\tilde{V}} : \|\bm{\theta}'\|_0 \leq s\} \eqqcolon \mathbb{B}_0^{\tilde{V}}(s)$, let $g(\bm{x},\bm{\theta}) \coloneqq f(\bm{x})$ when $f\in\mathcal{F}(L,\tilde{\bm p},s)$ satisfies $\bm{\Theta}(\bm{f}) = \bm{\theta}$. In other words, $g(\cdot,\bm{\theta})$ is the neural network in $\mathcal{F}(L,\tilde{\bm p},s)$ with parameter vector $\bm{\theta}$, so $\mathcal{F}(L,\tilde{\bm p},s) = \bigl\{g(\cdot,\bm{\theta}) : \bm{\theta} \in \mathbb{B}_0^{\tilde{V}}(s)\bigr\}$. We partition $\mathbb{B}_0^{\tilde{V}}(s)$ into sets $B_1,\ldots,B_{\binom{\tilde{V}}{s}}$, where the elements in each set $B_i$ are all supported on the same set of cardinality $s$. Then, by definition of $m$, 
    \begin{align}
        2^m &= \Bigl| \Bigl\{ \Bigl(\mathrm{sgn}\bigl(g(\bm{x}_1,\bm{\theta})- y_1\bigr), \ldots, \mathrm{sgn}\bigl(g(\bm{x}_m,\bm{\theta})- y_m\bigr)\Bigr)^\top : \bm{\theta}\in \mathbb{B}_0^{\tilde{V}}(s)\Bigr\} \Bigr|\nonumber\\
        &\leq \sum_{i=1}^{\binom{\tilde{V}}{s}} \Bigl| \Bigl\{ \Bigl(\mathrm{sgn}\bigl(g(\bm{x}_1,\bm{\theta})- y_1\bigr), \ldots, \mathrm{sgn}\bigl(g(\bm{x}_m,\bm{\theta})- y_m\bigr)\Bigr)^\top :  \bm{\theta}\in B_i\Bigr\} \Bigr| \eqqcolon \sum_{i=1}^{\binom{\tilde{V}}{s}} K_i. \label{eq:VC-bound-K-decomposition}
    \end{align}
    We will prove upper bounds for $K_1,\ldots,K_{\binom{\tilde{V}}{s}}$, which then imply an upper bound on $m$.  To this end, without loss of generality, we upper bound $K_1$. For $\ell\in[L+1]$, $\bm{x} \in \mathbb{R}^{p_0}$ and $\bm{\theta} \in B_1$, define $\bm{g}^{(\ell)}(\bm{x}, \bm{\theta}) \coloneqq \bm{A}_{\ell}\circ \bm{\sigma} \circ\cdots\circ \bm{\sigma}\circ \bm{A}_1(\bm{x})$, where $\bm{A}_1,\ldots,\bm{A}_{L+1}$ are defined analogously to~\eqref{Eq:Flp} (but with each $p_\ell$ there replaced with $\tilde{p}_\ell$) with weight matrices and bias vectors given by relevant components of~$\bm{\theta}$.  For $\ell \in [L+1]$ and $u\in[\tilde{p}_{\ell}]$, let $g^{(\ell)}_u$ be the $u$th coordinate function of $\bm g^{(\ell)}$, let $\mathcal{U}_\ell\coloneqq \{u\in[\tilde{p}_\ell] : g^{(\ell)}_u(\cdot,\bm\theta) \neq 0 \text{ for some } \bm\theta \in B_1\}$ be the coordinates of active neurons in the $\ell$th layer, and let $k_\ell \coloneqq |\mathcal{U}_\ell|$.  If $k_\ell = 0$ for some $\ell\in[L+1]$, then $g(\cdot,\bm\theta) = 0$ for all $\bm\theta \in B_1$, so $K_1=1$. Now assume that $k_\ell \geq 1$ for all $\ell\in[L+1]$.  We will construct a finite sequence of partitions    
    $\mathcal{P}_1,\ldots,\mathcal{P}_{L+1}$ of $B_1$, each refining the previous one, satisfying the following properties:
    \begin{enumerate}[label=(\roman*)]%[label=(P\arabic*)]
        \item \label{property1}$N_\ell \coloneqq |\mathcal{P}_\ell|$ satisfies $N_1 = 1$ and for $\ell\in [L]$, we have
        \begin{align*}
            \frac{N_{\ell+1}}{N_{\ell}} \leq 2\biggl( \frac{2em \ell k_{\ell}}{s} \biggr)^s \eqqcolon \phi_{\ell}.
        \end{align*}
        \item \label{property2} $\bm\theta \mapsto \bm g^{(\ell)}(\bm x_j, \bm\theta)$ is a polynomial of degree at most $\ell$ on $P_{\ell}$, for each $j\in[m]$, $\ell\in[L+1]$ and $P_\ell \in \mathcal{P}_\ell$.
        % For $\ell\in\{0,\ldots,L\}$, $P_{\ell} \in \mathcal{P}_{\ell}$, $j\in[m]$ and $u\in[p_{\ell+1}]$, the net input to the $u$th neuron in the $(\ell+1)$st layer of $g(x_j,\theta)$, i.e.\  the $u$th coordinate of $A_{\ell+1}\circ\sigma\cdots\circ\sigma\circ A_1(x_j)$ with the weight matrices and bias vectors given by $\theta$, is a linear function in at most $s$ coordinates of $\theta$ for $\theta\in P_{\ell}$.
    \end{enumerate}
    We start by defining $\mathcal{P}_1 \coloneqq \{B_1\}$. By the definition of $\bm A_1$ in~\eqref{Eq:Flp}, we have that $\bm\theta \mapsto \bm g^{(1)}(\bm x_j, \bm \theta) = \bm A_1(\bm x_j)$ is linear on $B_1$. Now suppose that we have constructed $\mathcal{P}_1,\ldots,\mathcal{P}_{\ell}$ for some $\ell \in [L]$ satisfying both Properties~\ref{property1} and~\ref{property2}.  By the induction hypothesis, $\mathcal{P}_{\ell} = \{P_{\ell,1},\ldots,P_{\ell,N_{\ell}}\}$ is such that, for each $r\in[N_\ell]$, $j\in[m]$ and $u\in\mathcal{U}_{\ell}$, the map $\bm \theta \mapsto g^{(\ell)}_u(\bm x_j,\bm \theta)$ is a polynomial of degree at most $\ell$ on $P_{\ell,r}$, depending on at most $s$ coordinates of $\bm \theta$. Then, for a fixed $r\in[N_\ell]$, by \citet[Lemma~17]{bartlett2019nearly}, since $m \geq s$, we have
    \begin{align*}
        \Bigl|\Bigl\{\Bigl( \mathrm{sgn}\bigl(g^{(\ell)}_u&(\bm x_j,\bm \theta)\bigr) : u\in[\tilde{p}_{\ell}],\, j\in[m] \Bigr) : \bm \theta\in P_{\ell,r}\Bigr\} \Bigr|\\
        & = \Bigl|\Bigl\{\Bigl( \mathrm{sgn}\bigl(g^{(\ell)}_u(\bm x_j,\bm \theta)\bigr) : u\in \mathcal{U}_{\ell},\, j\in[m] \Bigr) : \bm \theta\in P_{\ell,r}\Bigr\} \Bigr| \leq \phi_{\ell}. 
    \end{align*}
    % For any $u\in[p_{\ell+1}]$, $j\in[m]$ and $P_{\ell} \in \mathcal{P}_{\ell}$, let $h_{u,x_j,P_{\ell}}(\theta)$ be the net input to the $u$th neuron in the $(\ell+1)$th layer of $g(x_j,\theta)$. By induction hypothesis, $h_{u,x_j,P_{\ell}}$ is a linear function in at most $s$ coordinates of $\theta$ for $\theta\in P_{\ell}$. Then, by \citet[Lemma~17]{bartlett2019nearly}, we have
    % \begin{align*}
    %     \bigl|\bigl\{\bigl( \mathrm{sgn}(h_{u,x_j,P_{\ell}}(\theta)) : u\in[p_{\ell+1}],\, j\in[m] \bigr) : \theta\in B_1\bigr\} \bigr| \leq 2\biggl( \frac{2emp_{\ell+1}}{s} \biggr)^s \eqqcolon \phi_{\ell+1}. 
    % \end{align*}
    Therefore, for each $r\in[N_\ell]$, we can partition $P_{\ell,r}$ into regions $R_{r,1},\ldots,R_{r,\phi_{\ell}}$ such that for each $t \in [\phi_{\ell}]$, the map $\bm \theta \mapsto \Bigl( \mathrm{sgn}\bigl(g_u^{(\ell)}(\bm x_j,\bm \theta)\bigr) : u\in[\tilde{p}_{\ell}],\, j\in[m] \Bigr)$ is constant on~$R_{r,t}$.  We then define $\mathcal{P}_{\ell+1} \coloneqq \bigl\{  R_{r,t} : r\in[N_\ell],\, t\in[\phi_{\ell}] \bigr\}$. By construction, $\mathcal{P}_{\ell+1}$ satisfies Property~\ref{property1}. Moreover, writing $\bm \xi_{\bm x_j,\bm \theta} \coloneqq \Bigl(\mathrm{sgn}\bigl(g_u^{(\ell)}(\bm x_j,\bm \theta)\bigr):u\in [\tilde{p}_{\ell}]\Bigr)^\top \in\{0,1\}^{\tilde{p}_{\ell}}$, we have 
    \[  
    \bm g^{(\ell+1)}(\bm x_j,\bm \theta) = \bm A_{\ell+1}\circ\bm \sigma\circ \bm g^{(\ell)}(\bm x_j,\bm \theta) = \bm A_{\ell+1}\bigl(\mathrm{diag}(\bm \xi_{\bm x_j,\bm \theta}) \bm g^{(\ell)}(\bm x_j,\bm \theta) \bigr).
    \]
    Since for each $j\in[m]$, $r\in[N_\ell]$ and $t\in[\phi_{\ell}]$, the sign vector $\bm \xi_{\bm x_j,\bm \theta}$ is constant for all $\bm \theta\in R_{r,t}$, the map $\bm \theta\mapsto \bm g^{(\ell+1)}(\bm x_j,\bm \theta)$ is a polynomial of degree at most $\ell+1$ on $R_{r,t}$, which verifies Property~\ref{property2} for $\ell+1$ and hence completes the induction.
    
    % Moreover, for $P_{\ell+1} \in \mathcal{P}_{\ell+1}$ and $\theta\in P_{\ell+1}$, since the sign vector of the net input vector in the $(\ell+1)$th layer stays constant, after applying the ReLU activation function and a further affine map $A_{\ell+2}$, the net input to each neuron in the $(\ell+2)$th layer is a linear function of $\theta$. Finally, the fact that it depends only on $s$ coordinates of $\theta$ follows since all vectors in $B_1$ are supported on the same set with size at most $s$. Hence, property~\ref{property2} also holds for $\mathcal{P}_{\ell+1}$. This complete the induction.

    Next, by Property~\ref{property2}, for each $j\in[m]$ and $P_{L} \in \mathcal{P}_L$, we have that $\bm \theta \mapsto g(\bm x_j,\bm \theta) - y_j = \bm g^{(L+1)}(\bm x_j,\bm \theta) - y_j$ is a polynomial of degree at most $L$ on $P_L$, depending on at most~$s$ coordinates of $\bm \theta$. Thus, by \citet[Lemma~17]{bartlett2019nearly} again,
    \begin{align}
        \Bigl| \Bigl\{ \Bigl(\mathrm{sgn}\bigl(g(\bm x_1,\bm \theta) - y_1\bigr), \ldots, \mathrm{sgn}\bigl(g(\bm x_m,\bm \theta) - y_m\bigr)\Bigr)^\top : \bm \theta\in P_L \Bigr\} \Bigr| \leq 2\biggl( \frac{2emL}{s} \biggr)^{s}. \label{eq:VC-bound-ineq-1}
    \end{align}
    By Property~\ref{property1}, we also have 
    \begin{align}
        |\mathcal{P}_L| \leq 2^L\prod_{\ell=1}^L \biggl( \frac{2em\ell k_{\ell}}{s} \biggr)^s. \label{eq:VC-bound-ineq-2}
    \end{align}
    Combining~\eqref{eq:VC-bound-ineq-1} and~\eqref{eq:VC-bound-ineq-2} yields that
    \begin{align*}
        K_1 &\leq 2^{L+1}\biggl( \frac{2emL}{s} \biggr)^{s} \prod_{\ell=1}^{L} \biggl( \frac{2em \ell k_{\ell}}{s} \biggr)^s = 2^{L+1} \biggl( \frac{2em}{s} \biggr)^{s(L+1)} \biggl(\prod_{\ell=1}^L \ell k_\ell\biggr)^s\\
        &\leq 2^{L+1} \biggl( \frac{2em}{s} \biggr)^{s(L+1)} \biggl(\frac{1}{L}\sum_{\ell=1}^L \ell k_\ell\biggr)^{sL}\leq 2^{L+1} (2em)^{s(L+1)},
    \end{align*}
    where the second inequality is an application of the AM--GM inequality, and the final bound uses the fact that $\sum_{\ell=1}^L k_\ell \leq s$.  Therefore, by~\eqref{eq:VC-bound-K-decomposition}, we deduce that
    \begin{align*}
        2^m \leq \binom{\tilde{V}}{s} \cdot 2^{L+1} (2em)^{s(L+1)} &\leq  \bigl( 4em \tilde{V}^{1/(L+1)}\bigr)^{s(L+1)}. 
    \end{align*}
    By Lemma~\ref{lemma:LambertW}, since $4e\tilde{V}^{1/(L+1)} \geq 4$, we have
    \begin{align*}
        m &\leq 2s(L+1)\log_2\bigl(4es(L+1)\tilde{V}^{1/(L+1)}\bigr) \\
        &\leq 2s(L+1)\log_2\bigl(8es^3(2p_0)^{1/(L+1)}\bigr) \leq 6s(L+1)\log_2(3s) + 2s\log_2(2p_0),
    \end{align*}
    where the second inequality follows since $L+1 \leq s$ by our assumption that $k_\ell \geq 1$ for all $\ell\in[L+1]$ and
    \begin{align*}
        \tilde{V} \leq (L-1)s(s+1) + (p_0+1)s + s + 1 \leq 2(L+1)s^2 + p_0s \leq 2p_0(2s)^{L+1}.
    \end{align*}

    \medskip \emph{(b)} If $(\bm x_1,y_1),\ldots,(\bm x_m,y_m) \in \mathbb{R}^{p_0} \times \mathbb{R}$ are shattered by $\bigl\{(\bm x,y) \mapsto \mathrm{sgn}\bigl(T_B f(\bm x)-y\bigr) : f\in\mathcal{F}(L,\bm p,s)\bigr\}$, then we must have $y_j\in[-B,B)$ for all $j\in[m]$. Therefore, $(\bm x_1,y_1),\ldots,(\bm x_m,y_m)$ are also shattered by $\bigl\{(\bm x,y) \mapsto \mathrm{sgn}\bigl(f(\bm x)-y\bigr) : f\in\mathcal{F}(L,\bm p,s)\bigr\}$, so $\mathrm{Pdim}\bigl(T_B \circ \mathcal{F}(L,\bm p,s)\bigr) \leq \mathrm{Pdim}\bigl(\mathcal{F}(L,\bm p,s)\bigr)$. Then, by \citet[Theorem~9.4]{gyorfi2006distribution}, since $\mathrm{Pdim}\bigl(T_B \circ \mathcal{F}(L,\bm p,s)\bigr) \geq 2$, we have
    \begin{align*}
        \mathcal{N}\bigl( \epsilon, T_B \circ \mathcal{F}(L,\bm p,s), \|\cdot\|_{L_q(\mu)} \bigr) &\leq 3\biggl\{ \frac{2e(2B)^q}{\epsilon^q} \log \biggl(\frac{3e(2B)^q}{\epsilon^q}\biggr) \biggr\}^{\mathrm{Pdim}(T_B \circ \mathcal{F}(L,\bm p,s))}\\
        &\leq \biggl(\frac{10B}{\epsilon}\biggr)^{2q\cdot\mathrm{Pdim}(\mathcal{F}(L,\bm p,s))}.
    \end{align*}
    Taking logarithms and applying part~\emph{(a)} yields the desired result.
\end{proof}

\section{Proof of Proposition~\ref{prop:oracle-inequality}}

\begin{proof}[Proof of Proposition~\ref{prop:oracle-inequality}]
    For $i=0,1,\ldots,n$, let $\tilde{\bm Z}_i \coloneqq (\bm Z_i,\bm \Omega_i)$ and let $f^{\star}_{B_n}(\tilde{\bm Z}_i) \coloneqq \mathbb{E}(T_{B_n}Y_i \,|\, \tilde{\bm Z}_i)$. Then
    \begin{align*}
        \mathbb{E}\bigl\{ R(T_{B_n}\tilde{f}) - R(f^{\star}) \bigr\} &= \mathbb{E}\bigl\{ \bigl(T_{B_n}\tilde{f}(\tilde{\bm Z}_0) - Y_0\bigr)^2 - \bigl(T_{B_n}\tilde{f}(\tilde{\bm Z}_0) - T_{B_n}Y_0\bigr)^2 \bigr\}\\
        &\qquad + \mathbb{E}\bigl\{ \bigl(f^{\star}_{B_n}(\tilde{\bm Z}_0) - T_{B_n}Y_0\bigr)^2 - \bigl(f^{\star}(\tilde{\bm Z}_0) - Y_0\bigr)^2 \bigr\}\\
        &\qquad + \mathbb{E}\biggl[ \bigl(T_{B_n}\tilde{f}(\tilde{\bm Z}_0) - T_{B_n}Y_0\bigr)^2 - \bigl(f^{\star}_{B_n}(\tilde{\bm Z}_0) - T_{B_n}Y_0\bigr)^2\\
        &\qquad\quad -\frac{2}{n} \sum_{i=1}^n \Bigl\{ \bigl(T_{B_n}\tilde{f}(\tilde{\bm Z}_i) - T_{B_n}Y_i\bigr)^2 - \bigl(f^{\star}_{B_n}(\tilde{\bm Z}_i) - T_{B_n}Y_i\bigr)^2 \Bigr\}\biggr] \\
        &\qquad + \mathbb{E}\biggl[ \frac{2}{n} \sum_{i=1}^n \Bigl\{ \bigl(T_{B_n}\tilde{f}(\tilde{\bm Z}_i) - T_{B_n}Y_i\bigr)^2 - \bigl(f^{\star}(\tilde{\bm Z}_i) - Y_i\bigr)^2 \Bigr\} \biggr] \\
        &\qquad + \mathbb{E}\biggl[ \frac{2}{n} \sum_{i=1}^n \Bigl\{\bigl(f^{\star}(\tilde{\bm Z}_i) - Y_i\bigr)^2 - \bigl(f^{\star}_{B_n}(\tilde{\bm Z}_i) - T_{B_n}Y_i\bigr)^2\Bigr\} \biggr]\\
        &\eqqcolon E_1 + E_2 + E_3 + E_4 + E_5.
    \end{align*}

    \paragraph{Bounding $E_1$:} By Lemma~\ref{lemma:truncation-bias},
    \begin{align}
        E_1 \leq \mathbb{E}\bigl\{ Y_0^2 - (T_{B_n}Y_0)^2 \bigr\} + 2B_n\mathbb{E}\bigl\{ |Y_0 - T_{B_n}Y_0| \bigr\} \leq \frac{5\xi^2}{n}. \label{eq:oracle-inequality-E1}
    \end{align}

    \paragraph{Bounding $E_2 + E_5$:} Note that $E_5 = -2E_2$, so by Lemma~\ref{lemma:truncation-bias},
    \begin{align}
        E_2 + E_5 &= \mathbb{E}\bigl\{ \bigl(f^{\star}(\tilde{\bm Z}_0) - Y_0\bigr)^2 - \bigl(f^{\star}_{B_n}(\tilde{\bm Z}_0) - T_{B_n}Y_0\bigr)^2 \bigr\}\nonumber\\
        &= \mathbb{E}\bigl\{ \bigl(f^{\star}(\tilde{\bm Z}_0)\bigr)^2 - \bigl( f^{\star}_{B_n}(\tilde{\bm Z}_0) \bigr)^2 \bigr\} + \mathbb{E}\bigl\{ Y_0^2 - (T_{B_n} Y_0)^2 \bigr\}\nonumber\\
        &\hspace{5cm} - 2\mathbb{E}\bigl\{f^{\star}(\tilde{\bm Z}_0)\cdot Y_0 - f^{\star}_{B_n}(\tilde{\bm Z}_0) \cdot T_{B_n}Y_0\bigr\} \leq \frac{13\xi^2}{n}. \label{eq:oracle-inequality-E2+E5}
    \end{align}
    % Moreover, by the conditional version of Jensen's inequality,
    % \begin{align*}
    %     \mathbb{E}\bigl(e^{f^{\star}(\tilde{\bm Z}_0)^2 / \xi^2}\bigr) = \mathbb{E} \bigl[ \exp\bigl\{\bigl(\mathbb{E}(Y_0\,|\,\tilde{\bm Z}_0)\bigr)^2 / \xi^2\bigr\} \bigr] \leq \mathbb{E} \bigl[ \mathbb{E}\bigl\{ \exp(Y_0^2 / \xi^2) \,\big|\, \tilde{\bm Z}_0\bigr\} \bigr] = \mathbb{E}\bigl(e^{Y_0^2/\xi^2}\bigr) \leq 2,
    % \end{align*}
    % so $\|f^{\star}(\tilde{\bm Z}_0)\|_{\psi_2} \leq \xi$.
    % Thus, applying Lemma~\ref{lemma:truncation-bias} to~\eqref{eq:oracle-inequality-E2+E5-bound-1} yields
    % \begin{align}
    %     E_2 + E_5 \leq \frac{10\xi^2}{n}. \label{eq:oracle-inequality-E2+E5}
    % \end{align}

    \paragraph{Bounding $E_3$:}
    % \blue{For $(\bm z_1,\bm \omega_1),\ldots,(\bm z_n,\bm \omega_n) \in \mathbb{R}^d\times\mathcal{S}$, let $\mathbb{P}_n$ be their empirical measure.  Let 
    % \begin{align*}
    %     \mathcal{N}_{1,n} \coloneqq \sup_{(\bm u_i,\bm v_i)_{i=1}^n \in (\mathbb{R}^d\times\mathbb{R}^d)^n} \mathcal{N}\bigl((80B_n n)^{-1},\, T_{B_n}\circ\mathcal{F}(L,\bm p,s),\, \|\cdot\|_{L_1(\mathbb{P}_n)}\bigr).
    % \end{align*}}
    For $(\bm u_1,\bm v_1),\ldots,(\bm u_n,\bm v_n) \in \mathbb{R}^d\times\mathbb{R}^d$, let $\mathbb{P}_n$ be their empirical measure and let 
    \begin{align*}
        \mathcal{N}_{1,n} \coloneqq \sup_{(\bm u_i,\bm v_i)_{i=1}^n \in (\mathbb{R}^d\times\mathbb{R}^d)^n} \mathcal{N}\bigl((80B_n n)^{-1},\, T_{B_n}\circ\mathcal{F},\, \|\cdot\|_{L_1(\mathbb{P}_n)}\bigr).
    \end{align*}
    For $\delta \coloneqq \frac{24\cdot214 B_n^4 \log(14\mathcal{N}_{1,n})}{n} \geq 1/n$, we have
    \begin{align}
        E_3 &= \mathbb{E}\biggl[ \mathbb{E}\Bigl\{ \bigl(T_{B_n}\tilde{f}(\tilde{\bm Z}_0) - T_{B_n}Y_0\bigr)^2 - \bigl(f^{\star}_{B_n}(\tilde{\bm Z}_0) - T_{B_n}Y_0\bigr)^2 \,\Bigm|\, \mathcal{D} \Bigr\}\nonumber\\
        &\hspace{3.5cm} -\frac{2}{n} \sum_{i=1}^n \Bigl\{ \bigl(T_{B_n}\tilde{f}(\tilde{\bm Z}_i) - T_{B_n}Y_i\bigr)^2 - \bigl(f^{\star}_{B_n}(\tilde{\bm Z}_i) - T_{B_n}Y_i\bigr)^2 \Bigr\}\biggr]\nonumber\\
        &\leq \mathbb{E}\biggl( \sup_{g\in T_{B_n} \circ \mathcal{F}(L,\bm p,s)}  \biggl[ \mathbb{E}\Bigl\{ \bigl(g(\tilde{\bm Z}_0) - T_{B_n}Y_0\bigr)^2 - \bigl(f^{\star}_{B_n}(\tilde{\bm Z}_0) - T_{B_n}Y_0\bigr)^2  \Bigr\}\nonumber\\
        &\hspace{3.5cm} -\frac{2}{n} \sum_{i=1}^n \Bigl\{ \bigl(g(\tilde{\bm Z}_i) - T_{B_n}Y_i\bigr)^2 - \bigl(f^{\star}_{B_n}(\tilde{\bm Z}_i) - T_{B_n}Y_i\bigr)^2 \Bigr\} \biggr] \biggr) \nonumber\\
        &\leq \delta + \int_{\delta}^{\infty} \mathbb{P}\biggl( \sup_{g\in T_{B_n} \circ \mathcal{F}(L,\bm p,s)} \Bigl[ \mathbb{E}\bigl\{\bigl(g(\tilde{\bm Z}_0) - T_{B_n}Y_0\bigr)^2 - \bigl(f^{\star}_{B_n}(\tilde{\bm Z}_0) - T_{B_n}Y_0\bigr)^2\bigr\} \nonumber\\
        &\hspace{3.5cm} -\frac{2}{n} \sum_{i=1}^n \bigl\{ \bigl(g(\tilde{\bm Z}_i) - T_{B_n}Y_i\bigr)^2 - \bigl(f^{\star}_{B_n}(\tilde{\bm Z}_i) - T_{B_n}Y_i\bigr)^2 \bigr\} \Bigr] > t \biggr) \,\mathrm{d}t\nonumber\\
        \overset{(i)}&{\leq} \delta + 14\mathcal{N}_{1,n} \int_{\delta}^{\infty} \exp\Bigl( -\frac{nt}{24\cdot 214 B_n^4} \Bigr) \,\mathrm{d}t\nonumber\\
        &= \delta + 14\mathcal{N}_{1,n} \cdot \frac{24\cdot 214 B_n^4}{n}\exp\Bigl( -\frac{n\delta}{24\cdot 214 B_n^4} \Bigr)\nonumber\\
        &\leq \frac{24\cdot214 B_n^4 \log(14\mathcal{N}_{1,n})}{n} + \frac{24\cdot214B_n^4}{n}\nonumber\\
        \overset{(ii)}&{\lesssim} \frac{\xi^4 \log(e\xi) \log^3 n \cdot \bigl(sL\log(es) + s\log(ed)\bigr)}{n},
        \label{eq:oracle-inequality-E3}
    \end{align}
    were $(i)$ uses \citet[Theorem~11.4]{gyorfi2006distribution} with $\epsilon=1/2$ and $\alpha=\beta=t/2$ therein, and $(ii)$ uses Proposition~\ref{prop:VC-upper-bound}\emph{(b)}. 

    \paragraph{Bounding $E_4$:} The functions in $\mathcal{F}$ are parametrised by a subset $A$ of a Euclidean space, meaning that we can write $\mathcal{F} = \{f_{\bm\theta}: \bm\theta \in A\}$.  Since $A$ has a countable, dense subset $\tilde A$ \citep[e.g.][Proposition~9.26]{royden2010real} and the map $\bm\theta \mapsto \hat{R}_n(f_{\bm\theta})$ is continuous, it follows that $\inf_{\bm \theta \in A} \hat{R}_n(f_{\bm \theta}) = \inf_{\bm \theta \in \tilde{A}} \hat{R}_n(f_{\bm \theta})$.  Thus $\inf_{f \in \mathcal{F}} \hat{R}_n(f)$ is measurable, and
    \begin{align}
        E_4 &\leq \mathbb{E}\biggl[ \frac{2}{n} \sum_{i=1}^n \Bigl\{ \bigl(\tilde{f}(\tilde{\bm Z}_i) - Y_i\bigr)^2 - \bigl(f^{\star}(\tilde{\bm Z}_i) - Y_i\bigr)^2 \Bigr\} \biggr] \nonumber\\
        &= 2\mathbb{E}\Bigl\{\hat R_n(\tilde f) - \inf_{f\in\mathcal{F}}\hat{R}_n(f) \Bigr\} + 2\mathbb{E} \Bigl\{\inf_{f\in\mathcal{F}}\hat{R}_n(f) - \hat{R}_n(f^{\star}) \Bigr\}\nonumber\\
        &\leq 2\mathbb{E}\Bigl\{\hat R_n(\tilde f) - \inf_{f\in\mathcal{F}}\hat{R}_n(f) \Bigr\} + 2\inf_{f\in\mathcal{F}} \mathbb{E} \bigl\{\hat{R}_n(f) - \hat{R}_n(f^{\star}) \bigr\}\nonumber\\
        &= 2\mathbb{E}\Bigl\{\hat R_n(\tilde f) - \inf_{f\in\mathcal{F}}\hat{R}_n(f) \Bigr\} + 2\inf_{f\in\mathcal{F}} \mathbb{E} \bigl\{\bigl(f(\tilde{\bm Z_0}) - Y_0\bigr)^2 - \bigl(f^{\star}(\tilde{\bm Z_0}) - Y_0\bigr)^2 \bigr\}\nonumber\\
        &= 2\mathbb{E}\Bigl\{\hat R_n(\tilde f) - \inf_{f\in\mathcal{F}}\hat{R}_n(f) \Bigr\} + 2\inf_{f\in\mathcal{F}} \mathbb{E} \bigl\{\bigl(f(\tilde{\bm Z_0}) - f^{\star}(\tilde{\bm Z_0})\bigr)^2 \bigr\}. \label{eq:oracle-inequality-E4}
    \end{align}
    Combining \eqref{eq:oracle-inequality-E1}, \eqref{eq:oracle-inequality-E2+E5}, \eqref{eq:oracle-inequality-E3} and \eqref{eq:oracle-inequality-E4} yields the final result.
\end{proof}

\section{\texorpdfstring{$\mathcal{F}$-separability of partitions of $\mathcal{S}$}{F-separability of partitions of S}}

\begin{prop} \label{prop:F(L,p,s)-separable}
    Let $\mathcal{S} \subseteq \{0,1\}^d$ be such that $|\mathcal{S}| \geq 2$, let $K\in\{2,\ldots,|\mathcal{S}|\}$ and let $\{\mathcal{S}_1, \ldots, \mathcal{S}_K\}$ be a partition of $\mathcal{S}$.
    \begin{itemize}
        \item[(a)] Suppose that there exists $\mathcal{J} \subseteq [d]$ such that $\bm\omega, \bm\omega' \in \mathcal{S}$ belong to the same cell of the partition whenever $\bm\omega_{\mathcal{J}} = \bm\omega'_{\mathcal{J}}$. Let $\mathcal{S}_{\mathcal{J}} \coloneqq \{\bm\omega_{\mathcal{J}} : \bm\omega \in \mathcal{S}\} \subseteq \{0,1\}^{|\mathcal{J}|}$. Define 
        \begin{gather*}
            p_* \coloneqq 2\bigl\lceil |\mathcal{S}_{\mathcal{J}}|^{1/2} \bigr\rceil \quad\text{and}\quad
            \bm p \coloneqq (d,p_*,p_*,1)^\top \in \mathbb{N}^{4}.
        \end{gather*}
        Then $\{\mathcal{S}_1,\ldots,\mathcal{S}_K\}$ is $\mathcal{F}(2,\bm p)$-separable. In particular, if $\mathcal{S}_1,\ldots,\mathcal{S}_K$ is an arbitrary partition of $\mathcal{S}$, then we may take $\mathcal{J} = [d]$.

        \item[(b)] Suppose that for each $k\in[K]$, there exist $P_k \in \mathbb{N}$ and $(\bm v^{(k)}_{\ell}, b^{(k)}_{\ell})_{\ell \in [P_k]} \in (\mathbb{R}^d \times \mathbb{R})^{P_k}$ such that $\mathcal{S}_k = \bigl\{\bm\omega\in\mathcal{S}:\bm\omega^\top \bm v^{(k)}_{\ell} \leq b^{(k)}_{\ell} \text{ for all } \ell\in[P_k]\bigr\}$. Define
        \begin{gather*}
            \bm p \coloneqq \biggl(d,\, 2\sum_{k=1}^K P_k,\, K, 1\biggr)^\top \in \mathbb{N}^4.
        \end{gather*}
        Then $\{\mathcal{S}_1,\ldots,\mathcal{S}_K\}$ is $\mathcal{F}(2,\bm p)$-separable.
    \end{itemize}
\end{prop}
\begin{proof}
    \emph{(a)} Let $J \coloneqq |\mathcal{J}|$. Without loss of generality, suppose that $\mathcal{J} = [J]$, and let $\mathcal{S}_{\mathcal{J}} = \{\bm u_1,\ldots,\bm u_{|\mathcal{S}_{\mathcal{J}}|}\}\subseteq \{0,1\}^{J}$. For $i\in[|\mathcal{S}_{\mathcal{J}}|]$, define $x_i \coloneqq \sum_{j=1}^{J} 2^{j-1} u_{ij} \in \mathbb{N}_0$, where $u_{ij}$ is the $j$th coordinate of $\bm u_i$. Further define $y_i \coloneqq k/K$ where $k\in[K]$ is such that there exists $\bm\omega \in \mathcal{S}_{k}$ such that $\bm\omega_{\mathcal{J}} = \bm u_i$; note that $k$ is uniquely defined by our assumption. By \citet[Lemma~2.2]{shen2019nonlinear} and~\ref{NN-enlarging}, there exists $f_1 \in \mathcal{F}\bigl(2,(1,p_*,p_*,1)\bigr)$ such that $f_1(x_i) = y_i$ for all $i\in[|\mathcal{S}_{\mathcal{J}}|]$. %Moreover, by \citet[Proposition~2.2]{shen2020deep}, there exists $f_2 \in \mathcal{F}\bigl(M+2, (1,4N+4,\ldots,4N+4,1) \bigr)$ such that $f_2 = f_1$.  
    Thus, by~\ref{NN-composition}, there exists $f \in\mathcal{F}(2,\bm p)$ such that
    \begin{align*}
        f(\bm v) = f_1\biggl(\sum_{j=1}^J 2^{j-1} v_j\biggr) 
    \end{align*}
    for $\bm v = (v_1,\ldots,v_d)^\top \in \mathbb{R}^d$.  For any $k\in[K]$ and $\bm\omega = (\omega_1,\ldots,\omega_d)^\top \in \mathcal{S}_k$, there exists $i \in |\mathcal{S}_{\mathcal{J}}|$ such that $\bm\omega_{\mathcal{J}} = \bm u_i$. Therefore,
    \begin{align*}
        f(\bm\omega) = f_1\biggl(\sum_{j=1}^J 2^{j-1} \omega_j\biggr) = f_1\biggl(\sum_{j=1}^J 2^{j-1} u_{ij}\biggr) = f_1(x_i) = \frac{k}{K}.
    \end{align*}
    Hence, $\mathcal{S}_1,\ldots,\mathcal{S}_K$ are $\mathcal{F}(2,\bm p)$-separable with $\epsilon = 1/(2K)$.

    \medskip \emph{(b)} First define
    \begin{align*}
        \epsilon \coloneqq \min_{k\in[K]} \min_{\bm\omega \in \mathcal{S} \setminus \mathcal{S}_k} \max_{\ell\in[P_k]} \;\bigl(\bm\omega^\top \bm v^{(k)}_{\ell} - b^{(k)}_{\ell}\bigr).
    \end{align*}
    For each $k\in[K]$ and $\bm\omega \in \mathcal{S} \setminus \mathcal{S}_k$, we have $\max_{\ell\in[P_k]} \;\bigl(\bm\omega^\top \bm v^{(k)}_{\ell} - b^{(k)}_{\ell}\bigr) > 0$ by assumption, so $\epsilon > 0$.  For $k\in[K]$ and $\ell\in[P_k]$, define $\phi_{(k,\ell)} \in \mathcal{F}\bigl(1,(d,2,1)\bigr)$ by 
    \begin{align*}
        \phi_{(k,\ell)}(\bm x) \coloneqq \sigma\biggl(-\frac{\bm x^\top \bm v^{(k)}_{\ell} - b^{(k)}_{\ell}}{\epsilon} + 1\biggr) - \sigma\biggl(-\frac{\bm x^\top \bm v^{(k)}_{\ell} - b^{(k)}_{\ell}}{\epsilon}\biggr)
    \end{align*}
    for $\bm x\in\mathbb{R}^d$. If $\bm x^\top \bm v^{(k)}_{\ell} - b^{(k)}_{\ell} \leq 0$, then $\phi_{(k,\ell)}(\bm x) = 1$; if $\bm x^\top \bm v^{(k)}_{\ell} - b^{(k)}_{\ell} \geq \epsilon$, then $\phi_{(k,\ell)}(\bm x) = 0$. By~\ref{NN-parallelisation}, the function $\bm\phi: \mathbb{R}^d \to \mathbb{R}^{\sum_{k=1}^K P_k}$ defined by $\bm \phi(\bm x) \coloneqq \bigl(\phi_{(k,\ell)}(\bm x)\bigr)_{k\in[K],\ell\in[P_k]}$ belongs to $\mathcal{F}\bigl(1,(d,\, 2\sum_{k=1}^K P_k,\, \sum_{k=1}^K P_k)\bigr)$. For $k\in[K]$ and $\bm\omega \in \mathcal{S}_k$, we have $\sum_{\ell=1}^{P_k} \phi_{(k,\ell)}(\bm\omega) = P_k$, and for all $k'\neq k$, we have $\sum_{\ell'=1}^{P_{k'}} \phi_{(k',\ell')}(\bm\omega) \leq P_{k'}-1$ since there exists $\ell^* \in [P_{k'}]$ such that $\bm\omega^\top \bm v^{(k')}_{\ell^*} - b^{(k')}_{\ell^*} \geq \epsilon$. Next define $\psi \in \mathcal{F}\bigl(1,(\sum_{k=1}^K P_k,\, K,\, 1)\bigr)$ by
    \begin{align*}
        \psi(\bm u) \coloneqq \sum_{k=1}^K k \cdot\sigma\biggl(\sum_{\ell=1}^{P_k} u_{k,\ell} - P_k + 1\biggr)
    \end{align*}
    for $\bm u = (u_{k,\ell})_{k\in[K], \ell\in[P_k]} \in \mathbb{R}^{\sum_{k=1}^K P_k}$. Then, for $k\in[K]$ and $\bm\omega \in \mathcal{S}_k$, we have $\psi\bigl(\bm\phi(\bm\omega)\bigr) = k$. Finally, $\psi\circ\bm\phi \in \mathcal{F}(2,\bm p)$ by~\ref{NN-composition}, so $\{\mathcal{S}_1,\ldots,\mathcal{S}_K\}$ is $\mathcal{F}(2,\bm p)$-separable.
\end{proof}

\section{Approximation theory for neural networks}
The following lemma follows the arguments of \citet[Theorem~2.2]{lu2021deep} and \citet[Theorem~3.3]{jiao2023deep}, with very minor changes.
\begin{lemma} \label{lemma:approximate-holder-functions-[0,1]^d}
    Let $\beta,\gamma>0$, $\beta_0 \coloneqq \lceil\beta\rceil - 1$, $d\in\mathbb{N}$ and $g \in \mathcal{H}_d^{\beta}\bigl([0,1]^d,\gamma\bigr)$. For any $M,N \in \mathbb{N}$, let
    \begin{gather*}
        R \coloneqq \lfloor N^{1/d} \rfloor^2 \lfloor M^{2/d}\rfloor,\quad L\coloneqq 12(\beta_0+1)^2(M+2)\lceil \log_2(4M) \rceil,\\
        p_* \coloneqq 30(\beta_0+1)^2 d^{\beta_0+1} (N+1) \lceil\log_2(8N)\rceil \quad\text{and}\quad \bm p \coloneqq (d,p_*,\ldots,p_*,1)^\top \in \mathbb{N}^{L+2}.
    \end{gather*}
    For $\delta > 0$, let 
    \begin{align*}
        \Gamma\bigl( [0,1]^d, R, \delta \bigr) \coloneqq \bigcup_{j=1}^d \biggl\{ \bm x = (x_1,\ldots,x_d)^\top &\in [0,1]^d : x_j \in \bigcup_{r=1}^{R-1} \biggl( \frac{r}{R} - \delta,\; \frac{r}{R} \biggr) \biggr\}.
    \end{align*}
    Then, for any $\delta \in \bigl(0, \frac{1}{3R}\bigr]$, there exists $f\in\mathcal{F}(L,\bm p)$ such that $|f(\bm x)| \leq \gamma$ for all $\bm x \in \mathbb{R}^d$, and
    \begin{align*}
        |f(\bm x) - g(\bm x)| \leq 9\gamma (\beta_0+3)^2 8^{\beta} d^{\beta_0+\beta/2} (NM)^{-2\beta/d},
    \end{align*}
    for all $\bm x \in [0,1]^d \setminus \Gamma\bigl( [0,1]^d, R, \delta \bigr)$.    
\end{lemma}
\begin{proof}
    \textbf{Step 1 (discretisation):}
    By \citet[Proposition~4.3]{lu2021deep}, there exists $\phi_1 \in \mathcal{F}\bigl(4M+5, (1, p_{*1}, \ldots, p_{*1}, 1) \bigr)$, where $p_{*1} \coloneqq 4\lfloor N^{1/d} \rfloor +3$ such that
    \begin{align*}
        \phi_1(x) = \frac{r}{R} \quad\text{if } x\in \biggl[ \frac{r}{R},\; \frac{r+1}{R} - \delta \cdot \mathbbm{1}_{\{r \leq R-2\}} \biggr] \text{ for } r\in \{0,1,\ldots,R-1\}.
    \end{align*}
    For $\bm v = (v_1,\ldots,v_d)^\top \in \{0,1,\ldots,R-1\}^d$, define
    \begin{align*}
        Q_{\bm v} \coloneqq \biggl\{\bm x = (x_1,\ldots,x_d)^\top \in [0,1]^d : x_j \in \biggl[ \frac{v_j}{R},\; \frac{v_j+1}{R} - \delta\cdot \mathbbm{1}_{\{v_j \leq R-2\}} \biggr] \text{ for all } j\in[d]\biggr\},
    \end{align*}
    so that $[0,1]^d \setminus \Gamma\bigl( [0,1]^d, R, \delta \bigr) = \bigcup_{\bm v \in \{0,\ldots,R-1\}^d} Q_{\bm v}$. By Property~\ref{NN-parallelisation}, the function $\bm \phi_{\mathrm{dsc}}:\mathbb{R}^d \rightarrow \mathbb{R}^d$ given by $\bm\phi_{\mathrm{dsc}}(\bm x) \coloneqq \bigl(\phi_1(x_1), \ldots, \phi_1(x_d)\bigr)^\top$ for $\bm{x} = (x_1,\ldots,x_d)^\top$ belongs to $\mathcal{F}\bigl(4M+5, (d, dp_{*1}, \ldots, dp_{*1}, d) \bigr)$.  Thus,
    \begin{align*}
        \bm\phi_{\mathrm{dsc}}(\bm x) = \biggl( \frac{v_1}{R}, \ldots, \frac{v_d}{R} \biggr)^\top \quad\text{for all } \bm x \in Q_{\bm v},\, \bm v \in \{0,1,\ldots,R-1\}^d. 
    \end{align*}

    \medskip \textbf{Step 2 (approximation of Taylor coefficients):} The function $\phi_2:\mathbb{R}^d \rightarrow \mathbb{R}$ given by $\phi_2(\bm x) \coloneqq \sum_{j=1}^d R^j \phi_1(x_j)$ is a composition of $\bm \phi_{\mathrm{dsc}}$ with a projection along the vector $(R,R^2,\ldots,R^d)^\top$ and hence by~\ref{NN-composition}, $\phi_2 \in \mathcal{F}\bigl(4M+5, (d, dp_{*1}, \ldots, dp_{*1}, 1) \bigr)$.  Moreover, for all $\bm v \in \{0,1,\ldots,R-1\}^d$ and $\bm x \in Q_{\bm v}$, we have
    \begin{align}
        \phi_2(\bm x) = \sum_{j=1}^d R^{j-1}v_j \eqqcolon I_{\bm v} \in \{0,1,\ldots,R^d-1\}. \label{eq:phi_3-property}
    \end{align}
    For each $\bm\alpha \in \mathbb{N}_0^d$ with $\|\bm\alpha\|_1 \leq \beta_0$ we have $\frac{\partial^{\bm\alpha} g(\bm x) + \gamma}{2\gamma} \in [0,1]$ for $\bm x \in [0,1]^d$.  Hence, since $R^d \leq N^2M^2$ and since $\bm{v} \mapsto I_{\bm{v}}$ is a bijection, by \citet[Proposition~4.4]{lu2021deep} there exists $\psi_{\bm\alpha}^\circ \in \mathcal{F}\bigl(5(M+2)\lceil\log_2(4M)\rceil, (1,p_{*2},\ldots,p_{*2},1)\bigr)$, where $p_{*2} \coloneqq 16(\beta_0+1)(N+1) \lceil\log_2(8N)\rceil$, such that
    \begin{align}
        \biggl| \psi_{\bm\alpha}^\circ(I_{\bm v}) - \frac{\partial^{\bm\alpha} g(\bm v/R) + \gamma}{2\gamma} \biggr| \leq (NM)^{-2(\beta_0+1)} \quad\text{for all } \bm v \in \{0,1,\ldots,R-1\}^d. \label{eq:psi^circ_alpha-property}
    \end{align}
    Now
    \[
    dp_{*1} \vee p_{*2} \leq 16d(\beta_0+1)(N+1)\lceil\log_2(8N)\rceil \eqqcolon p_{*3}
    \]
    and 
    \[
    4M+5+5(M+2)\lceil\log_2(4M)\rceil \leq 7(M+2)\lceil\log_2(4M)\rceil \eqqcolon L_0
    \]
    Hence, by Properties~\ref{NN-enlarging}, \ref{NN-composition} and~\ref{NN-padding}, there exists $\psi_{\bm\alpha} \in \mathcal{F}\bigl(L_0,\, (d,p_{*3},\ldots,p_{*3},1)\bigr)$ such that $\psi_{\bm\alpha}(\bm x) = 2\psi_{\bm\alpha}^\circ\bigl(\phi_2(\bm x)\bigr) - 1 \in [-1,1]$ for all $\bm x \in [0,1]^d$.  By~\eqref{eq:phi_3-property} and~\eqref{eq:psi^circ_alpha-property}, we deduce that for all $\bm v \in \{0,1,\ldots,R-1\}^d$ and $\bm x \in Q_{\bm v}$, we have
    \begin{align}
        \biggl| \psi_{\bm\alpha}(\bm x) - \frac{1}{\gamma} \cdot \partial^{\bm\alpha} g\bigl(\bm{\phi}_{\mathrm{dsc}}(\bm{x})\bigr) \biggr| \leq 2(NM)^{-2(\beta_0+1)}. \label{eq:psi_alpha-property}
    \end{align}

    \medskip \textbf{Step 3 (local Taylor expansion):} By \citet[Lemma~4.2]{lu2021deep}, there exists $\phi_{\times} \in \mathcal{F}\bigl(2(\beta_0+1)(M+1),\, (1,9(N+1)+1,\ldots,9(N+1)+1,1)\bigr)$ such that for all $x,y\in[-1,1.1]$, we have
    \begin{align}
        |\phi_{\times}(x,y) - xy| \leq 27(N+1)^{-2(\beta_0+1)(M+1)}. \label{eq:phi_times-property}
    \end{align}
    Moreover, by \citet[Proposition~4.1]{lu2021deep}, for $\bm\alpha\in\mathbb{N}_0^d$ with $\|\bm\alpha\|_1 \leq \beta_0$, there exists $\mathrm{Poly}_{\bm\alpha} \in \mathcal{F}\bigl(7(\beta_0+1)^2M,\, (d,9(N+1)+\beta_0,\ldots,9(N+1)+\beta_0,1)\bigr)$ such that
    \begin{align}
        |\mathrm{Poly}_{\bm\alpha}(\bm x) - \bm x^{\bm\alpha}| \leq 9(\beta_0+1)(N+1)^{-7(\beta_0+1)M}, \label{eq:poly_alpha-property}
    \end{align}
    for all $\bm x \in [0,1]^d$.
    Moreover, since $9(\beta_0+1)(N+1)^{-7(\beta_0+1)M} \leq 9(\beta_0+1)2^{-7(\beta_0+1)} < 0.1$, we have $\mathrm{Poly}_{\bm\alpha}(\bm x) \in [-1,1.1]$ for all $\bm x\in[0,1]^d$. Now define $f^\circ:[0,1]^d \rightarrow \mathbb{R}$ by  
    \begin{align*}
        f^\circ(\bm x) \coloneqq \sum_{\bm\alpha \in \mathbb{N}_0^d : \|\bm\alpha\|_1 \leq \beta_0} \frac{\gamma}{\bm\alpha!} \cdot \phi_{\times} \Bigl(\psi_{\bm\alpha}(\bm x),\, \mathrm{Poly}_{\bm\alpha}\bigl(\bm x - \bm\phi_{\mathrm{dsc}}(\bm x) \bigr)\Bigr).
    \end{align*}
    Since $\bigl|\{\bm\alpha \in \mathbb{N}_0^d : \|\bm\alpha\|_1 \leq \beta_0\}\bigr| = \binom{d+\beta_0}{\beta_0} \leq (\beta_0+1)d^{\beta_0}$, we have by Properties~\ref{NN-enlarging}--\ref{NN-parallelisation} that $f^\circ \in \mathcal{F}(L-1,\bm p^\circ)$ where $\bm p^\circ \coloneqq (d,p_*,\ldots,p_*,1)^\top \in \mathbb{N}^{L+1}$. Moreover, for $\bm v \in \{0,\ldots,R-1\}^d$ and $\bm x \in Q_{\bm v}$, we have
    \begin{align}
        &\biggl|f^\circ(\bm x) - \sum_{\bm\alpha \in \mathbb{N}_0^d : \|\bm\alpha\|_1 \leq \beta_0} \frac{\partial^{\bm\alpha}g\bigl(\bm\phi_{\mathrm{dsc}}(\bm x)\bigr)}{\bm\alpha!} \bigl(\bm x - \bm\phi_{\mathrm{dsc}}(\bm x)\bigr)^{\bm\alpha} \biggr| \nonumber\\
        &\leq \sum_{\bm\alpha \in \mathbb{N}_0^d : \|\bm\alpha\|_1 \leq \beta_0} \frac{\gamma}{\bm\alpha!} \cdot \biggl| \phi_{\times} \Bigl(\psi_{\bm\alpha}(\bm x),\, \mathrm{Poly}_{\bm\alpha}\bigl(\bm x - \bm\phi_{\mathrm{dsc}}(\bm x) \bigr)\Bigr) - \psi_{\bm\alpha}(\bm x) \cdot \mathrm{Poly}_{\bm\alpha}\bigl(\bm x - \bm\phi_{\mathrm{dsc}}(\bm x) \bigr) \biggr| \nonumber\\
        &\qquad + \sum_{\bm\alpha \in \mathbb{N}_0^d : \|\bm\alpha\|_1 \leq \beta_0} \frac{\gamma}{\bm\alpha!} \cdot \biggl| \psi_{\bm\alpha}(\bm x) - \frac{1}{\gamma}\cdot \partial^{\bm\alpha}g\bigl(\bm\phi_{\mathrm{dsc}}(\bm x)\bigr) \biggr| \cdot \mathrm{Poly}_{\bm\alpha}\bigl(\bm x - \bm\phi_{\mathrm{dsc}}(\bm x) \bigr) \nonumber\\
        &\qquad + \sum_{\bm\alpha \in \mathbb{N}_0^d : \|\bm\alpha\|_1 \leq \beta_0} \frac{\partial^{\bm\alpha}g\bigl(\bm\phi_{\mathrm{dsc}}(\bm x)\bigr)}{\bm\alpha!} \Bigl| \mathrm{Poly}_{\bm\alpha}\bigl(\bm x - \bm\phi_{\mathrm{dsc}}(\bm x) \bigr) - \bigl(\bm x - \bm\phi_{\mathrm{dsc}}(\bm x) \bigr)^{\bm\alpha} \Bigr| \nonumber\\
        &\leq \gamma(\beta_0+1)d^{\beta_0}\bigl\{ 27 (N+1)^{-2(\beta_0+1)(M+1)} + 3(NM)^{-2(\beta_0+1)} + 9(\beta_0+1)(N+1)^{-7(\beta_0+1)M} \bigr\} \nonumber\\
        &\leq \gamma(\beta_0+1)(9\beta_0+39)d^{\beta_0}(NM)^{-2(\beta_0+1)}, \label{eq:f-property}
    \end{align}
    where the second inequality follows from~\eqref{eq:phi_times-property}, \eqref{eq:psi_alpha-property} and \eqref{eq:poly_alpha-property}. Therefore, we deduce that, for $\bm v \in \{0,\ldots,R-1\}^d$ and $\bm x \in Q_{\bm v}$,
    \begin{align*}
        |f^\circ(\bm x) - g(\bm x)| &\leq \biggl|f^\circ(\bm x) - \sum_{\bm\alpha \in \mathbb{N}_0^d : \|\bm\alpha\|_1 \leq \beta_0} \frac{\partial^{\bm\alpha}g\bigl(\bm\phi_{\mathrm{dsc}}(\bm x)\bigr)}{\bm\alpha!} \bigl(\bm x - \bm\phi_{\mathrm{dsc}}(\bm x)\bigr)^{\bm\alpha} \biggr| \nonumber\\
        &\qquad + \biggl|g(\bm x) - \sum_{\bm\alpha \in \mathbb{N}_0^d : \|\bm\alpha\|_1 \leq \beta_0} \frac{\partial^{\bm\alpha}g\bigl(\bm\phi_{\mathrm{dsc}}(\bm x)\bigr)}{\bm\alpha!} \bigl(\bm x - \bm\phi_{\mathrm{dsc}}(\bm x)\bigr)^{\bm\alpha} \biggr| \nonumber\\
        &\leq \gamma(\beta_0+1)(9\beta_0+39)d^{\beta_0}(NM)^{-2(\beta_0+1)} + \gamma d^{\beta_0}(\sqrt{d}/R)^{\beta} \nonumber\\
        &\leq \gamma(\beta_0+1)(9\beta_0+39)d^{\beta_0}(NM)^{-2(\beta_0+1)} + \gamma d^{\beta_0+\beta/2} 8^{\beta}(NM)^{-2\beta/d} \nonumber\\
        &\leq 9\gamma (\beta_0+3)^2 8^{\beta} d^{\beta_0+\beta/2} (NM)^{-2\beta/d}, 
    \end{align*}
    where the second inequality follows from~\eqref{eq:f-property} and Lemma~\ref{lemma:taylor-approximation-error}, and the third inequality follows since $R \geq \frac{1}{8}(NM)^{2/d}$. Finally, we define $f:\mathbb{R}^d \rightarrow [-\gamma,\gamma]$ by
    \begin{align*}
        f(\bm x) \coloneqq \sigma\bigl(f^\circ(\bm x) + \gamma\bigr) - \sigma\bigl(f^\circ(\bm x) - \gamma\bigr) - \gamma = T_{\gamma}f^\circ(\bm x).
    \end{align*}
    By~\ref{NN-enlarging} and~\ref{NN-composition}, we have $f\in\mathcal{F}(L,\bm p)$. Moreover, since $g(\bm x) \in [-\gamma,\gamma]$ for all $\bm x \in [0,1]^d$, we have $|f(\bm x) - g(\bm x)| \leq |f^\circ(\bm x) - g(\bm x)| \leq 9\gamma (\beta_0+3)^2 8^{\beta} d^{\beta_0+\beta/2} (NM)^{-2\beta/d}$ for $\bm v \in \{0,\ldots,R-1\}^d$ and $\bm x \in Q_{\bm v}$.
\end{proof}

Our next lemma quantifies the extent to which a H\"older function on a bounded hypercube that depends only on a subset of the variables can be approximated by a neural network once we excise a finite set of strips in each coordinate.
\begin{lemma} \label{lemma:approximate-Holder-functions-[a,c]^d}
    Let $-\infty < a < c < \infty$, $\beta,\gamma>0$, $\beta_0 \coloneqq \lceil \beta \rceil - 1$, $d\in\mathbb{N}$, $t\in[d]\cup\{0\}$ and $g\in\mathcal{H}_t^{\beta}\bigl([a,c]^d,\gamma\bigr)$. For any $M,N\in\mathbb{N}$, let 
    \begin{gather*}
        R\coloneqq \lfloor N^{1/t} \rfloor^2 \lfloor M^{2/t} \rfloor,\quad L\coloneqq 12(\beta_0+1)^2(M+2)\lceil \log_2(4M) \rceil,\\
        p_* \coloneqq 30(\beta_0+1)^2 t^{\beta_0+1} (N+1) \lceil\log_2(8N)\rceil \vee 1 \quad\text{and}\quad \bm p \coloneqq (d,p_*,\ldots,p_*,1)^\top \in \mathbb{N}^{L+2}.
    \end{gather*}
    Then, for any $\delta \in \bigl(0, \frac{1}{3R}\bigr]$, there exists $f\in\mathcal{F}(L,\bm p)$ such that $|f(\bm x)| \leq \gamma \vee (c-a)^{\beta_0}\gamma$ for all $\bm x \in \mathbb{R}^d$, and
    \begin{align*}
        |f(\bm x) - g(\bm x)| \leq 9\bigl(1 \vee (c-a)^{\beta_0}\bigr)\gamma (\beta_0+3)^2 8^{\beta} t^{\beta_0+\beta/2} (NM)^{-2\beta/t},
    \end{align*}
    for all $\bm x \in [a,c]^d \setminus \Gamma\bigl( [a,c]^d, R, \delta \bigr)$, where
    \begin{align*}
        \Gamma\bigl( [a,c]^d, R, \delta \bigr) \coloneqq \bigcup_{j=1}^d \biggl\{ \bm x = (x_1,\ldots,x_d)^\top &\in [a,c]^d : \frac{x_j-a}{c-a} \in \bigcup_{r=1}^{R-1} \biggl( \frac{r}{R} - \delta,\; \frac{r}{R} \biggr) \biggr\}
    \end{align*}
    if $t \in [d]$, and $\Gamma\bigl( [a,c]^d, R, \delta \bigr) \coloneqq \emptyset$ if $t=0$.
\end{lemma}

\begin{proof}
    In the case where $t=0$, the function $g$ is constant, so belongs to $\mathcal{F}(L,\bm p)$ and we have zero approximation error.
    We next consider the case where $t=d$.
    Define $h:[0,1]^d \to \mathbb{R}$ by
    \begin{align*}
        h(\bm x) \coloneqq g\bigl(a\bm{1}_d + (c-a)\bm x  \bigr),
    \end{align*}
    so that $h\in\mathcal{H}_d^{\beta}\bigl([0,1]^d, \gamma \vee (c-a)^{\beta_0}\gamma\bigr)$. By Lemma~\ref{lemma:approximate-holder-functions-[0,1]^d}, there exists $f^\circ \in\mathcal{F}(L,\bm p)$ such that $|f^\circ(\bm x)| \leq \gamma \vee (c-a)^{\beta_0}\gamma$ for all $\bm x \in \mathbb{R}^d$, and
    \begin{align}
        |f^\circ(\bm x) - h(\bm x)| \leq 9 \bigl(\gamma \vee (c-a)^{\beta_0}\gamma\bigr) (\beta_0+3)^2 8^{\beta} t^{\beta_0+\beta/2} (NM)^{-2\beta/t} \label{eq:f^circ-property}
    \end{align}
    for all $\bm x \in [0,1]^d \setminus \Gamma \bigl([0,1]^d,R,\delta\bigr)$. 
    By~\ref{NN-composition}, there exists $f \in \mathcal{F}(L,\bm p)$ such that 
    \begin{align*}
        f(\bm x) = f^\circ\biggl(\frac{\bm x - a\bm{1}_d}{c-a}\biggr),
    \end{align*}
    and $|f(\bm x)| \leq \gamma \vee (c-a)^{\beta_0}\gamma$ for all $\bm x \in \mathbb{R}^d$.
    Thus, by~\eqref{eq:f^circ-property},
    \begin{align*}
        |f(\bm x) - g(\bm x)| \leq 9 \bigl(\gamma \vee (c-a)^{\beta_0}\gamma\bigr) (\beta_0+3)^2 8^{\beta} t^{\beta_0+\beta/2} (NM)^{-2\beta/t}
    \end{align*}
    for all $\bm x \in [a,c]^d \setminus \Gamma\bigl([a,c]^d, R, \delta\bigr)$. This proves the claim when $t=d$.

    \medskip Now assume that $t\in[d-1]$. Without loss of generality, assume that $g$ depends only on the coordinates in $[t]$. Define $g^{[t]} : [a,c]^t \to \mathbb{R}$ by $g^{[t]}(\bm y) \coloneqq g(\bm y, \bm 0_{d-t})$, so that $g^{[t]} \in \mathcal{H}_t^{\beta}\bigl([a,c]^t, \gamma\bigr)$. Then, by the case where $t=d$, there exists $f^{[t]} \in \mathcal{F}(L, \bm p^{[t]})$ where $\bm p^{[t]} \coloneqq (t,p_*,\ldots,p_*,1)^\top \in \mathbb{N}^{L+2}$ such that $|f^{[t]}(\bm y)| \leq \gamma \vee (c-a)^{\beta_0}\gamma$ for all $\bm y\in \mathbb{R}^t$, and
    \begin{align*}
        \bigl|f^{[t]}(\bm y) - g^{[t]}(\bm y)\bigr| \leq 9 \bigl(\gamma \vee (c-a)^{\beta_0}\gamma\bigr) (\beta_0+3)^2 8^{\beta} t^{\beta_0+\beta/2} (NM)^{-2\beta/t}
    \end{align*}
    for all $\bm y \in [a,c]^t \setminus \Gamma\bigl([a,c]^t, R, \delta\bigr)$. By~\ref{NN-composition}, we can define $f\in\mathcal{F}(L,\bm p)$ by $f(\bm x) \coloneqq f^{[t]}(\bm M \bm x)$, where $\bm M \coloneqq \begin{pmatrix}
        \bm I_t & \bm 0
    \end{pmatrix} \in \mathbb{R}^{t \times d}$, and $f$ satisfies the requirements in the statement.
\end{proof}

By Assumption~\ref{assumption:composition-of-smooth-functions}, for each $k\in[K]$, there exist $\bm g^{(k)}_{1}, \ldots, \bm g^{(k)}_{q_k}$ such that
\begin{gather*}
    \bm g^{(k)}_{r}=\bigl(g^{(k)}_{r,1}, \ldots, g^{(k)}_{r,d_{r+1}^{(k)}}\bigr)^\top : \mathbb{R}^{d_{r}^{(k)}} \to \mathbb{R}^{d_{r+1}^{(k)}}, \nonumber\\
    g_{r,j}^{(k)} \in \mathcal{H}_{t_{r}^{(k)}}^{\beta_{r}^{(k)}}\bigl(\mathbb{R}^{d^{(k)}_r}, \gamma^{(k)}_r\bigr) \text{ for all } r\in[q_k], j\in[d_{r+1}^{(k)}], 
\end{gather*}
and
\begin{align*}
    f^{\mathcal{S}_k}(\bm z) = \bm g^{(k)}_{q_k} \circ \bm g^{(k)}_{q_k-1} \circ \cdots \circ \bm g^{(k)}_{1}(\bm z)
\end{align*}
for all $\bm z\in \mathbb{R}^d$. Now, for $k \in [K]$, define 
\begin{gather}
    a_1^{(k)} \coloneqq -\xi_1\log(2dn),\quad c_1^{(k)} \coloneqq \xi_1\log(2dn)\nonumber\\
    \text{and } a_r^{(k)} \coloneqq -\gamma_{r-1}^{(k)},\quad c_r^{(k)} \coloneqq \gamma_{r-1}^{(k)} \quad\text{for } r\in\{2,\ldots,q_k+1\}. \label{eq:defn-a_r^(k)-c_r^(k)}
\end{gather}
For $r\in\{2,\ldots,q_k+1\}$, we have $\bm g_{r-1}^{(k)}(\bm{z}) \in [a_r^{(k)},\, c_r^{(k)}]^{d_r^{(k)}}$ for all $\bm{z} \in \mathbb{R}^{d^{(k)}_{r-1}}$ by the H\"older property of these functions. Thus, it is sufficient to restrict the domain of $\bm g_r^{(k)}$ to $[a_r^{(k)},\, c_r^{(k)}]^{d_r^{(k)}}$, for $r\in\{2,\ldots,q_k\}$. Then, Assumption~\ref{assumption:composition-of-smooth-functions} yields that
\begin{align*}
    f^{\mathcal{S}_k}(\bm z) = \bm g^{(k)}_{q_k} \circ \bm g^{(k)}_{q_k-1} \cdots \circ \bm g^{(k)}_{1}(\bm z)
\end{align*}
for all $\bm z\in [a_1^{(k)},c_1^{(k)}]^d = \bigl[-\xi_1\log(2dn),\,\xi_1\log(2dn)\bigr]^d$, where
\begin{gather}
    \bm g^{(k)}_{r}=\bigl(g^{(k)}_{r,1}, \ldots, g^{(k)}_{r,d_{i+1}^{(k)}}\bigr)^\top : [a_{r}^{(k)}, c_{r}^{(k)}]^{d_{r}^{(k)}} \to [a_{r+1}^{(k)}, c_{r+1}^{(k)}]^{d_{r+1}^{(k)}}, \text{ and} \nonumber\\
    g_{r,j}^{(k)} \in \mathcal{H}_{t_{r}^{(k)}}^{\beta_{r}^{(k)}}\bigl([a^{(k)}_r,c^{(k)}_r]^{d^{(k)}_r}, \gamma^{(k)}_r\bigr) \text{ for all } r\in[q_k], j\in[d_{r+1}^{(k)}]. \label{eq:defn-g_r^k}
\end{gather}
The reason that we are restricting the domain of $f^{\mathcal{S}_k}$ to $\bigl[-\xi_1\log(2dn),\,\xi_1\log(2dn)\bigr]^d$ is that by Assumption~\ref{assumption:piecewise-assumption}, each coordinate of $\bm Z_0$ is sub-exponential, so $\mathbb{P}\bigl(|Z_{0,j}| \geq \xi_1\log(2dn)\bigr) \leq 1/(dn)$ for all $j\in[d]$. Thus, by a union bound,
\begin{align}
    \mathbb{P}\Bigl\{ \bm Z_0 \notin \bigl[-\xi_1\log(2dn),\,\xi_1\log(2dn)\bigr]^d \Bigr\} \leq \frac{1}{n}, \label{eq:truncate-Z}
\end{align}
so it will suffice to approximate $f^{\mathcal{S}_k}$ on $\bigl[-\xi_1\log(2dn),\,\xi_1\log(2dn)\bigr]^d$.  To this end, we begin by approximating each function in the composition that defines $f^{\mathcal{S}_k}$. Finally, recall that for $k\in[K]$ and $r\in[q_k]$, we define
\begin{align*}
    \bar{\beta}^{(k)}_r \coloneqq \beta^{(k)}_{r} \prod_{\ell=r+1}^{q_k} (\beta^{(k)}_{\ell} \wedge 1). 
\end{align*}
We also define $r_*^{(k)} \coloneqq \mathrm{sargmax}_{r\in[q_k]} t^{(k)}_r / \bar{\beta}^{(k)}_r$ and let $\bar\beta^{(k)}_* \coloneqq \bar\beta^{(k)}_{r_*^{(k)}}$, $\beta^{(k)}_* \coloneqq \beta^{(k)}_{r_*^{(k)}}$, $t^{(k)}_* \coloneqq t^{(k)}_{r_*^{(k)}}$ and $\gamma_*^{(k)} \coloneqq \gamma^{(k)}_{r_*^{(k)}}$.

\begin{lemma} \label{lemma:approximate-each-Holder-component}
    Suppose that Assumptions~\ref{assumption:piecewise-assumption} and~\ref{assumption:composition-of-smooth-functions} hold. With the notation in~\eqref{eq:defn-a_r^(k)-c_r^(k)} and~\eqref{eq:defn-g_r^k}, for $M,N\in\mathbb{N}$ and $k\in[K]$, let 
    \begin{gather*}
    L^{(k)} \coloneqq \max_{r\in[q_k]} 12\lceil \beta^{(k)}_r \rceil^2(M+2)\lceil\log_2(4M)\rceil,\\
    p^{(k)}_* \coloneqq \max_{r\in[q_k]} \Bigl\{30 d_{r+1}^{(k)}  \lceil\beta^{(k)}_r\rceil^2 (t^{(k)}_r)^{\lceil\beta^{(k)}_r\rceil} (N+1)\lceil\log_2(8N)\rceil \vee 2 d_{r+1}^{(k)}\Bigr\} \text{ and}\\
    \bm p^{(k)}_r \coloneqq (d^{(k)}_r, p^{(k)}_*,\ldots, p^{(k)}_*, d^{(k)}_{r+1})^\top \in \mathbb{N}^{L^{(k)}+2} \text{ for } r\in[q_k].
    \end{gather*}
    For each $k\in[K]$, $r \in [q_k]$ and $n\in\mathbb{N}$, there exist $\bm f^{(k)}_r \in \mathcal{F}(L^{(k)},\bm p^{(k)}_r)$ and $E_r^{(k)} \subseteq [a^{(k)}_r,c^{(k)}_r]^{d^{(k)}_r}$ satisfying:
    \begin{itemize}
        \item[(i)] $\|\bm f_r^{(k)}(\bm x)\|_{\infty} \leq \gamma_r^{(k)} \vee (c_r^{(k)} - a_r^{(k)})^{\lceil \beta_r^{(k)} \rceil-1} \gamma_r^{(k)}$ for all $\bm x \in \mathbb{R}^{d_r^{(k)}}$ and
        \begin{align*}
            &\bigl\|\bm f^{(k)}_r(\bm x) - \bm g^{(k)}_r(\bm x)\bigr\|_{\infty}\\
            & \leq 9 \bigl( 1 \vee (c_r^{(k)} - a_r^{(k)})^{\lceil \beta_r^{(k)} \rceil - 1} \bigr) \gamma^{(k)}_r (\lceil \beta^{(k)}_r \rceil +2)^2 8^{\beta_r^{(k)}} (t^{(k)}_r)^{\lceil \beta_r^{(k)} \rceil - 1 + \beta^{(k)}_r/2} (NM)^{-2\beta^{(k)}_r / t^{(k)}_r},
        \end{align*}
        for all $\bm x \in [a^{(k)}_r,c^{(k)}_r]^{d^{(k)}_r} \setminus E^{(k)}_r$.
        \item[(ii)] $\mu_{\bm Z_0}(E^{(k)}_1) \leq \frac{1}{Kn q_k}$ and $\mu_{\bm Z_0}\bigl((\bm F^{(k)}_{r-1})^{-1}(E^{(k)}_r)\bigr) \leq \frac{1}{Kn q_k}$ for all $r\in\{2,\ldots,q_k\}$, where $\bm F^{(k)}_r \coloneqq \bm f^{(k)}_r \circ \cdots\circ \bm f^{(k)}_1$. 
    \end{itemize}
\end{lemma}
\begin{proof}
    Fixing $k\in[K]$, we construct $(f^{(k)}_r,E^{(k)}_r)_{r=1}^{q_k}$ inductively.  For $r \in [q_k]$, let $R^{(k)}_r \coloneqq \lfloor N^{1/t^{(k)}_r} \rfloor^2 \lfloor M^{2/t^{(k)}_r} \rfloor$ for $r \in [q_k]$. Let $\delta_{\max} \coloneqq \frac{1}{3\max_{r\in[q_k]}R_r^{(k)}}$. By Lemma~\ref{lemma:approximate-Holder-functions-[a,c]^d},~\ref{NN-enlarging},~\ref{NN-padding} and~\ref{NN-parallelisation}, for any $\delta\in(0,\delta_{\max}]$, there exists $\bm f^{(k)}_{1,\delta} \in \mathcal{F}(L^{(k)}, \bm p^{(k)}_1)$ such that $\|\bm f_{1,\delta}^{(k)}(\bm x)\|_{\infty} \leq \bigl( 1 \vee (c_1^{(k)} - a_1^{(k)})^{\lceil \beta_1^{(k)} \rceil - 1} \bigr) \gamma_1^{(k)}$ for all $\bm x \in \mathbb{R}^{d_1^{(k)}}$ and
    \begin{align*}
        &\bigl\|\bm f^{(k)}_{1,\delta}(\bm x) - \bm g^{(k)}_1(\bm x)\bigr\|_{\infty}\\
        &\quad \leq 9 \bigl( 1 \vee (c_1^{(k)} - a_1^{(k)})^{\lceil \beta_1^{(k)} \rceil - 1} \bigr) \gamma^{(k)}_1 (\lceil \beta^{(k)}_1 \rceil +2)^2 8^{\beta_1^{(k)}} (t^{(k)}_1)^{\lceil \beta_1^{(k)} \rceil - 1 + \beta^{(k)}_1/2} (NM)^{-2\beta^{(k)}_1 / t^{(k)}_1},
    \end{align*}
    for all $\bm x \in [a^{(k)}_1, c^{(k)}_1]^{d^{(k)}_1} \setminus E^{(k)}_{1,\delta}$, where $E^{(k)}_{1,\delta} \coloneqq \Gamma\bigl([a^{(k)}_1, c^{(k)}_1]^{d^{(k)}_1},\, R^{(k)}_1,\, \delta\bigr)$ as defined in Lemma~\ref{lemma:approximate-Holder-functions-[a,c]^d}. Further note that $(E^{(k)}_{1,\delta})_{\delta\in(0,\delta_{\max}]}$ are nested and have empty intersection.
    Thus, there exists $\delta^* \in (0,\delta_{\max}]$ such that $\mu_{\bm Z_0}(E^{(k)}_{1,\delta^*}) \leq \frac{1}{Kn q_k}$. We set $\bm f^{(k)}_1 \coloneqq \bm f^{(k)}_{1,\delta^*}$ and $E^{(k)}_1 \coloneqq E^{(k)}_{1,\delta^*}$, so that $\bm f^{(k)}_1$ and $E^{(k)}_1$ satisfy conditions~\emph{(i)} and~\emph{(ii)}.

    Now suppose that $\bm f^{(k)}_1,\ldots,\bm f^{(k)}_{r-1}$ and $E^{(k)}_1,\ldots,E^{(k)}_{r-1}$ satisfying conditions~\emph{(i)} and~\emph{(ii)} have been constructed for some $r\in\{2,\ldots,q_k\}$. We apply the same argument as above to construct $\bm f^{(k)}_r$ and $E^{(k)}_r$. By Lemma~\ref{lemma:approximate-Holder-functions-[a,c]^d},~\ref{NN-enlarging},~\ref{NN-padding} and~\ref{NN-parallelisation}, for any $\delta\in(0,\delta_{\max}]$, there exists $\bm f^{(k)}_{r,\delta} \in \mathcal{F}(L^{(k)}, \bm p^{(k)}_r)$ such that $\|\bm f_{r,\delta}^{(k)}(\bm x)\|_{\infty} \leq \bigl( 1 \vee (c_r^{(k)} - a_r^{(k)})^{\lceil \beta_r^{(k)} \rceil - 1} \bigr) \gamma_r^{(k)}$ for all $\bm x \in \mathbb{R}^{d_r^{(k)}}$ and
    \begin{align*}
        &\bigl\|\bm f^{(k)}_{r,\delta}(\bm x) - \bm g^{(k)}_r(\bm x)\bigr\|_{\infty}\\
        & \leq 9 \bigl( 1 \vee (c_r^{(k)} - a_r^{(k)})^{\lceil \beta_r^{(k)} \rceil - 1} \bigr) \gamma^{(k)}_r (\lceil \beta^{(k)}_r \rceil +2)^2 8^{\beta_r^{(k)}} (t^{(k)}_r)^{\lceil \beta_r^{(k)} \rceil - 1 + \beta^{(k)}_r/2} (NM)^{-2\beta^{(k)}_r / t^{(k)}_r},
    \end{align*}
    for all $\bm x \in [a^{(k)}_r, c^{(k)}_r]^{d^{(k)}_r} \setminus E^{(k)}_{r,\delta}$, where $E^{(k)}_{r,\delta} \coloneqq \Gamma\bigl([a^{(k)}_r, c^{(k)}_r]^{d^{(k)}_r},\, R^{(k)}_r,\, \delta\bigr)$. Again, since $\bigl((\bm F_{r-1}^{(k)})^{-1}(E_{r,\delta}^{(k)})\bigr)_{\delta \in (0,\delta_{\max}]}$ are nested and have empty intersection, there exists $\delta^* \in (0,\delta_{\max}]$ such that $\mu_{\bm Z_0}\bigl((\bm F_{r-1}^{(k)})^{-1}(E^{(k)}_{r,\delta^*})\bigr) \leq \frac{1}{Kn q_k}$. We set $\bm f^{(k)}_r \coloneqq \bm f^{(k)}_{r,\delta^*}$ and $E^{(k)}_r \coloneqq E^{(k)}_{r,\delta^*}$, so that $\bm f^{(k)}_r$ and $E^{(k)}_r$ satisfy conditions \emph{(i)} and \emph{(ii)}.
\end{proof}
We are now in a position to approximate each piece of the Bayes regression function by a neural network on a hypercube, once we have excised certain strips.
\begin{lemma} \label{lemma:approximate-composition-of-Holder-functions}
    Suppose that Assumptions~\ref{assumption:piecewise-assumption} and~\ref{assumption:composition-of-smooth-functions} hold and let $n\geq 2$. For $M,N\in\mathbb{N}$ and $k\in[K]$, let
    \begin{gather*}
    L^{(k)} \coloneqq \max_{r\in[q_k]} 12\lceil \beta^{(k)}_r \rceil^2(M+2)\lceil\log_2(4M)\rceil \text{ and}\\
    p^{(k)}_* \coloneqq \max_{r\in[q_k]} \Bigl\{ 30 d_{r+1}^{(k)}  \lceil\beta^{(k)}_r\rceil^2 (t^{(k)}_r)^{\lceil\beta^{(k)}_r\rceil} (N+1)\lceil\log_2(8N)\rceil \vee 2d_{r+1}^{(k)} \Bigr\}.
    \end{gather*}
    For each $k\in[K]$, there exist $E^{(k)} \subseteq \bigl[-\xi_1\log(2dn),\,\xi_1\log(2dn)\bigr]^d$ and $f^{(k)} \in \mathcal{F}(q_kL^{(k)}, \bm p^{(k)})$ where $\bm p^{(k)} \coloneqq (d, p^{(k)}_*,\ldots,p^{(k)}_*,1) \in \mathbb{N}^{q_kL^{(k)}+2}$ such that $\mu_{\bm Z_0}(E^{(k)}) \leq \frac{1}{Kn}$, $|f^{(k)}(\bm z)| \leq \gamma_{q_k}^{(k)} \vee (2\gamma_{q_k-1}^{(k)})^{\lceil\beta_{q_k}^{(k)}\rceil-1} \gamma_{q_k}^{(k)}$ for all $\bm z \in \mathbb{R}^d$, and
    \begin{align*}
        \bigl|f^{(k)}(\bm z) - f^{\mathcal{S}_k}(\bm z)\bigr| \leq C_1(\xi_1, \bm t^{(k)}, \bm \beta^{(k)}, \bm \gamma^{(k)}) \cdot \log^{\bar{\beta}_1^{(k)}} (2dn) \sum_{r=1}^{q_k} \frac{1}{(NM)^{2\bar{\beta}^{(k)}_r / t^{(k)}_r}}
    \end{align*}
    for all $\bm z\in \bigl[-\xi_1\log(2dn),\,\xi_1\log(2dn)\bigr]^d \setminus E^{(k)}$, where $C_1(\xi_1, \bm t^{(k)}, \bm \beta^{(k)}, \bm \gamma^{(k)}) > 0$ depends only on $(\xi_1, \bm t^{(k)}, \bm \beta^{(k)}, \bm \gamma^{(k)})$, and $C_1(\xi_1, \bm t^{(k)}, \bm \beta^{(k)}, \bm \gamma^{(k)}) \leq A(\xi_1,\bm \beta^{(k)}, \bm \gamma^{(k)})\|\bm t^{(k)}\|_\infty^{B(\bm \beta^{(k)})}$ for some $A(\xi_1,\bm \beta^{(k)}, \bm \gamma^{(k)}),B(\bm \beta^{(k)})  > 0$. 
\end{lemma}

\begin{proof}
    We use the notation in~\eqref{eq:defn-a_r^(k)-c_r^(k)} and~\eqref{eq:defn-g_r^k}.
    Fix $k \in [K]$, and let $(\bm f^{(k)}_r)_{r=1}^{q_k}$ and $(E^{(k)}_r)_{r=1}^{q_k}$ be defined as in Lemma~\ref{lemma:approximate-each-Holder-component}. By~\ref{NN-composition}, the function $f^{(k)} \coloneqq \bm f^{(k)}_{q_k} \circ\cdots\circ \bm f^{(k)}_1$ belongs to $\mathcal{F}(q_kL^{(k)}, \bm p^{(k)})$. 
    % By Lemma~\ref{lemma:approximate-each-Holder-component}, we have $|f^{(k)}(\bm z)| \leq \gamma_{q_k}^{(k)} \vee 2^{\lceil \beta_{q_k}^{(k)} \rceil} (c_{q_k}^{(k)} - a_{q_k}^{(k)})^{\lceil \beta_{q_k}^{(k)} \rceil} \gamma_{q_k}^{(k)}$ for all $\bm z \in [0,1]^d$.  
    For $r\in[q_k]$, recall the definition of $\bm F^{(k)}_r$ from Lemma~\ref{lemma:approximate-each-Holder-component} and let $\bm G^{(k)}_r \coloneqq \bm g^{(k)}_{r} \circ\cdots\circ \bm g^{(k)}_1$.
    Define $E^{(k)} \coloneqq \bigcup_{r\in[q_k]} (\bm F^{(k)}_{r-1})^{-1}(E^{(k)}_r)$ with $\bm F^{(k)}_0$ being the identity function. By Lemma~\ref{lemma:approximate-each-Holder-component} and a union bound, $\mu_{\bm Z_0}(E^{(k)}) \leq \frac{1}{Kn}$.
    For $\bm f:\mathbb{R}^d \to \mathbb{R}^m$ and $D\subseteq\mathbb{R}^d$, we define $\|\bm f\|_{L^{\infty}(D)} \coloneqq \sup_{\bm x\in D} \|\bm f(\bm x)\|_{\infty}$. For $r\in[q_k]$, let $D^{(k)}_r \coloneqq [a^{(k)}_r, c^{(k)}_r]^{d^{(k)}_r} \setminus E^{(k)}_r$ and let $D^{(k)} \coloneqq \bigl[-\xi_1\log(2dn),\,\xi_1\log(2dn)\bigr]^d \setminus E^{(k)}$. Then, for $r\in\{2,\ldots,q_k\}$,
    \begin{align}
        \bigl\|\bm F^{(k)}_r - \bm G^{(k)}_r&\bigr\|_{L^{\infty}(D^{(k)})} \leq \bigl\|\bm f^{(k)}_r \circ \bm f^{(k)}_{r-1} \circ\cdots\circ \bm f^{(k)}_1 - \bm g^{(k)}_r \circ \bm f^{(k)}_{r-1} \circ\cdots\circ \bm f^{(k)}_1\bigr\|_{L^{\infty}(D^{(k)})} \nonumber\\
        &\hspace{3cm} + \bigl\|\bm g^{(k)}_r \circ \bm f^{(k)}_{r-1} \circ\cdots\circ \bm f^{(k)}_1 - \bm g^{(k)}_r \circ \bm g^{(k)}_{r-1} \circ\cdots\circ \bm g^{(k)}_1\bigr\|_{L^{\infty}(D^{(k)})} \nonumber\\
        &\leq \bigl\|\bm f^{(k)}_r - \bm g^{(k)}_r\bigr\|_{L^{\infty} (\bm F^{(k)}_{r-1}(D^{(k)}))} + \gamma^{(k)}_r \Bigl\{ \bigl(t^{(k)}_r\bigr)^{1/2}\bigl\|\bm F^{(k)}_{r-1} - \bm G^{(k)}_{r-1}\bigr\|_{L^{\infty}(D^{(k)})} \Bigr\}^{\beta^{(k)}_r \wedge 1} \nonumber\\
        &\leq \bigl\|\bm f^{(k)}_r - \bm g^{(k)}_r\bigr\|_{L^{\infty} (D^{(k)}_r)} + \gamma^{(k)}_r \bigl(t^{(k)}_r\bigr)^{(\beta^{(k)}_r \wedge 1)/2} \bigl\|\bm F^{(k)}_{r-1} - \bm G^{(k)}_{r-1}\bigr\|_{L^{\infty}(D^{(k)})}^{\beta^{(k)}_r \wedge 1}, \label{eq:composition-error}
    \end{align}
    where the second inequality follows since each coordinate of $\bm g_r^{(k)}$ belongs to the class  $\mathcal{H}_{t_r^{(k)}}^{\beta_r^{(k)}}\bigl([a_r^{(k)},c_r^{(k)}]^{d_r^{(k)}}, \gamma_r^{(k)}\bigr)$.
    Thus, 
    \begin{align*}
        \bigl\|&f^{(k)} - f^{\mathcal{S}_k}\bigr\|_{L^{\infty}(D^{(k)})} = \bigl\|\bm F^{(k)}_{q_k} - \bm G^{(k)}_{q_k}\bigr\|_{L^{\infty}(D^{(k)})}\\
        &\leq \sum_{r=1}^{q_k} \biggl\{ \prod_{j=r+1}^{q_k} (\gamma^{(k)}_j)^{\bar{\beta}_j^{(k)} / \beta_j^{(k)}}\bigl(t^{(k)}_j\bigr)^{\bar{\beta}_{j-1}^{(k)} / (2\beta_{j-1}^{(k)})} \biggr\} \bigl\|\bm f^{(k)}_r - \bm g^{(k)}_r\bigr\|_{L^{\infty} (D^{(k)}_r)}^{\bar{\beta}^{(k)}_r / \beta^{(k)}_r}\\
        &\leq \sum_{r=1}^{q_k} \biggl\{ \prod_{j=r+1}^{q_k} (\gamma^{(k)}_j)^{\bar{\beta}_j^{(k)} / \beta_j^{(k)}}\bigl(t^{(k)}_j\bigr)^{\bar{\beta}_{j-1}^{(k)} / (2\beta_{j-1}^{(k)})} \biggr\} \Bigl\{9 \bigl( 1 \vee (c_r^{(k)} - a_r^{(k)})^{\lceil \beta_r^{(k)} \rceil - 1} \bigr) \gamma^{(k)}_r \\
        &\hspace{4.5cm} \times (\lceil \beta^{(k)}_r \rceil +2)^2 8^{\beta_r^{(k)}} (t^{(k)}_r)^{\lceil \beta_r^{(k)} \rceil - 1 + \beta^{(k)}_r/2} \Bigr\}^{\bar{\beta}^{(k)}_r / \beta^{(k)}_r} (NM)^{-2\bar{\beta}^{(k)}_r / t^{(k)}_r}\\
        &\leq C_1(\xi_1, \bm t^{(k)}, \bm \beta^{(k)}, \bm \gamma^{(k)}) \cdot \log^{\bar{\beta}_1^{(k)}} (2dn) \sum_{r=1}^{q_k} \frac{1}{(NM)^{2\bar{\beta}^{(k)}_r / t^{(k)}_r}},
    \end{align*}
    where $C_1(\xi_1, \bm t^{(k)}, \bm \beta^{(k)}, \bm \gamma^{(k)})$ has the properties claimed in the statement of the result.  Here, the first inequality follows by applying~\eqref{eq:composition-error} iteratively and using the fact that $(a+b)^t \leq a^t + b^t$ for $a,b \geq 0$ and $t\in[0,1]$, the second inequality follows from Lemma~\ref{lemma:approximate-each-Holder-component}, and the third inequality follows by substituting the definitions of $(a^{(k)}_r, c^{(k)}_r)_{r=1}^{q_k}$ in~\eqref{eq:defn-a_r^(k)-c_r^(k)}. Moreover, $\|f^{(k)}\|_{L^{\infty}(\mathbb{R}^d)} \leq \|\bm f^{(k)}_{q_k}\|_{L^{\infty}(\mathbb{R}^d)} \leq \gamma_{q_k}^{(k)} \vee (2\gamma_{q_k-1}^{(k)})^{\lceil\beta_{q_k}^{(k)}\rceil-1} \gamma_{q_k}^{(k)}$ by Lemma~\ref{lemma:approximate-each-Holder-component}.
\end{proof}
The following lemma is used in the proof of Lemma~\ref{lemma:use-embedding-to-extract-coordinates} below, which quantifies the extent to which we can extract coordinates, based on a function $\bm f_2$ that separates $\{\mathcal{S}_1,\ldots,\mathcal{S}_K\}$, using a neural network.
\begin{lemma} \label{lemma:approximate-multiplication-with-zero-one}
    For any $B>0$ and $N,L\in\mathbb{N}$, there exists $\phi \in \mathcal{F}\bigl(L,(2,9N+1,\ldots,9N+1,1)\bigr)$ such that for all $x\in[-B,B]$, we have
    \begin{align*}
        \phi(x,0) = 0 \quad\text{and}\quad |\phi(x,1)-x| \leq \frac{12B}{N^L}.
    \end{align*}
\end{lemma}
\begin{proof}
    By \citet[Lemma~5.2]{lu2021deep}, there exists $\phi_1 \in \mathcal{F}\bigl(L,(2,9N,\ldots,9N,1)\bigr)$ such that $|\phi_1(a,b)-ab| \leq 6N^{-L}$ for all $a,b\in[0,1]$. Moreover, by \citet[Eq~(5.2)]{lu2021deep}, the function $\phi_1$ is defined by $\phi_1(a,b) = 2\bigl(\psi(\frac{a+b}{2}) - \psi(\frac{a}{2}) - \psi(\frac{b}{2})\bigr)$, where $\psi$ is the neural network constructed in \citet[Lemma~5.1]{lu2021deep}, which satisfies $\psi(0) = 0$. Therefore, $\phi_1(a,0) = 0$ for all $a\in[0,1]$. 
    By~\ref{NN-composition} and~\ref{NN-parallelisation}, there exists $\phi\in\mathcal{F}\bigl(L,(2,9N+1,\ldots,9N+1,1)\bigr)$ such that $\phi(x,y) = 2B\phi_1(\frac{x+B}{2B}, y) - By$ for all $(x,y)\in\mathbb{R}\times[0,\infty)$. Moreover, for all $x\in[-B,B]$, we have $\phi(x,0) = 2B\phi_1(\frac{x+B}{2B}, 0) = 0$ and 
    \begin{align*}
        |\phi(x,1) - x| = 2B\biggl|\phi_1\biggl(\frac{x+B}{2B}, 1\biggr) - \frac{x+B}{2B}\biggr| \leq \frac{12B}{N^L},
    \end{align*}
    as required.
\end{proof}

\begin{lemma} \label{lemma:use-embedding-to-extract-coordinates}
    Suppose that $\{\mathcal{S}_1,\ldots,\mathcal{S}_K\}$ is separated by $\bm f_2 \in \mathcal{F}(L_2,\bm p_2)$. Let $m \coloneqq p_{2,L_2+1}$ and $B > 0$. For any $M,N_1,\ldots,N_K \in \mathbb{N}$, let
    \begin{gather*}
        s \coloneqq (13m+7)K + (2M+2)\sum_{k=1}^K(9N_k+1)^2 + 7,\\
        p_* \coloneqq (4m+1)K \vee  \sum_{k=1}^K (9N_k+1) \quad \text{and} \quad \bm p \coloneqq (K+m, p_*, \ldots, p_*, 1) \in \mathbb{N}^{M+5}.
    \end{gather*}
    Then there exists $f_3\in\mathcal{F}(M+3, \bm p, s)$ such that $|f_3(\bm u, \bm v)| \leq B$ for all $(\bm u, \bm v) \in \mathbb{R}^K \times \mathbb{R}^m$, and for all $\bm u=(u_1,\ldots,u_K)^\top \in [-B,B]^K$, $k\in[K]$ and $\bm\omega \in \mathcal{S}_k$, we have
    \begin{align*}
        \bigl|f_3 \bigl(\bm u, \bm f_2(\bm\omega)\bigr) - u_k \bigr| \leq\frac{12B}{N_k^M}.
    \end{align*}
\end{lemma}
\begin{proof}
    Suppose that $\{\mathcal{S}_1,\ldots,\mathcal{S}_K\}$ and $\bm f_2 \in \mathcal{F}(L_2,\bm p_2)$ satisfy Definition~\ref{defn:F(L,p,s)-separable} for some $\epsilon>0$ and $\bm v_1,\ldots,\bm v_K \in \mathbb{R}^m$.
    For $a\in\mathbb{R}$, define $h_a:\mathbb{R} \to [0,1]$ by
    \begin{align*}
        h_a(x) \coloneqq \sigma\biggl(\frac{2x-2a+2\epsilon}{\epsilon}\biggr) - \sigma\biggl(\frac{2x-2a+\epsilon}{\epsilon}\biggr) - \sigma\biggl(\frac{2x-2a-\epsilon}{\epsilon}\biggr) +\sigma\biggl(\frac{2x-2a-2\epsilon}{\epsilon}\biggr).
    \end{align*}
    Then $h_a \in \mathcal{F}\bigl(1,(1,4,1),12\bigr)$, with $h_a(x) = 1$ for all $x\in\bigl[a-\frac{\epsilon}{2}, a+\frac{\epsilon}{2}\bigr]$ and $h_a(x) = 0$ for all $x\in(-\infty,a-\epsilon] \cup [a+\epsilon,\infty)$. Further, for $k\in[K]$, write $\bm v_k = (v_{k,1},\ldots,v_{k,m})^\top \in \mathbb{R}^m$ and define $g_k \in \mathcal{F}\bigl(2,(m,4m,1,1),13m+2\bigr)$ by
    \begin{align*}
        g_k(\bm x) \coloneqq \sigma\biggl(\sum_{j=1}^m h_{v_{k,j}}(x_j) - m + 1 \biggr).
    \end{align*}
    Then, for $\bm\omega \in \mathcal{S}_k$, we have $g_k\bigl(\bm f_2(\bm\omega)\bigr) = \sigma(m-m+1) = 1$; for $\bm\omega \notin \mathcal{S}_k$, we have $\sum_{j=1}^m h_{v_{k,j}}\bigl(f_{2,j}(\bm\omega)\bigr) \leq m-1$, so $g_k\bigl(\bm f_2(\bm\omega)\bigr) = 0$. By Lemma~\ref{lemma:approximate-multiplication-with-zero-one}, for $k\in[K]$, there exists $\phi_k \in \mathcal{F}\bigl(M,(2,9N_k+1,\ldots,9N_k+1,1)\bigr)$ such that $\phi_k(x,0) = 0$ and $|\phi_k(x,1) - x| \leq 12BN_k^{-M}$ for all $x\in[-B,B]$. Further, by~\ref{NN-enlarging},~\ref{NN-composition} and~\ref{NN-parallelisation}, there exists $f_3^\circ\in\mathcal{F}\bigl(M+2,\, (K+m,p_*,\ldots,p_*,1),\,  s-7\bigr)$ such that $f_3^\circ(\bm u, \bm v) = \sum_{k=1}^K \phi_k\bigl(u_k, g_k(\bm v)\bigr)$ for all $\bm u \in[-B,B]^K$ and $v\in\mathbb{R}^m$. 
    Moreover, for $\bm\omega \in \mathcal{S}_k$ and $\bm u = (u_1,\ldots,u_K)^\top \in [-B,B]^K$, we have
    \begin{align*}
        \bigl|f_3^\circ\bigl(\bm u, \bm f_2(\bm\omega)\bigr) - u_k \bigr| &= \biggl|\sum_{\ell=1}^K \Bigl\{\phi_k\Bigl(u_{\ell},g_{\ell}\bigl(\bm f_2(\bm\omega)\bigr)\Bigr) - u_{\ell}g_{\ell}\bigl(\bm f_2(\bm\omega)\bigr) \Bigr\}\biggr|\\
        &\leq \sum_{\ell=1}^K \Bigl| \phi_k \Bigl(u_{\ell},g_{\ell}\bigl(\bm f_2(\bm\omega)\bigr)\Bigr) - u_{\ell}g_{\ell}\bigl(\bm f_2(\bm\omega)\bigr) \Bigr|\\
        &=|\phi_k(u_k,1) - u_k| + \sum_{\ell\neq k} |\phi_\ell(u_\ell, 0) - 0|\leq \frac{12B}{N_k^M}. 
    \end{align*}
    Finally, define $f_3: \mathbb{R}^{K+m} \to [-B,B]$ by
    \begin{align*}
        f_3 \coloneqq \sigma(f_3^\circ + B) - \sigma(f_3^\circ - B) - B = T_B f_3^\circ.
    \end{align*}
    By~\ref{NN-enlarging} and~\ref{NN-composition}, we have $f_3 \in \mathcal{F}(M+3,\bm p,s)$ and since $u_k \in [-B,B]$, we deduce that $\bigl|f_3\bigl(\bm u, \bm f_2(\bm\omega)\bigr) - u_k \bigr| \leq \bigl|f_3^\circ\bigl(\bm u, \bm f_2(\bm\omega)\bigr) - u_k \bigr|$ for all $\bm u=(u_1,\ldots,u_K)^\top \in [-B,B]^K$, $k\in[K]$ and $\bm\omega \in \mathcal{S}_k$.
\end{proof}

\section{Proof of Theorem \ref{thm:PENN-ub}}

\begin{proof}[Proof of Theorem \ref{thm:PENN-ub}]
    For any $M_1 \in \mathbb{N}$, let 
    \[
    L_1 \coloneqq \max_{k\in[K],\, r\in[q_k]} 12q_k\lceil \beta^{(k)}_r \rceil^2 (M_1+2) \lceil\log_2(4M_1)\rceil \in \mathbb{N}
    \]
    and for $k \in [K]$, let
    \begin{gather*}
        N^{(k)}_1 \coloneqq \bigl\lceil M_1^{-1} n_k^{t^{(k)}_* / (4\bar\beta^{(k)}_* + 2t^{(k)}_*)} \bigr\rceil,\\
        p^{(k)}_{1,*} \coloneqq \max_{r\in[q_k]} \Bigl\{ 30 d_{r+1}^{(k)} \lceil\beta^{(k)}_r\rceil^2 (t^{(k)}_r)^{\lceil\beta^{(k)}_r\rceil} (N^{(k)}_1+1) \lceil\log_2(8N^{(k)}_1)\rceil \vee 2d_{r+1}^{(k)} \Bigr\},\\
        \bm{p}^{(k)}_1 \coloneqq (d,p^{(k)}_{1,*},\ldots,p^{(k)}_{1,*},1) \in \mathbb{N}^{L_1+2}.
    \end{gather*}
    By Lemma~\ref{lemma:approximate-composition-of-Holder-functions} and~\ref{NN-padding}, for $k \in [K]$ we can find $E^{(k)} \subseteq \bigl[-\xi_1\log(2dn),\,\xi_1\log(2dn)\bigr]^d$ and $f^{(k)}_1 \in \mathcal{F}(L_1,\bm{p}^{(k)}_1)$ such that $\mu_{\bm Z_0}(E^{(k)}) \leq \frac{1}{Kn}$, 
    \[
    |f^{(k)}_1(\bm z)| \leq \max_{k\in[K]} \bigl\{\gamma_{q_k}^{(k)} \vee (2\gamma_{q_k-1}^{(k)})^{\lceil\beta_{q_k}^{(k)}\rceil-1} \gamma_{q_k}^{(k)}\bigr\} \eqqcolon B
    \]
    for all $\bm z \in \mathbb{R}^d$ and 
    \begin{align}
        \bigl|f^{(k)}_1(\bm{z}) - f^{\mathcal{S}_k}(\bm{z})\bigr| & \leq  C_1(\xi_1, \bm t^{(k)}, \bm \beta^{(k)}, \bm \gamma^{(k)}) \cdot \log^{\bar{\beta}_1^{(k)}} (2dn) \cdot q_k \cdot n_k^{-\bar\beta^{(k)}_* / (2\bar\beta^{(k)}_* + t^{(k)}_*)}\nonumber\\
        & \leq C_2(\xi_1, d, \bm t^{(k)}, \bm \beta^{(k)}, \bm \gamma^{(k)}) \cdot \log^{\bar{\beta}_1^{(k)}}(n) \cdot n_k^{-\bar\beta^{(k)}_* / (2\bar\beta^{(k)}_* + t^{(k)}_*)} \label{eq:approx-error-for-f^{S_k}}
    \end{align}
    for all $\bm{z}\in \bigl[-\xi_1\log(2dn),\,\xi_1\log(2dn)\bigr]^d \setminus E^{(k)}$.  
    % Here, $C_2(\xi_1, d, q_k,\bm t^{(k)}, \bm \beta^{(k)}, \bm \gamma^{(k)}) \leq A_2(\xi_1,d,q_k,\bm \beta^{(k)}, \bm \gamma^{(k)})\|\bm t^{(k)}\|_\infty^{B_2(\bm \beta^{(k)})}$ for some $A(\xi_1,\bm \beta^{(k)}, \bm \gamma^{(k)}),B(\bm \beta^{(k)})  > 0$. depends at most polynomially on $d$ and $\bm t^{(k)}$. 
    Furthermore, each $f^{(k)}_1$ has $V^{(k)}_1 \coloneqq (d+L_1+1)p^{(k)}_{1,*} + (L_1-1)(p_{1,*}^{(k)})^2 + 1$ parameters. Thus, by~\ref{NN-parallelisation}, the function $\bm{f}_1 \coloneqq (f^{(1)}_1, \ldots, f^{(K)}_1)^\top : \mathbb{R}^d \to [-B,B]^K$ belongs to $\mathcal{F}\bigl(L_1, \bm{p}_1, s_1 \bigr)$, where $\bm{p}_1 \coloneqq \bigl(d, \sum_{k=1}^K p^{(k)}_{1,*}, \ldots, \sum_{k=1}^K p^{(k)}_{1,*}, K\bigr) \in \mathbb{N}^{L_1+2}$ and $s_1 \coloneqq \sum_{k=1}^K V^{(k)}_1$. 

    By assumption, there exists $\bm f_2 \in \mathcal{F}(L_2,\bm p_2,s_2)$ such that $\{\mathcal{S}_1,\ldots,\mathcal{S}_K\}$ is separated by~$\bm f_2$.  For any $M_3 \in \mathbb{N}$, let $m\coloneqq p_{2,L_2+1}$, and let
    \begin{gather*}
        L_3 \coloneqq \biggl\lceil 2M_3 \max_{k\in[K]} \frac{\bar{\beta}^{(k)}_*}{t^{(k)}_*} \biggr\rceil  +3,\quad
        N_{3,k} \coloneqq \Bigl\lceil n_k^{\bar\beta^{(k)}_* / \{(L_3-3)(2\bar\beta^{(k)}_* + t^{(k)}_*)\}} \Bigr\rceil \text{ for }k\in[K],\\
        p_{3,*} \coloneqq (4m+1)K \vee \sum_{k=1}^K (9N_{3,k}+1),\quad
        \bm{p}_3 \coloneqq \bigl(K+m,\, p_{3,*},\,\ldots,\,p_{3,*}, \,1\bigr) \in \mathbb{N}^{L_3+2} \text{ and}\\
        s_3 \coloneqq (13m+7)K + (2M_3+2)\sum_{k=1}^K (9N_{3,k}+1)^2 + 7.
    \end{gather*}
    By Lemma~\ref{lemma:use-embedding-to-extract-coordinates}, there exists $f_3 \in \mathcal{F}(L_3,\bm p_3, s)$ such that $|f_3(\bm u, \bm v)| \leq B$ for all $(\bm u,\bm v) \in \mathbb{R}^K \times \mathbb{R}^m$, and that for all $\bm z \in \mathbb{R}^d$, $k\in[K]$ and $\bm\omega \in \mathcal{S}_k$, we have
    \begin{align}
        \bigl| f_3\bigl(\bm f_1(\bm z), \bm f_2(\bm\omega)\bigr) - f_1^{(k)}(\bm z) \bigr| \leq 12B \cdot n_k^{-\bar\beta^{(k)}_* / (2\bar\beta^{(k)}_* + t^{(k)}_*)}. \label{eq:extract-coordinate-error}
    \end{align}
    Further note that for $k\in[K]$,
    \begin{align*}
        N_{3,k} \leq 2n_k^{t_*^{(k)} / (4\bar{\beta}_*^{((k)} + 2t_*^{(k)})}. 
    \end{align*}
    Thus,
    \begin{align*}
        s_3 \leq C_3(m,M_3) \sum_{k=1}^K n_k^{t_*^{(k)} / (2\bar{\beta}_*^{((k)} + t_*^{(k)})}.
    \end{align*}
    Now, define $\bar{f} \in \mathcal{F}$ by
    \begin{align*}
        \bar{f}(\bm{z},\bm{\omega}) \coloneqq f_3 \bigl( \bm{f}_1(\bm{z}),\, \bm f_2(\bm{\omega}) \bigr),
    \end{align*}
    so that $|\bar{f}(\bm{z},\bm{\omega})| \leq B$ for all $\bm z \in \mathbb{R}^d$ and $\bm\omega \in \mathcal{S}$.
    Let $D\coloneqq \bigl[-\xi_1\log(2dn),\,\xi_1\log(2dn)\bigr]^d \setminus \bigl(\bigcup_{k\in[K]}E^{(k)}\bigr)$, so that by Lemma~\ref{lemma:approximate-composition-of-Holder-functions} and~\eqref{eq:truncate-Z}, we have
    \begin{align}
        \mathbb{P}(\bm Z_0 \notin D) \leq \mathbb{P}\bigl(\bm Z_0 \notin \bigl[-\xi_1\log(2dn),\,\xi_1\log(2dn)\bigr]^d \bigr) + \sum_{k=1}^K \mathbb{P}(\bm Z_0 \in E^{(k)}) \leq \frac{2}{n}. \label{eq:D-complement-prob}
    \end{align}
    By~\eqref{eq:approx-error-for-f^{S_k}} and~\eqref{eq:extract-coordinate-error}, for all $\bm z \in D$, $k\in[K]$ and $\bm\omega \in \mathcal{S}_k$, we have 
    \begin{align}
        \bigl|\bar{f}(\bm{z},\bm{\omega}) &- f^{\star}(\bm{z},\bm{\omega})\bigr| \leq \bigl|f_3 \bigl( \bm{f}_1(\bm{z}),\, \bm f_2(\bm{\omega}) \bigr) - f_1^{(k)}(\bm{z}) \bigr| +  \bigl| f_1^{(k)}(\bm{z}) - f^{\mathcal{S}_k}(\bm{z})\bigr|\nonumber\\
        &\leq \bigl\{C_2(\xi_1, d, \bm t^{(k)}, \bm \beta^{(k)}, \bm \gamma^{(k)}) + 12B \bigr\} \cdot \log^{\bar{\beta}_1^{(k)}}(n) \cdot n_k^{-\bar\beta^{(k)}_* / (2\bar\beta^{(k)}_* + t^{(k)}_*)}. \label{eq:f-bar-property-1}
    \end{align}
    Hence there exists $C_4>0$, depending only on $\xi_1,d,m$ and $(\bm t^{(k)}, \bm \beta^{(k)}, \bm \gamma^{(k)})_{k=1}^K$ such that
    \begin{align}
        \inf_{f\in\mathcal{F}} \mathbb{E}\Bigl\{ \bigl(f(\bm Z_0,\bm \Omega_0) &- f^{\star}(\bm Z_0,\bm \Omega_0)\bigr)^2 \Bigr\}\nonumber\\
        &\leq \mathbb{E}\Bigl\{ \bigl(\bar{f}(\bm Z_0,\bm \Omega_0) - f^{\star}(\bm Z_0,\bm \Omega_0)\bigr)^2 \mathbbm{1}_{\{\bm Z_0 \in D\}} \Bigr\} + \frac{8B^2}{n}\nonumber\\
        &= \sum_{k=1}^K \pi_k \mathbb{E}\Bigl\{ \bigl(\bar{f}(\bm Z_0,\bm \Omega_0) - f^{\star}(\bm Z_0,\bm \Omega_0)\bigr)^2 \mathbbm{1}_{\{\bm Z_0 \in D\}} \,\bigm|\, \bm\Omega_0 \in \mathcal{S}_k \Bigr\}  + \frac{8B^2}{n}\nonumber\\
        &\leq \sum_{k=1}^K \pi_k \sup_{\bm z\in D, \bm\omega\in\mathcal{S}_k} \bigl\{\bar{f}(\bm{z},\bm{\omega}) - f^{\star}(\bm{z},\bm{\omega})\bigr\}^2  + \frac{8B^2}{n}\nonumber\\
        &\leq C_4 (\log n)^{2\max_{k\in[K]}\bar{\beta}_1^{(k)}} \cdot \sum_{k=1}^K \pi_k n_k^{-2\bar\beta^{(k)}_* / (2\bar\beta^{(k)}_* + t^{(k)}_*)}, \label{eq:approx-error}
    \end{align}
    where the first inequality follows from~\eqref{eq:D-complement-prob} and $|\bar{f}(\bm z,\bm\omega) - f^{\star}(\bm z,\bm\omega)| \leq 2B$ for all $\bm z \in \mathbb{R}^d$ and $\bm\omega\in\mathcal{S}$, and the final inequality follows from~\eqref{eq:f-bar-property-1}.
    Let $L_0 \coloneqq L_3+(L_1 \vee L_2)$, $p_0 \coloneqq \bigl(2d, 2\|\bm{p}_1\|_{\infty} + 2\|\bm{p}_2\|_{\infty}, \ldots, 2\|\bm{p}_1\|_{\infty} + 2\|\bm{p}_2\|_{\infty}, \bm{p}_3\bigr) \in \mathbb{N}^{L_0+2}$ and $s_0 \coloneqq 2(s_1+s_2+s_3) + 2(K \vee m)|L_1 - L_2|$. By~\ref{NN-enlarging}--\ref{NN-parallelisation}, we have $\mathcal{F}\subseteq \mathcal{F}(L_0, \bm{p}_0, s_0)$. Moreover, \begin{align}
        \frac{s_0L_0\log(es_0) + s_0\log(ed)}{n}
        & \leq C_5 \log^3 n \cdot \frac{\sum_{k=1}^K n_k^{t_*^{(k)} / (2\bar{\beta}_*^{(k)} + t_*^{(k)})} + s_2\log s_2}{n}\nonumber\\
        &= C_5\log^3 n \cdot  \biggl\{\sum_{k=1}^K \pi_k n_k^{-2\bar\beta^{(k)}_*/(2\bar\beta^{(k)}_* + t^{(k)}_*)} + \frac{s_2 \log s_2}{n}\biggr\}, \label{eq:entropy-bound}
    \end{align}
    where $C_5>0$ depends only on $m,L_2,M_1,M_3,d$ and $(\bm d^{(k)}, \bm t^{(k)}, \bm\beta^{(k)}, \bm\gamma^{(k)})_{k=1}^K$. The final result then follows by applying Proposition~\ref{prop:oracle-inequality} in conjunction with~\eqref{eq:approx-error} and~\eqref{eq:entropy-bound}.  
    % \blue{By~\ref{NN-enlarging}--\ref{NN-parallelisation}, $\mathcal{F}\subseteq \mathcal{F}(L_0, \bm{p}_0, s_0)$, so $\mathrm{Pdim}(\mathcal{F})  \leq \mathrm{Pdim}\bigl(\mathcal{F}(L_0, \bm{p}_0, s_0)\bigr)$. For $(\bm u_1,\bm v_1),\ldots,(\bm u_n,\bm v_n) \in \mathbb{R}^d\times\mathbb{R}^d$, let $\mathbb{P}_n$ be their empirical measure, and let 
    % \begin{align*}
    %     \mathcal{N}_{1,n} \coloneqq \sup_{(\bm u_i,\bm v_i)_{i=1}^n \in (\mathbb{R}^d\times\mathbb{R}^d)^n} \mathcal{N}\bigl((80B_n n)^{-1},\, T_{B_n}\circ\mathcal{F},\, \|\cdot\|_{L_1(\mathbb{P}_n)}\bigr).
    % \end{align*}
    % Then by Proposition~\ref{prop:VC-upper-bound}\emph{(b)}, we have
    % \begin{align}
    %     \frac{\log(14\mathcal{N}_{1,n})}{n} &\lesssim \frac{\bigl\{s_0L_0\log(es_0) + s_0\log(2ed)\bigr\} \log(80B_n^2n)}{n} \nonumber\\
    %     & \leq C_5 \log^4 n \cdot \frac{\sum_{k=1}^K n_k^{t_*^{(k)} / (2\bar{\beta}_*^{(k)} + t_*^{(k)})} + s_2\log s_2}{n}\nonumber\\
    %     &= C_5\log^4 n \cdot  \biggl\{\sum_{k=1}^K \pi_k n_k^{-2\bar\beta^{(k)}_*/(2\bar\beta^{(k)}_* + t^{(k)}_*)} + \frac{s_2 \log s_2}{n}\biggr\}, 
    % \end{align}
    % where $C_5>0$ depends only on $m,\xi_2,L_2,M_1,M_3,d$ and $(\bm d^{(k)}, \bm t^{(k)}, \bm\beta^{(k)}, \bm\gamma^{(k)})_{k=1}^K$. The final result then follows by applying Proposition~\ref{prop:oracle-inequality} in conjunction with~\eqref{eq:approx-error} and~\eqref{eq:entropy-bound}.}
\end{proof}

\section{Proof of Theorem \ref{thm:minimax-lb}}
\begin{proof}[Proof of Theorem \ref{thm:minimax-lb}]
    Without loss of generality, we may assume that $j_* = d$. We also define $\beta_*^{(k)} \coloneqq \beta^{(k)}_{r_*^{(k)}}$. Let $\bm X_0$ be uniformly distributed on $[0,1]^d$, so that $\|X_{0,j}\|_{\psi_1} \leq 0.8 \leq \xi_1$ for $j\in[d]$. For $k\in[K-1]$, let $A_k \coloneqq \bigl[\sum_{\ell=1}^{k-1} \pi_\ell, \sum_{\ell=1}^{k} \pi_\ell\bigr)$ and let $A_K \coloneqq \bigl[\sum_{\ell=1}^{K-1} \pi_\ell, 1\bigr]$.
    Given $f^{(1)}, \ldots, f^{(K)} : [0,1]^d \to [0,1/2]$, we define $g_{(f^{(1)},\ldots,f^{(K)})} : [0,1]^d \to [0,1/2]$ by $g_{(f^{(1)},\ldots,f^{(K)})}(\bm x) \coloneqq \sum_{k=1}^K \mathbbm{1}_{\{x_d \in A_k\}} \cdot f^{(k)}(\bm x)$ where $\bm x = (x_1,\ldots,x_d)^\top \in [0,1]^d$. 
    Further, let $P_{(f^{(1)}, \ldots, f^{(K)})}$ be the distribution of $(\bm Z_0,\bm\Omega_0,Y_0)$, where $\bm\Omega_0 \,|\, \bm X_0 \sim \mathrm{Unif}(\mathcal{S}_k)$ if $X_{0,d} \in A_k$, $\bm Z_0 = \mathsf{Imp}(\tilde{\bm X}_0)$ (where $\tilde{\bm{X}}_0$ is the partially observed version of $\bm{X}_0$) and $Y_0 = g_{(f^{(1)},\ldots,f^{(K)})}(\bm X_0) + \varepsilon_0$, where $\varepsilon_0 \sim N(0,\xi_2^2/25)$ is independent of $(\bm X_0, \bm\Omega_0)$.  Thus 
    \[
    \|Y_0\|_{\psi_2} \leq \|g_{(f^{(1)},\ldots,f^{(K)})}(\bm X_0)\|_{\psi_2} + \|\varepsilon_0\|_{\psi_2} \leq \frac{1}{2\sqrt{\log 2}} + \sqrt{\frac{8}{75}}\xi_2 \leq \xi_2.
    \]
    Moreover, observe that when $\bm{\omega} \in \mathcal{S}_k$, we must have $X_{0,d} \in A_k$ and hence $g_{(f^{(1)},\ldots,f^{(K)})}(\bm X_0) = f^{(k)}(\bm{X}_0)$.  Hence, if for each $k\in[K]$, the function $f^{(k)}$ depends only on the coordinates in $\tilde{\mathcal{J}}^{(k)} \coloneqq\mathcal{J}^{(k)}\setminus\{d\}$, then when $\bm{\omega} \in \mathcal{S}_k$,
    \begin{align*}
    f^{\star}(\bm z, \bm\omega) = \mathbb{E}(\bm{Y}_0 \, | \, \bm{Z}_0 = \bm{z},\bm{\Omega}_0 = \bm{\omega}) &= \mathbb{E}\bigl(g_{(f^{(1)},\ldots,f^{(K)})}(\bm X_0) \, | \, \bm{Z}_0 = \bm{z},\bm{\Omega}_0 = \bm{\omega}\bigr) \\
    &= \mathbb{E}\bigl(f^{(k)}(\bm X_0) \, | \, \bm{Z}_0 = \bm{z},\bm{\Omega}_0 = \bm{\omega}\bigr) = f^{(k)}(\bm z).
    \end{align*}
    In general then, $f^{\star}(\bm z, \bm\omega) = \sum_{k=1}^K f^{(k)}(\bm z)\mathbbm{1}_{\{\bm{\omega} \in \mathcal{S}_k\}}$.  For $k \in [K]$, let 
    \begin{align*}
    \mathcal{F}^{(k)} \coloneqq \Bigl\{ f \in \mathcal{H}_{\mathrm{comp}}(q_k,\bm d^{(k)}, &\bm t^{(k)}, \bm \beta^{(k)}, \bm\gamma^{(k)}): f \text{ takes values in } [0,1/2] \text{ and}\\
    &\text{depends only on the coordinates in $\tilde{\mathcal{J}}^{(k)}$} \Bigr\}.
    \end{align*}
    We have established that if $f^{(k)} \in \mathcal{F}^{(k)}$ for all $k\in[K]$, then $P_{(f^{(1)}, \ldots, f^{(K)})} \in \mathcal{P}$.
    Let $\mu_{\bm{Z}_0,\bm{\Omega}_0}$ be the joint distribution of $(\bm{Z}_0,\bm{\Omega}_0)$ when $(\bm{Z}_0,\bm{\Omega}_0,Y_0) \sim P_{(f^{(1)}, \ldots, f^{(K)})} \in \mathcal{P}$, and for $k\in[K]$ and $f^{(k)} \in \mathcal{F}^{(k)}$, let $P_{f^{(k)}} \coloneqq P_{(0,\ldots,0,f^{(k)},0,\ldots,0)}$.  Then
    \begin{align}
        &\inf_{\hat{f}\in \hat{\mathcal{F}}} \sup_{P \in \mathcal{P}} \;\mathbb{E}_{P^{\otimes n}}\bigl\{R(\hat{f}) - R(f^{\star})\bigr\} 
        \geq \inf_{\hat{f}\in \hat{\mathcal{F}}} \sup_{\substack{f^{(1)},\ldots,f^{(K)} : \\ f^{(\ell)} \in \mathcal{F}^{(\ell)} \,\forall \ell\in[K]}} \mathbb{E}_{P_{(f^{(1)}, \ldots, f^{(K)})}^{\otimes n}}\bigl\{\|\hat{f} - f^{\star}\|_{L_2(\mu_{\bm Z_0,\bm\Omega_0})}^2\bigr\} \nonumber\\
        &\hspace{0.4cm}= \inf_{\hat{f} \in \hat{\mathcal{F}}} \sup_{\substack{f^{(1)},\ldots,f^{(K)} : \\ f^{(\ell)} \in \mathcal{F}^{(\ell)} \,\forall \ell\in[K]}} \sum_{k=1}^K \mathbb{E}_{P_{(f^{(1)}, \ldots, f^{(K)})}^{\otimes n}} \int_{[0,1]^d \times \mathcal{S}} \bigl\{\hat{f}(\bm z,\bm\omega) - f^{(k)}(\bm z)\bigr\}^2 \mathbbm{1}_{\{\bm\omega\in\mathcal{S}_k\}} \, \mathrm{d}\mu_{\bm Z_0,\bm\Omega_0}(\bm z, \bm \omega) \nonumber\\
        &\hspace{0.4cm}\geq \sum_{k=1}^K \inf_{\hat{f} \in \hat{\mathcal{F}}} \sup_{\substack{f^{(1)},\ldots,f^{(K)} : \\ f^{(\ell)} \in \mathcal{F}^{(\ell)} \,\forall \ell\in[K]}}  \mathbb{E}_{P_{(f^{(1)}, \ldots, f^{(K)})}^{\otimes n}} \int_{[0,1]^d \times \mathcal{S}} \bigl\{\hat{f}(\bm z,\bm\omega) - f^{(k)}(\bm z)\bigr\}^2 \mathbbm{1}_{\{\bm\omega\in\mathcal{S}_k\}} \, \mathrm{d}\mu_{\bm Z_0,\bm\Omega_0}(\bm z, \bm \omega)\nonumber\\
        &\hspace{0.4cm}\geq \sum_{k=1}^K \inf_{\hat{f} \in \hat{\mathcal{F}}} \sup_{f^{(k)} \in \mathcal{F}^{(k)}}  \mathbb{E}_{P_{f^{(k)}}^{\otimes n}} \int_{[0,1]^d \times \mathcal{S}} \bigl\{\hat{f}(\bm z,\bm\omega) - f^{(k)}(\bm z)\bigr\}^2 \mathbbm{1}_{\{\bm\omega\in\mathcal{S}_k\}} \, \mathrm{d}\mu_{\bm Z_0,\bm\Omega_0}(\bm z, \bm \omega). \label{eq:minimax-lower-bound-k-simplified}
    \end{align}

    Now fix $k\in[K]$ and take $f^{(\ell)} = 0$ for all $\ell \neq k$. 
    Since any $f^{(k)} \in \mathcal{F}^{(k)}$ depends only on the coordinates in $\tilde{\mathcal{J}}^{(k)}$, there exists $\tilde{f}^{(k)}: [0,1]^{|\tilde{\mathcal{J}}^{(k)}|} \to [0, 1/2]$ such that $f^{(k)}(\bm x) = \tilde{f}^{(k)}(\bm x_{\tilde{\mathcal{J}}^{(k)}})$ for all $\bm x = (x_1,\ldots,x_d)^\top \in[0,1]^d$, where $\bm x_{\tilde{\mathcal{J}}^{(k)}} \coloneqq (x_j)_{j \in \tilde{\mathcal{J}}^{(k)}} \in [0,1]^{|\tilde{\mathcal{J}}^{(k)}|}$. Thus, writing $\|f\|_{L_2} \coloneqq \bigl(\int_{\bm x \in [0,1]^m} f(\bm x)^2 \,\mathrm{d}\bm x\bigr)^{1/2}$ for square-integrable $f:[0,1]^m \to \mathbb{R}$ and $m\in\mathbb{N}$, we deduce that for $f_1^{(k)}, f_2^{(k)} \in \mathcal{F}^{(k)}$,
    \begin{align}
        \int_{[0,1]^d \times \mathcal{S}} \bigl\{f^{(k)}_1(\bm z) &- f^{(k)}_2(\bm z) \bigr\}^2\mathbbm{1}_{\{\bm\omega\in\mathcal{S}_k\}} \, \mathrm{d}\mu_{\bm Z_0,\bm\Omega_0}(\bm z, \bm \omega) = \mathbb{E} \Bigl[ \bigl\{f^{(k)}_1(\bm Z_0) - f^{(k)}_2(\bm Z_0)\bigr\}^2 \mathbbm{1}_{\{\bm\Omega_0 \in \mathcal{S}_k\}} \Bigr]\nonumber\\
        &= \mathbb{E} \Bigl[ \bigl\{\tilde{f}^{(k)}_1(\bm X_{0,\tilde{\mathcal{J}}^{(k)}}) - \tilde{f}^{(k)}_2(\bm X_{0,\tilde{\mathcal{J}}^{(k)}})\bigr\}^2 \mathbbm{1}_{\{X_{0,d} \in A_k\}} \Bigr]\nonumber\\
        &= \pi_k \mathbb{E} \Bigl[ \bigl\{\tilde{f}^{(k)}_1(\bm X_{0,\tilde{\mathcal{J}}^{(k)}}) - \tilde{f}^{(k)}_2(\bm X_{0,\tilde{\mathcal{J}}^{(k)}})\bigr\}^2 \,\Big|\, X_{0,d} \in A_k \Bigr] \nonumber\\
        &= \pi_k \bigl\| f^{(k)}_1 - f^{(k)}_2 \bigr\|_{L_2}^2. \label{eq:mu-Z-Omega-norm-bound}
    \end{align}
    Next, we use a construction similar to \citet[Theorem~3]{schmidt-hieber2020nonparametric} to prove the lower bound.  Define $R \coloneqq \lfloor \lambda n_k^{1/(2\bar\beta_*^{(k)} + t_*^{(k)})} \rfloor \in \mathbb{N}$ where $\lambda \geq 16$ will be chosen later, let $\rho\coloneqq 1/R$ and let $\mathcal{U} \coloneqq \{0,\rho,2\rho,\ldots,(R-1)\rho\}^{t_*^{(k)}}$.  Define $h: \mathbb{R} \rightarrow [0,1]$ by $h(x) \coloneqq c_1e^{-1/\{x(1-x)\}}\mathbbm{1}_{\{x\in(0,1)\}}$, where $c_1 > 0$ depends only on $\beta^{(k)}_*$ and is chosen such that $h \in \mathcal{H}_1^{\beta^{(k)}_*}(\mathbb{R},1)$.  Further, without loss of generality, suppose that $\tilde{\mathcal{J}}^{(k)} = \{1,\ldots,|\tilde{\mathcal{J}}^{(k)}|\}$. For $\bm u = (u_1,\ldots,u_{t_*^{(k)}})^\top \in \mathcal{U}$, define $\psi_{\bm u} : [0,1]^{t_*^{(k)}} \to [0,1]$ by
    \begin{align*}
        \psi_{\bm u}(x_1,\ldots,x_{t_*^{(k)}}) \coloneqq \rho^{\beta_*^{(k)}}\prod_{j=1}^{t_*^{(k)}} h\biggl(\frac{x_j-u_j}{\rho}\biggr).
    \end{align*}
    For $\bm\alpha \in \mathbb{N}_0^d$ with $\|\bm\alpha\|_1 \leq \lceil \beta_*^{(k)} \rceil - 1$, we have $\|\partial^{\bm\alpha} \psi_{\bm u}\|_{\infty} \leq 1$ since $h \in \mathcal{H}_1^{\beta^{(k)}_*}(\mathbb{R},1)$.  Moreover, for $\bm\alpha \in \mathbb{N}_0^d$ with $\|\bm\alpha\|_1 = \lceil \beta_*^{(k)} \rceil - 1$ and for all $\bm x,\bm y \in [0,1]^{t_*^{(k)}}$, we have
    \begin{align*}
        &\frac{|\partial^{\bm\alpha} \psi_{\bm u}(\bm x) - \partial^{\bm\alpha} \psi_{\bm u}(\bm y)|}{\|\bm x - \bm y\|_2^{\beta_*^{(k)}+1-\lceil\beta_*^{(k)}\rceil}} = \rho^{\beta_*^{(k)}+1-\lceil\beta_*^{(k)}\rceil} \frac{\Bigl| \prod_{j=1}^{t_*^{(k)}} h^{(\alpha_j)}\bigl(\frac{x_j-u_j}{\rho}\bigr) - \prod_{j=1}^{t_*^{(k)}} h^{(\alpha_j)}\bigl(\frac{y_j-u_j}{\rho}\bigr) \Bigr|}{\|\bm x - \bm y\|_2^{\beta_*^{(k)}+1-\lceil\beta_*^{(k)}\rceil}}\\
        &\hspace{3cm} \leq \frac{\rho^{\beta_*^{(k)}+1-\lceil\beta_*^{(k)}\rceil}}{\|\bm x - \bm y\|_2^{\beta_*^{(k)}+1-\lceil\beta_*^{(k)}\rceil}} \sum_{\ell=1}^{t_*^{(k)}}\Biggl| \prod_{j=\ell}^{t_*^{(k)}} h^{(\alpha_j)}\Bigl(\frac{x_j-u_j}{\rho}\Bigr)\prod_{j=1}^{\ell-1} h^{(\alpha_j)}\Bigl(\frac{y_j-u_j}{\rho}\Bigr)\\
        &\hspace{8cm} - \prod_{j=\ell+1}^{t_*^{(k)}} h^{(\alpha_j)}\Bigl(\frac{x_j-u_j}{\rho}\Bigr)\prod_{j=1}^{\ell} h^{(\alpha_j)}\Bigl(\frac{y_j-u_j}{\rho}\Bigr) \Biggr|\\
        &\hspace{3cm} \leq \sum_{\ell=1}^{t_*^{(k)}} \frac{\bigl| h^{(\alpha_\ell)}\bigl(\frac{x_\ell-u_\ell}{\rho}\bigr) - h^{(\alpha_\ell)}\bigl(\frac{y_\ell-u_\ell}{\rho}\bigr) \bigr|}{|(x_\ell - y_\ell)/\rho|^{\beta_*^{(k)}+1-\lceil\beta_*^{(k)}\rceil}} \leq t_*^{(k)},
    \end{align*}
    where the first inequality follows from the triangle inequality, and the second and third inequalities follow since $h \in \mathcal{H}_1^{\beta^{(k)}_*}(\mathbb{R},1)$. Thus, we have shown that $\psi_{\bm u} \in \mathcal{H}_{t_*^{(k)}}^{\beta_*^{(k)}}\bigl([0,1]^{t_*^{(k)}}, t_*^{(k)}\bigr)$. For $\bm v = (v_{\bm u})_{\bm u \in \mathcal{U}} \in \{0,1\}^{|\mathcal{U}|}$, define $\phi_{\bm v} : [0,1]^{t_*^{(k)}} \to [0,1]$ by
    \begin{align*}
        \phi_{\bm v} \coloneqq \sum_{\bm u \in \mathcal{U}} v_{\bm u} \psi_{\bm u}.
    \end{align*}
    Since $\psi_{\bm u}$ and $\psi_{\bm u'}$ have disjoint support for $\bm u \neq \bm u'$, we have $\phi_{\bm v} \in \mathcal{H}_{t_*^{(k)}}^{\beta_*^{(k)}}\bigl([0,1]^{t_*^{(k)}}, t_*^{(k)}\bigr)$.
    For $r< r_*^{(k)}$, if $d_{r+1}^{(k)} \leq d_{r}^{(k)}$, then we define $\bm g_r^{(k)} : [0,1]^{d_{r}^{(k)}} \to [0,1]^{d_{r+1}^{(k)}}$ by $\bm g_r^{(k)}(\bm x) \coloneqq (x_1,\ldots,x_{d_{r+1}^{(k)}})^\top$; otherwise, we define $\bm g_r^{(k)} : [0,1]^{d_{r}^{(k)}} \to [0,1]^{d_{r+1}^{(k)}}$ by $\bm g_r^{(k)}(\bm x) \coloneqq (\bm x, 0,\ldots,0)^\top$.  For $r > r_*^{(k)}$, define $\bm g_r^{(k)} : [0,1]^{d_{r}^{(k)}} \to [0,1]^{d_{r+1}^{(k)}}$ by $\bm g_r^{(k)}(\bm x) \coloneqq (x_1^{\beta_r^{(k)} \wedge 1}, 0,\ldots,0)^\top$. For $\bm v \in \{0,1\}^{|\mathcal{U}|}$, define $\bm g_{r_*^{(k)},\bm v}^{(k)}: [0,1]^{d_{r_*^{(k)}}^{(k)}} \to [0,1]^{d_{r_*^{(k)}+1}^{(k)}}$ by $\bm g_{r_*^{(k)},\bm v}^{(k)}(\bm x) \coloneqq \bigl(\phi_{\bm v}(x_1,\ldots,x_{t_*^{(k)}}), 0, \ldots, 0\bigr)^\top$.
    Then each component of $\bm g_r^{(k)}$ belongs to $\mathcal{H}_{t^{(k)}_r}^{\beta^{(k)}_r} ([0,1]^{d_{r}^{(k)}}, \gamma^{(k)}_r)$ for $r\in[q_k] \setminus \{r_*^{(k)}\}$, and each component of $\bm g_{r_*^{(k)},\bm{v}}^{(k)}$ belongs to $\mathcal{H}_{t^{(k)}_{*}}^{\beta^{(k)}_*} ([0,1]^{d_{r_*^{(k)}}^{(k)}}, \gamma^{(k)}_*)$ for $\bm{v} \in \{0,1\}^{|\mathcal{U}|}$. Let $B \coloneqq \prod_{r=r_*^{(k)}+1}^{q_k} (\beta_r^{(k)} \wedge 1) = \bar{\beta}_*^{(k)} / \beta_*^{(k)}$. For $\bm v = (v_{\bm u})_{\bm u \in \mathcal{U}} \in \{0,1\}^{|\mathcal{U}|}$, define $f^{(k)}_{\bm v} : [0,1]^d \to [0,1/2]$ by
    \begin{align*}
        f^{(k)}_{\bm v}(\bm x) &\coloneqq \frac{1}{2}\bm g_{q_k}^{(k)} \circ \cdots \circ \bm g_{r_*^{(k)}+1}^{(k)} \circ \bm g_{r_*^{(k)},\bm v}^{(k)} \circ \bm g_{r_*^{(k)}-1}^{(k)} \circ \cdots \circ \bm g_{1}^{(k)}(\bm x)\\
        &\phantom{:}= \frac{1}{2}\phi_{\bm v}(\bm x)^B = \frac{1}{2}\sum_{\bm u \in \mathcal{U}} v_{\bm u} \psi_{\bm u}(x_1,\ldots,x_{t_*^{(k)}})^B,
    \end{align*}
    which satisfies $f^{(k)}_{\bm v} \in \mathcal{F}^{(k)}$ by construction.
    For $\bm u \in \mathcal{U}$, we have $\|\psi_{\bm u}^B\|_{L_2}^2 = \rho^{2\bar{\beta}_*^{(k)} + t_*^{(k)}} \|h^B\|_{L_2}^{2t_*^{(k)}}$. Therefore, for $\bm v, \bm v' \in \{0,1\}^{|\mathcal{U}|}$,
    \begin{align}
        \|f^{(k)}_{\bm v} - f^{(k)}_{\bm v'}\|_{L_2}^2 = \frac{1}{4}\|\bm v - \bm v'\|_1 \cdot \rho^{2\bar{\beta}_*^{(k)} + t_*^{(k)}} \|h^B\|_{L_2}^{2t_*^{(k)}}. \label{eq:f_w-difference-upper-bound}
    \end{align}
    By the Gilbert--Varshamov lemma \citet[][Exercise~8.9]{samworth2025statistics}, there exists $\mathcal{V} \subseteq \{0,1\}^{|\mathcal{U}|}$ such that $|\mathcal{V}| \geq e^{|\mathcal{U}|/8}$ and $\|\bm v - \bm v'\|_1 > |\mathcal{U}|/4$ for all $\bm v, \bm v' \in \mathcal{V}$ with $\bm v \neq \bm v'$. Thus, for $\bm v, \bm v' \in \mathcal{V}$ with $\bm v \neq \bm v'$, we have by~\eqref{eq:mu-Z-Omega-norm-bound} and~\eqref{eq:f_w-difference-upper-bound} that
    \begin{align}
        \int_{[0,1]^d \times \mathcal{S}}\bigl\{f^{(k)}_{\bm v}(\bm z) &- f^{(k)}_{\bm v'}(\bm z) \bigr\}^2 \mathbbm{1}_{\{\bm\omega\in\mathcal{S}_k\}} \, \mathrm{d}\mu_{\bm Z_0,\bm\Omega_0}(\bm z, \bm \omega) = \pi_k \|f^{(k)}_{\bm v} - f^{(k)}_{\bm v'}\|_{L_2}^2\nonumber\\
        &> \frac{\pi_k}{4}  \cdot \frac{\rho^{-t_*^{(k)}}}{4} \cdot \rho^{2\bar{\beta}_*^{(k)} + t_*^{(k)}} \|h^B\|_{L_2}^{2t_*^{(k)}} \geq \frac{\|h^B\|_{L_2}^{2t_*^{(k)}}}{16 \lambda^{2\bar{\beta}_*^{(k)}}} \cdot \pi_k n_k^{-2\bar{\beta}_*^{(k)} / (2\bar{\beta}_*^{(k)} + t_*^{(k)})}. \label{eq:separation-lower-bound}
    \end{align}
    Moreover, with $\lambda \coloneqq 2\Bigl(\frac{100\|h^B\|_{L_2}^{2t_*^{(k)}}}{\xi_2^2}\Bigr)^{1/(2\bar{\beta}_*^{(k)} + t_*^{(k)})} \vee 16$,
    \begin{align}
        \mathrm{KL}\Bigl(P^{\otimes n}_{f^{(k)}_{\bm v}}, P^{\otimes n}_{f^{(k)}_{\bm v'}}\Bigr) &= n\mathrm{KL}\Bigl(P_{f^{(k)}_{\bm v}}, P_{f^{(k)}_{\bm v'}}\Bigr) \nonumber \\
        &= 
        \frac{25n}{2\xi_2^2} \cdot \int_{[0,1]^d \times \mathcal{S}}\bigl\{f^{(k)}_{\bm v}(\bm z) - f^{(k)}_{\bm v'}(\bm z) \bigr\}^2 \mathbbm{1}_{\{\bm\omega\in\mathcal{S}_k\}} \, \mathrm{d}\mu_{\bm Z_0,\bm\Omega_0}(\bm z, \bm \omega) \nonumber\\
        &= \frac{25n\pi_k}{2\xi_2^2} \|f^{(k)}_{\bm v} - f^{(k)}_{\bm v'}\|_{L_2}^2 \leq \frac{25\|h^B\|_{L_2}^{2t_*^{(k)}}}{8\xi_2^2} \cdot n_k \rho^{2\bar{\beta}_*^{(k)}}  \nonumber\\
        &\leq \frac{1}{32\rho^{t_*^{(k)}}}  \leq \frac{\log(|\mathcal{V}|)}{4}. \label{eq:kl-upper-bound}
    \end{align}
    By applying Fano's lemma \citet[Corollary~8.12]{samworth2025statistics} in conjunction with~\eqref{eq:separation-lower-bound} and~\eqref{eq:kl-upper-bound}, we conclude that there exists $c^{(k)} > 0$, depending only on $(\xi_2, \bar{\beta}_*^{(k)}, t_*^{(k)})$, such that
    \begin{align}
        &\inf_{\hat{f} \in \hat{\mathcal{F}}} \sup_{f^{(k)} \in \mathcal{F}^{(k)}}  \mathbb{E}_{P_{f^{(k)}}^{\otimes n}}\int_{[0,1]^d \times \mathcal{S}} \bigl\{\hat{f}(\bm z,\bm\omega) - f^{(k)}(\bm z)\bigr\}^2 \mathbbm{1}_{\{\bm\omega\in\mathcal{S}_k\}} \, \mathrm{d}\mu_{\bm Z_0,\bm\Omega_0}(\bm z, \bm \omega) \nonumber\\
        &\hspace{9cm}\geq c^{(k)} \pi_k n_k^{-2\bar{\beta}_*^{(k)} / (2\bar{\beta}_*^{(k)} + t_*^{(k)})}. \label{eq:minimax-lb-kth-term}
    \end{align}
    The final result follows from~\eqref{eq:minimax-lower-bound-k-simplified} and~\eqref{eq:minimax-lb-kth-term}.
\end{proof}

\section{Auxiliary lemmas}
\begin{lemma}\label{lemma:LambertW}
    For $c\geq 4$ and $d\geq 1$, if $m \geq 0$ satisfies $2^m \leq (cm)^d$, then $m\leq 2d\log_2(cd)$.
\end{lemma}
\begin{proof}
    First, in the case $d=1$, we have that $x \geq 0$ satisfies $2^x\leq cx$ if and only if $f(x)\coloneqq x-\log_2 x - \log_2 c \leq 0$. Since $f'(x) \geq 0$ for $x \geq \log_2 e$, any $m_0\geq \log_2 e$ satisfying $f(m_0) > 0$ is an upper bound of $m$.  But for $m_0 = 2\log_2 c$, 
    \[  
    f(m_0) = \log_2 c - \log_2\log_2 c - 1 \geq 0
    \]
    for all $c \geq 4$. Hence $m\leq m_0$ as desired. For general $d\geq 1$, defining $y \coloneqq m/d$, we have $2^y\leq cd y$, so the result follows from the case $d=1$.
\end{proof}

\begin{lemma} \label{lemma:truncation-bias}
    Let $Y$ be a random variable such that $\|Y\|_{\psi_2} \leq \xi$ for some $\xi > 0$.  Let $Z$ be a random variable taking values in a measurable space $\mathcal{Z}$.  Writing $B_n\coloneqq \xi\sqrt{2\log n}$, we have
    \begin{align*}
        \mathbb{E}\{|Y - T_{B_n}Y|\} \leq \frac{\sqrt{\pi}\xi}{n^2}, \quad \mathbb{E}\bigl\{Y^2 - (T_{B_n}Y)^2\bigr\} \leq \frac{2\xi^2}{n^2}
    \end{align*}
    and
    \begin{align*}
        \mathbb{E}\Bigl\{\bigl(\mathbb{E}(Y\,|\,Z)\bigr)^2 - \bigl(\mathbb{E}(T_{B_n}Y\,|\,Z)\bigr)^2 \Bigr\} \leq \frac{4\xi^2}{n}, \, \Bigl|\mathbb{E}\Bigl\{\mathbb{E}(Y\,|\,Z)Y - \mathbb{E}(T_{B_n}Y\,|\,Z) \cdot T_{B_n}Y\Bigr\}\Bigr| \leq \frac{4\xi^2}{n}.
    \end{align*}
    % \begin{align*}
    %     \mathbb{E}\bigl(XY - T_{B_n}X \cdot T_{B_n}Y\bigr) \leq \frac{4\xi^2}{n}.
    % \end{align*}
\end{lemma}
\begin{proof}
    By Markov's inequality,
    \begin{align*}
        \mathbb{P}(|Y|\geq t) \leq \mathbb{E}(e^{|Y|^2/\xi^2}) \cdot e^{-t^2/\xi^2} \leq 2e^{-t^2/\xi^2},
    \end{align*}
    for all $t\geq 0$. For the first inequality,
    \begin{align*}
        \mathbb{E}\{|Y - T_{B_n}Y|\} &= \int_0^{\infty} \mathbb{P}(|Y - T_{B_n}Y| \geq t) \,\mathrm{d}t\\
        &\leq \int_0^{\infty} \mathbb{P}(|Y| \geq B_n + t) \,\mathrm{d}t\\
        &\leq 2\int_0^{\infty} e^{-(B_n^2 + t^2)/\xi^2}\,\mathrm{d}t = \frac{2}{n^2} \int_0^{\infty} e^{-t^2/\xi^2}\,\mathrm{d}t = \frac{\sqrt{\pi} \xi}{n^2}.
    \end{align*}
    For the second inequality,
    \begin{align*}
        \mathbb{E}\bigl\{Y^2 - (T_{B_n}Y)^2\bigr\} &= \int_0^{\infty} \mathbb{P}\bigl(Y^2 - (T_{B_n}Y)^2 \geq t\bigr)\,\mathrm{d}t\\
        &\leq \int_0^{\infty} \mathbb{P}\bigl(Y^2 \geq B_n^2 + t\bigr) \,\mathrm{d}t \leq 2 \int_0^{\infty} e^{-(B_n^2 + t)/\xi^2} \,\mathrm{d}t = \frac{2\xi^2}{n^2}.
    \end{align*}
    For the third inequality, by the Cauchy--Schwarz inequality,
    \begin{align}
        \mathbb{E}\Bigl\{\bigl(\mathbb{E}&(Y\,|\,Z)\bigr)^2 - \bigl(\mathbb{E}(T_{B_n}Y\,|\,Z)\bigr)^2 \Bigr\}\nonumber\\
        &= \mathbb{E}\Bigl\{\bigl(\mathbb{E}(Y\,|\,Z) -\mathbb{E}(T_{B_n}Y\,|\,Z)\bigr) \cdot \bigl(\mathbb{E}(Y\,|\,Z) +\mathbb{E}(T_{B_n}Y\,|\,Z)\bigr) \Bigr\} \nonumber\\
        &\leq \sqrt{\mathbb{E}\Bigl\{\bigl(\mathbb{E}(Y-T_{B_n}Y\,|\,Z)\bigr)^2\Bigr\} \cdot \mathbb{E}\Bigl\{2\bigl(\mathbb{E}(Y\,|\,Z)\bigr)^2 + 2\bigl(\mathbb{E}(T_{B_n}Y\,|\,Z)\bigr)^2 \Bigr\}}. \label{eq:truncation-ineq-5}
    \end{align}
    Moreover, by the conditional version of Jensen's inequality,
    \begin{align}
        \mathbb{E}\Bigl\{\bigl(\mathbb{E}(Y-T_{B_n}Y\,|\,Z)\bigr)^2\Bigr\} \leq \mathbb{E}\Bigl\{\mathbb{E}\bigl((Y-T_{B_n}Y)^2\,|\,Z\bigr)\Bigr\} \leq \mathbb{E}\bigl\{Y^2 - (T_{B_n}Y)^2\bigr\} \leq \frac{2\xi^2}{n^2}, \label{eq:truncation-ineq-6}
    \end{align}
    and
    \begin{align}
        \mathbb{E}\Bigl\{2\bigl(\mathbb{E}(Y\,|\,Z)\bigr)^2 + 2\bigl(\mathbb{E}(&T_{B_n}Y\,|\,Z)\bigr)^2 \Bigr\} \leq \mathbb{E}\Bigl\{2\mathbb{E}(Y^2\,|\,Z) + 2\mathbb{E}\bigl((T_{B_n}Y)^2\,|\,Z\bigr) \Bigr\}\nonumber\\
        &\leq 4\mathbb{E}(Y^2) = 4 \int_0^{\infty} \mathbb{P}(Y^2\geq t)\,\mathrm{d}t \leq 8 \int_0^{\infty} e^{-t/\xi^2}\,\mathrm{d}t = 8\xi^2. \label{eq:truncation-ineq-7}
    \end{align}
    The third inequality then follows by combining \eqref{eq:truncation-ineq-5}, \eqref{eq:truncation-ineq-6} and \eqref{eq:truncation-ineq-7}.
    
    Finally, for the fourth inequality, by Cauchy--Schwarz again,
    \begin{align}
        \Bigl|\mathbb{E}\Bigl\{\mathbb{E}&(Y\,|\,Z)Y - \mathbb{E}(T_{B_n}Y\,|\,Z) \cdot T_{B_n}Y\Bigr\}\Bigr|\nonumber\\
        &\leq \Bigl|\mathbb{E}\Bigl\{\mathbb{E}(Y\,|\,Z)(Y-T_{B_n}Y)\Bigr\}\Bigr| + \Bigl|\mathbb{E}\Bigl\{\bigl( \mathbb{E}(Y\,|\,Z)-\mathbb{E}(T_{B_n}Y\,|\,Z) \bigr) \cdot T_{B_n}Y\Bigr\}\Bigr| \nonumber\\
        &\leq \sqrt{\mathbb{E}\bigl\{\bigl(\mathbb{E}(Y\,|\,Z)\bigr)^2\bigr\}\cdot \mathbb{E}\bigl\{(Y-T_{B_n}Y)^2\bigr\}} + B_n \mathbb{E}\bigl\{ \bigl| \mathbb{E}(Y - T_{B_n}Y\,|\,Z) \bigr| \bigr\}. \label{eq:truncation-ineq-1}
        %&\leq \sqrt{\mathbb{E}(X^2)\mathbb{E}\{(Y-T_{B_n}Y)^2\}} + \xi\sqrt{2\log n} \cdot \frac{\sqrt{\pi}\xi}{n^2}. \label{eq:truncation-ineq-1}
    \end{align}
    Now, by~\eqref{eq:truncation-ineq-7},
    \begin{align}
        \mathbb{E}\bigl\{\bigl(\mathbb{E}(Y\,|\,Z)\bigr)^2\bigr\} \leq \mathbb{E}(Y^2) \leq 2\xi^2, \label{eq:truncation-ineq-2}
    \end{align} 
    and
    \begin{align}
        \mathbb{E}\bigl\{(Y-T_{B_n}Y)^2\bigr\} \leq \mathbb{E}\bigl\{Y^2-(T_{B_n}Y)^2\bigr\} \leq \frac{2\xi^2}{n^2}. \label{eq:truncation-ineq-3}
    \end{align}
    Lastly,
    \begin{align}
        \mathbb{E}\bigl\{ \bigl| \mathbb{E}(Y - T_{B_n}Y\,|\,Z) \bigr| \bigr\} \leq \mathbb{E}\bigl(|Y - T_{B_n}Y|\bigr) \leq \frac{\sqrt{\pi}\xi}{n^2}. \label{eq:truncation-ineq-4}
    \end{align}
    Combining~\eqref{eq:truncation-ineq-1}, \eqref{eq:truncation-ineq-2}, \eqref{eq:truncation-ineq-3} and \eqref{eq:truncation-ineq-4} yields that
    \begin{align*}
        \Bigl|\mathbb{E}\Bigl\{\mathbb{E}(Y\,|\,Z)Y - \mathbb{E}(T_{B_n}Y\,|\,Z) \cdot T_{B_n}Y\Bigr\}\Bigr| \leq \frac{4\xi^2}{n},
    \end{align*}
    as desired.
\end{proof}

\begin{lemma} \label{lemma:taylor-approximation-error}
    Let $\beta,\gamma>0$, $\beta_0 \coloneqq \lceil\beta\rceil - 1$, $d\in\mathbb{N}$ and $g\in \mathcal{H}_d^{\beta}\bigl([0,1]^d,\gamma\bigr)$. Then
    \begin{align*}
        \biggl| g(\bm y) - \sum_{\bm\alpha \in \mathbb{N}_0^d : \|\bm\alpha\|_1 \leq \beta_0} \frac{\partial^{\bm\alpha}g(\bm x)}{\bm\alpha!}(\bm y - \bm x)^{\bm\alpha} \biggr| \leq \gamma d^{\beta_0} \|\bm y - \bm x\|_2^{\beta}
    \end{align*}
    for all $\bm x,\bm y \in [0,1]^d$.
\end{lemma}
\begin{proof}
    Let $\bm u \coloneqq \bm y - \bm x$. By Taylor's theorem, there exists $\theta\in(0,1)$ such that
    \begin{align*}
        g(\bm y) = \sum_{\bm\alpha \in \mathbb{N}_0^d : \|\bm\alpha\|_1 \leq \beta_0-1} \frac{\partial^{\bm\alpha}g(\bm x)}{\bm\alpha!}\bm u^{\bm\alpha} + \sum_{\bm\alpha \in \mathbb{N}_0^d : \|\bm\alpha\|_1 = \beta_0} \frac{\partial^{\bm\alpha}g(\bm x + \theta\bm u)}{\bm\alpha!}\bm u^{\bm\alpha}.
    \end{align*}
    Thus,
    \begin{align*}
        \biggl| g(\bm y) - \sum_{\bm\alpha \in \mathbb{N}_0^d : \|\bm\alpha\|_1 \leq \beta_0} \frac{\partial^{\bm\alpha}g(\bm x)}{\bm\alpha!} \bm u^{\bm\alpha} \biggr| &\leq \sum_{\bm\alpha \in \mathbb{N}_0^d : \|\bm\alpha\|_1 = \beta_0} \bigl|\partial^{\bm\alpha}g(\bm x + \theta\bm u) - \partial^{\bm\alpha}g(\bm x)\bigr| \bm u^{\bm\alpha}\\
        &\leq d^{\beta_0} \gamma \|\bm u\|_2^{\beta-\beta_0}\cdot \|\bm u\|_2^{\beta_0}
    \end{align*}
    where the final inequality uses the facts that $\bigl|\{\bm\alpha \in \mathbb{N}_0^d : \|\bm\alpha\|_1 = \beta_0\}\bigr|= \binom{d+\beta_0-1}{\beta_0}\leq d^{\beta_0}$, that $g \in \mathcal{H}_d^{\beta}\bigl([0,1]^d,\gamma\bigr)$ and that $\bm u^{\bm\alpha} = \prod_{j=1}^d u_j^{\alpha_j} \leq \|\bm u\|_{\infty}^{\beta_0} \leq \|\bm u\|_2^{\beta_0}$.
\end{proof}

\begin{lemma}\label{lemma:arbitrary-partition}
    Given an arbitrary partition $\{\mathcal{S}_1,\ldots,\mathcal{S}_K\}$ of $\mathcal{S}$, there exists a Bayes regression function~$f^{\star}$ that satisfies Assumption~\ref{assumption:piecewise-assumption}.
\end{lemma}
\begin{proof}
    Indeed, let $\bm{X} = (X_1,\ldots,X_d)^\top \sim \mathrm{Unif}[0,1]^d$ and let $Y = \sum_{k=1}^K \mathbbm{1}_{\{X_{1}\geq (k-1)/K\}} + \varepsilon$ where $\varepsilon \sim N(0,1)$ is independent of $\bm X$. Further define $\bm \Omega$ by $\bm \Omega \,|\, \bm X \sim \mathrm{Unif}( \mathcal{S}_k)$ when $X_1 \in \bigl[(k-1)/K,k/K\bigr)$. Then $f^{\star}(\bm z,\bm \omega) = k$ for all $\bm z \in [0,1]^d$ whenever $\bm \omega\in\mathcal{S}_k$.
\end{proof}

\section{Visualisation of the learned pattern embeddings for synthetic data} \label{sec:embeddings-syn}

In Figures~\ref{fig:embedding-model1}--\ref{fig:embedding-model4}, we plot visualisations of the learned embeddings for our synthetic data examples from Models~1--4 in Section~\ref{sec:simulated-data}.  Since the original embeddings were in four dimensions, we plot the first two principal components for the visualisations.  The clustering structure of the different colours in these figures reinforces the message that the embeddings of the revelation vectors contain relevant information for predictive performance.  In some cases, such as Figure~\ref{fig:embedding-model2}, the clusters may be almost disjoint, while in others, such as Figure~\ref{fig:embedding-model1}, there may be non-trivial overlap between the clusters, at least following the projection onto the first two principal components. 

\begin{figure}[htbp]
\centering

\begin{subfigure}[t]{\textwidth}
    \includegraphics[width=\textwidth]{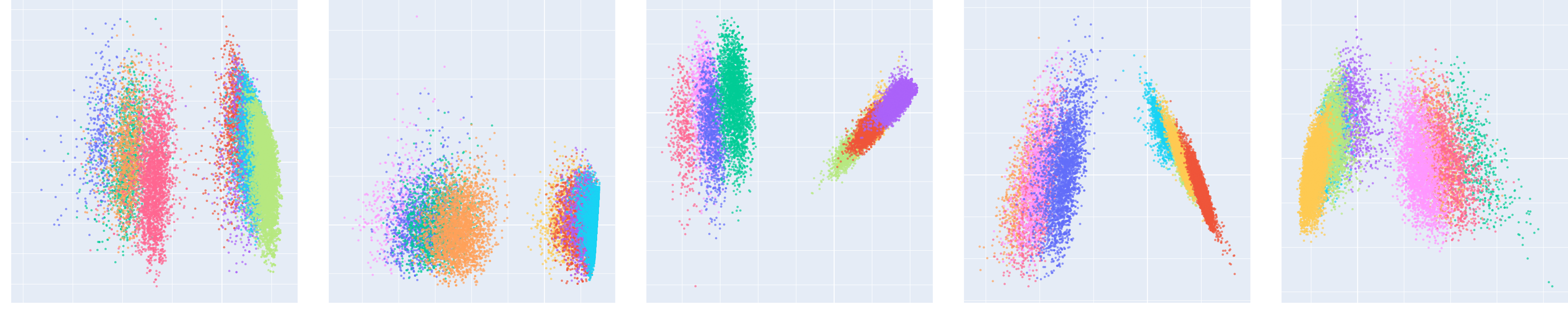}
    \caption{Model 1 MI}
\end{subfigure}

\begin{subfigure}[t]{\textwidth}
    \includegraphics[width=\textwidth]{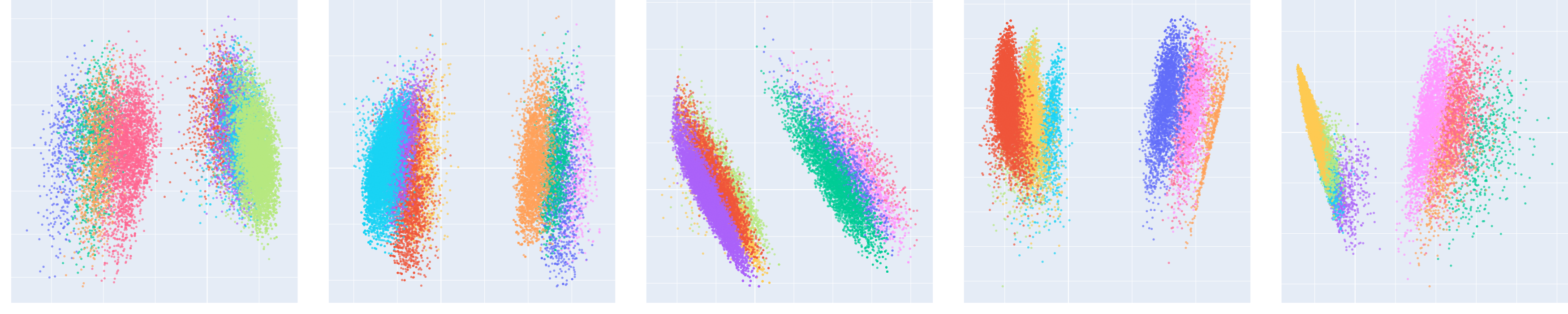}
    \caption{Model 1 MF}
\end{subfigure}

\begin{subfigure}[t]{\textwidth}
    \includegraphics[width=\textwidth]{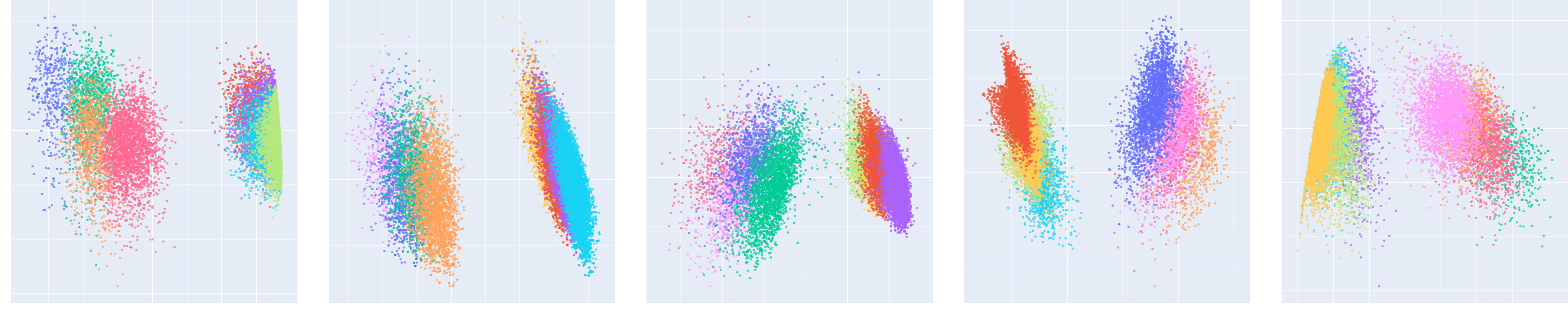}
    \caption{Model 1 II}
\end{subfigure}

\caption{Visualisation of the learned embeddings for five different realisations of synthetic data from Model 1, and three different imputation methods: Mean Imputation (MI, top), MissForest (MF, middle) and Iterative Imputer (II, bottom).  There are eight different cells (depending on $\omega_1,\omega_2,\omega_3$) in the Bayes regression function partition, and these are plotted with different colours.} 
\label{fig:embedding-model1}
\end{figure}

\begin{figure}[htbp]
\centering

\begin{subfigure}[t]{\textwidth}
    \includegraphics[width=\textwidth]{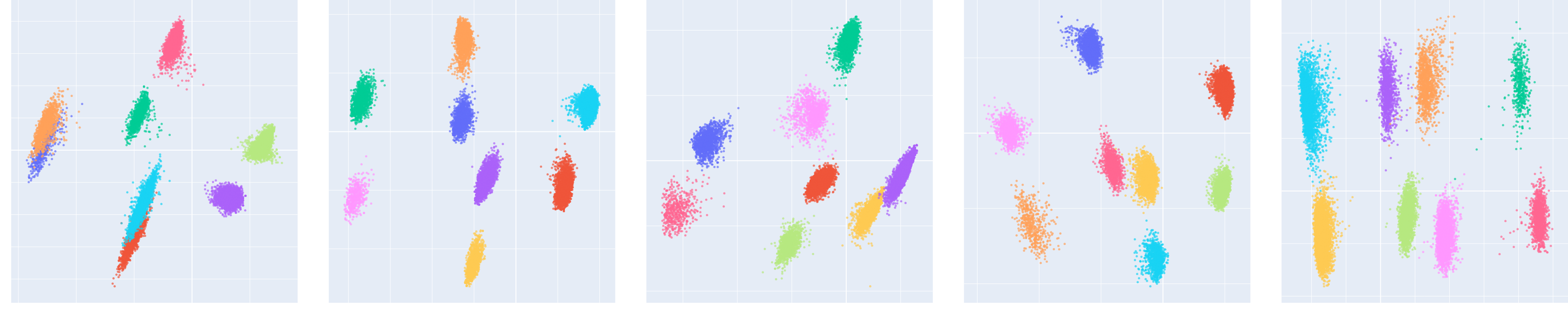}
    \caption{Model 2 MI}
\end{subfigure}

\begin{subfigure}[t]{\textwidth}
    \includegraphics[width=\textwidth]{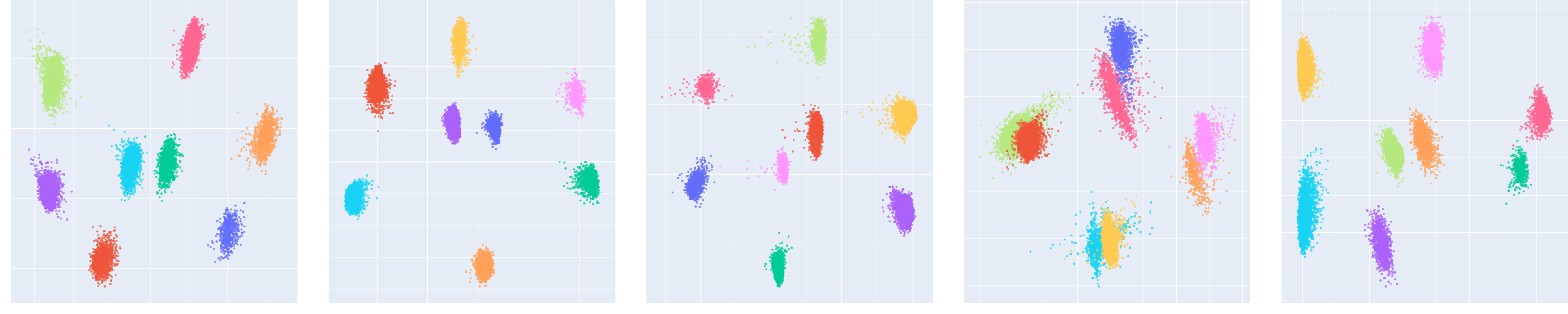}
    \caption{Model 2 MF}
\end{subfigure}

\begin{subfigure}[t]{\textwidth}
    \includegraphics[width=\textwidth]{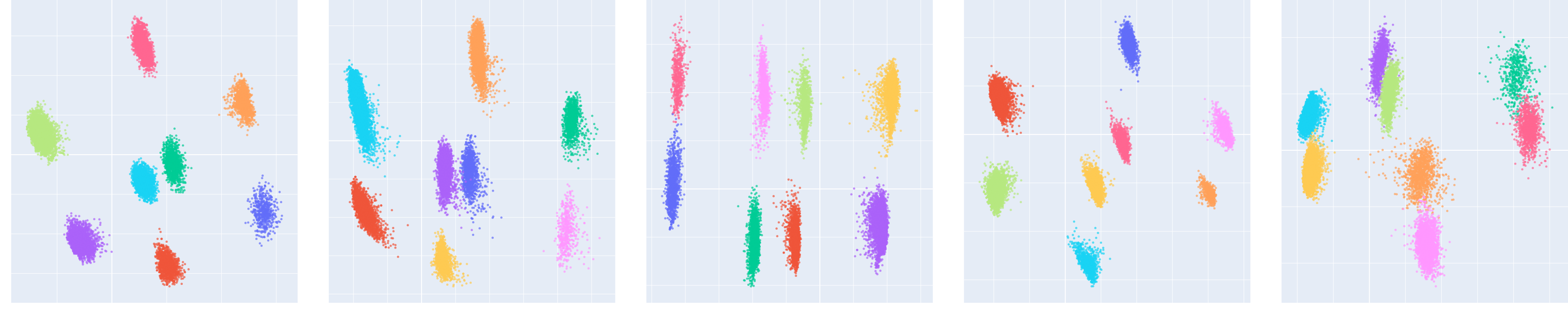}
    \caption{Model 2 II}
\end{subfigure}

\caption{Visualisation of the learned embeddings for five different realisations of synthetic data from Model 2, and three different imputation methods: Mean Imputation (MI, top), MissForest (MF, middle) and Iterative Imputer (II, bottom). There are eight different cells (depending on $\omega_1,\omega_2,\omega_3$) in the Bayes regression function partition, and these are plotted with different colours.} 
\label{fig:embedding-model2}
\end{figure}

\begin{figure}[htbp]
\centering

\begin{subfigure}[t]{\textwidth}
    \includegraphics[width=\textwidth]{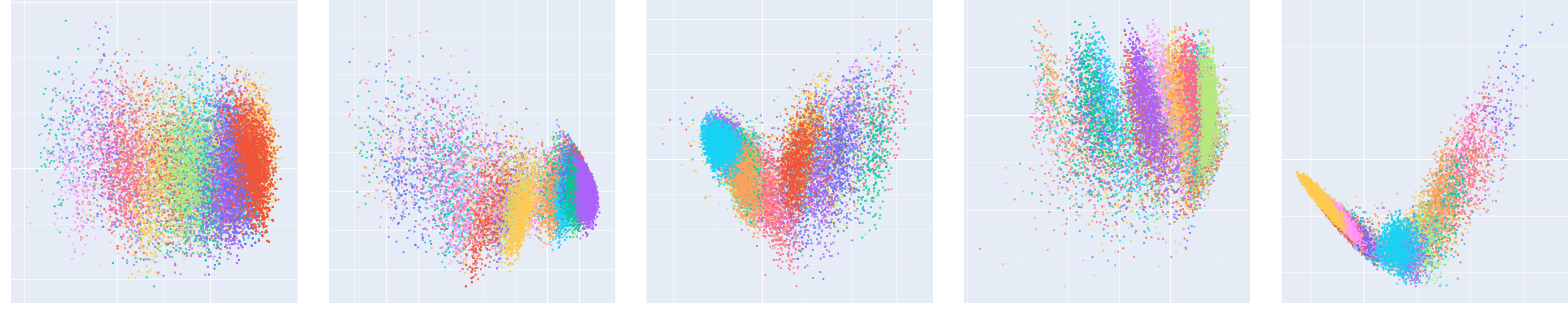}
    \caption{Model 3 MI}
\end{subfigure}

\begin{subfigure}[t]{\textwidth}
    \includegraphics[width=\textwidth]{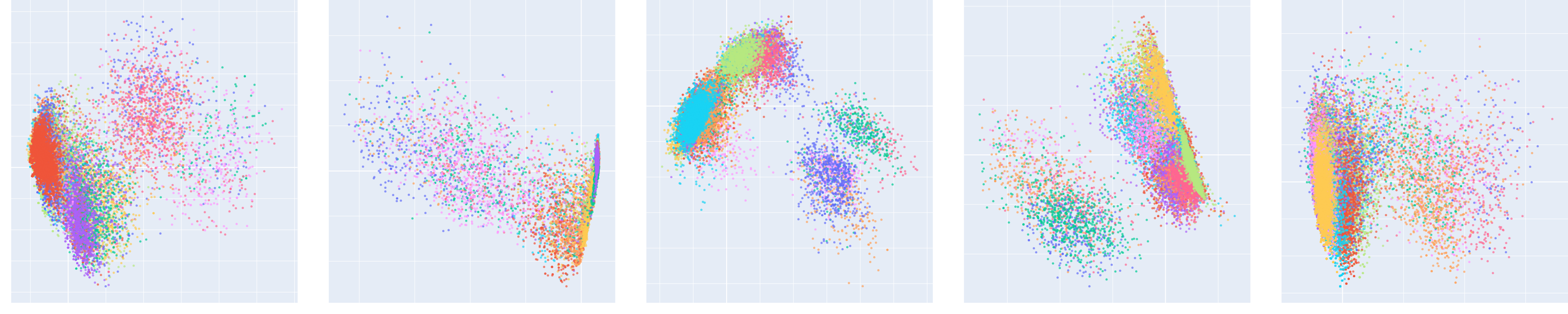}
    \caption{Model 3 MF}
\end{subfigure}

\begin{subfigure}[t]{\textwidth}
    \includegraphics[width=\textwidth]{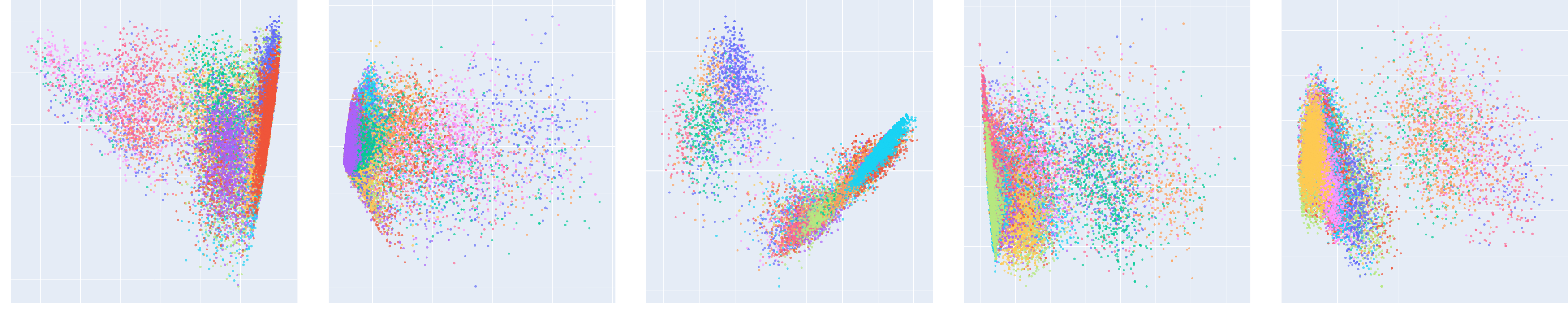}
    \caption{Model 3 II}
\end{subfigure}

\caption{Visualisation of the learned embeddings for five different realisations of synthetic data from Model 3, and three different imputation methods: Mean Imputation (MI, top), MissForest (MF, middle) and Iterative Imputer (II, bottom). There are 32 different cells (depending on $\omega_1,\ldots,\omega_5$) in the Bayes regression function partition, and these are plotted with different colours.} 
\label{fig:embedding-model3}
\end{figure}

\begin{figure}[htbp]
\centering

\begin{subfigure}[t]{\textwidth}
    \includegraphics[width=\textwidth]{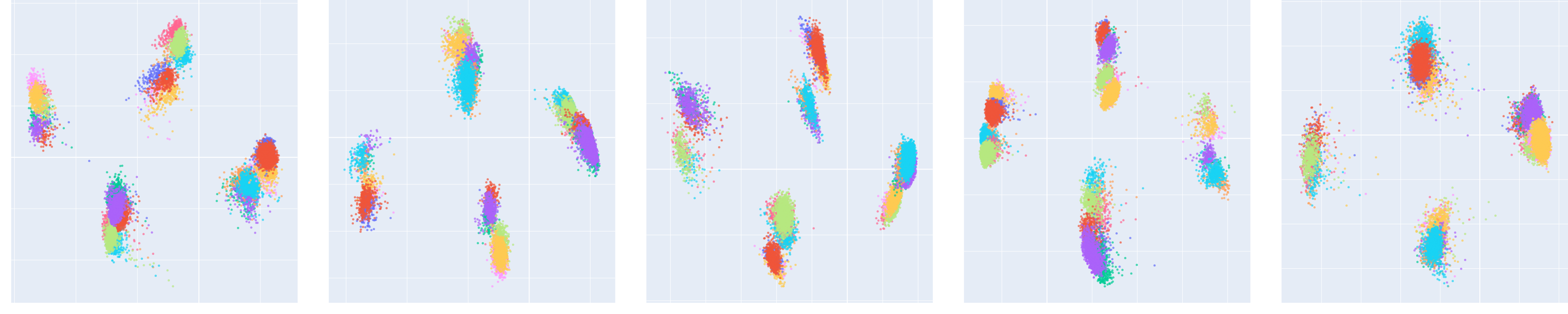}
    \caption{Model 4 MI}
\end{subfigure}

\begin{subfigure}[t]{\textwidth}
    \includegraphics[width=\textwidth]{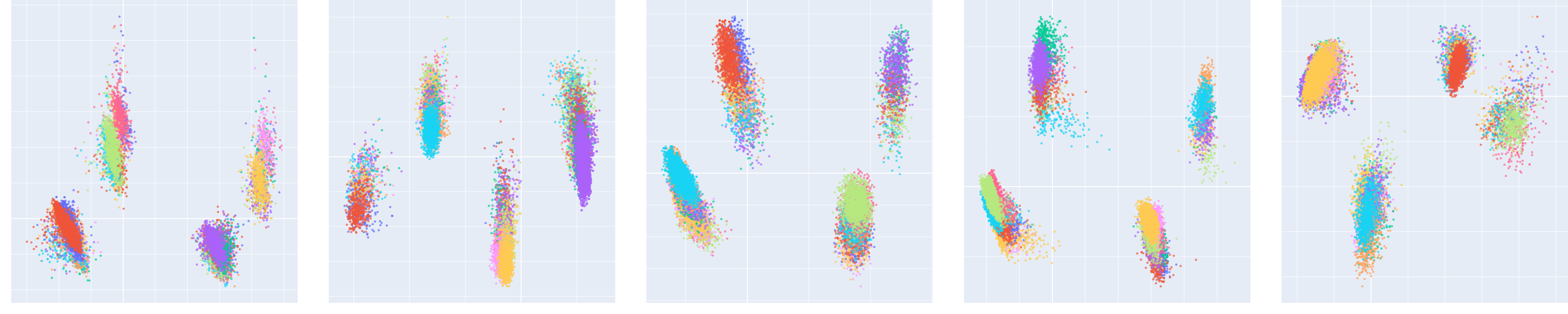}
    \caption{Model 4 MF}
\end{subfigure}

\begin{subfigure}[t]{\textwidth}
    \includegraphics[width=\textwidth]{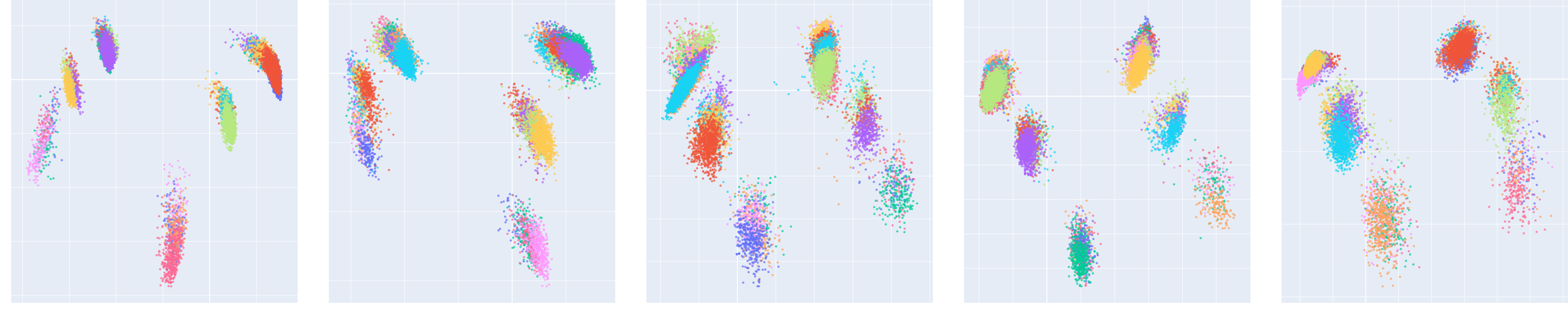}
    \caption{Model 4 II}
\end{subfigure}

\caption{Visualisation of the learned embeddings for five different realisations of synthetic data from Model 4, and three different imputation methods: Mean Imputation (MI, top), MissForest (MF, middle) and Iterative Imputer (II, bottom). There are 32 different cells (depending on $\omega_1,\ldots,\omega_5$) in the Bayes regression function partition, and these are plotted with different colours.} 
\label{fig:embedding-model4}
\end{figure}

\end{document}